\documentclass{article}
\usepackage{wrapfig}

\usepackage{hyperref}       
\usepackage{url}    

\usepackage{subcaption}
\usepackage{float}
\usepackage{graphicx}    

\usepackage{booktabs}       
\usepackage{amsfonts}       
\usepackage{nicefrac}       
\usepackage{microtype}      
\usepackage{xcolor}         
\usepackage{amsmath}
\usepackage{algorithm}
\usepackage{algorithmic}

 \usepackage{float}
 \usepackage{pdfpages}
 \usepackage{tikz}
\usepackage{diagbox}
\usepackage{xcolor}
\usepackage{graphics}
\usepackage{units}
 \usepackage{multirow, color}

 \usepackage{textcase, booktabs,makecell}
 \usepackage{thmtools,thm-restate}
 \usepackage{mathtools}
\usepackage{enumitem}
\usepackage{tikz}
\usetikzlibrary{positioning}

 \usepackage{thmtools,thm-restate}
 \usepackage{mathtools}
\usepackage{enumitem}

\usepackage{amsthm}

\newtheorem{definition}{Definition}
\newtheorem{assu}{Assumption}
\newtheorem{lemma}{Lemma}
\newtheorem{theorem}{Theorem}

\def \E {\mathbb{E}}

\def \R {\mathbb{R}}

\newtheorem{claim}{Claim}

\makeatother

\newtheorem{theoremfourtemp}{Theorem}
\newenvironment{theoremfour}[1]{%
  \renewcommand\thetheoremfourtemp{1 (#1)} 
  \theoremfourtemp%
}{%
  \endtheoremfourtemp%
}

\usepackage{microtype}

\usepackage{booktabs}

\usepackage{subcaption}
\usepackage{graphicx}
\usepackage{varwidth}

\usepackage[normalem]{ulem} 
\usepackage{times}
\usepackage{hyperref}       
\usepackage{url}      
\usepackage{booktabs}       
\usepackage{amsfonts}       
\usepackage{amssymb}
\usepackage{nicefrac}    
\usepackage[numbers,sort]{natbib}
\usepackage{microtype}      
\usepackage{xcolor}      

\usepackage[a4paper, margin=1in]{geometry}

\begin{document}

\begin{center}

{\bf{\Large{Last-iterate Convergence for Symmetric, General-sum, $2 \times 2$ Games \\ Under The Exponential Weights Dynamic}}}

\vspace*{.2in}

{\large{
\begin{tabular}{cccc}
 Guanghui Wang$^1$, Krishna Acharya$^2$, Lokranjan Lakshmikanthan$^2$\\
 Juba Ziani$^{2}$, Vidya Muthukumar$^{3,2}$   \\
\end{tabular}}}

\vspace*{.05in}

\begin{tabular}{c}
$^1$College of Computing, Georgia Institute of Technology\\
$^2$School of Industrial and Systems Engineering, Georgia Institute of Technology\\
$^3$H. Milton Stewart School of Electrical and Computer Engineering, Georgia Institute of Technology\\
  \texttt{\{gwang369,krishna.acharya,llakshmi3,jziani3,vmuthukumar8\}@gatech.edu}
\end{tabular}

\vspace*{.2in}
\date{}

\end{center}
\begingroup
\renewcommand{\thefootnote}{}
\footnotetext{A preliminary version of this paper has been accepted for publication at ALT 2026.}
\endgroup

\begin{abstract}
We conduct a comprehensive analysis of the discrete-time exponential-weights dynamic with a constant step size on all \emph{general-sum and symmetric} $2 \times 2$ normal-form games, i.e. games with $2$ pure strategies per player, and where the ensuing payoff tuple is of the form $(A,A^\top)$ (where $A$ is the $2 \times 2$ payoff matrix corresponding to the first player). Such symmetric games commonly arise in real-world interactions between ``symmetric" agents who have identically defined utility functions --- such as Bertrand competition and multi-agent performative prediction, and display a rich multiplicity of equilibria despite the seemingly simple setting.
Somewhat surprisingly, we show through a first-principles analysis that the exponential weights dynamic, which is popular in online learning, converges in the last iterate for such games regardless of initialization with an appropriately chosen step size. For certain games and/or initializations, we further show that the convergence rate is in fact exponential and holds for any step size.

We illustrate our theory with extensive simulations and applications to the aforementioned game-theoretic interactions. In the case of multi-agent performative prediction, we formulate a new ``mortgage competition" game between lenders (i.e. banks) who interact with a population of customers, and show that it fits into our framework.
\end{abstract}

\section{Introduction}\label{sec:intro}

Game theory is a classical mathematical framework in which the formal study of multi-agent non-cooperative interaction is situated~\citep{myerson2013game,karlin2017game}. 
In a ``one-shot'' normal-form game, if the players have full knowledge of the game, are selfish and rational, and no mediator is present, they are expected to play a \emph{Nash equilibrium} (NE) strategy~\citep{myerson2013game} (see also Section~\ref{sec:prelim} for a self-contained definition).
Emerging applications of game theory abound in modern machine learning and artificial intelligence, such as the training of generative adversarial networks and robust models~\citep{daskalakis2018training,madry2018towards}, multi-agent reinforcement learning~\citep{zhang2021multi}, and situations where the pipeline for data may respond \emph{strategically} to ML model deployment (commonly called \emph{strategic classification}~\citep{hardt2016strategic} or more recently \emph{performative prediction}~\citep{perdomo2020performative}). 
Here, agents are actively \emph{learning} about the game being played, often through repeated interaction.
This raises the natural questions of: a) what learning algorithms should the agents use when interacting with one another, and b) do the ensuing \emph{game dynamics} converge \emph{in a day-to-day sense} to some Nash equilibrium?

These questions have a decades-long history in an interdisciplinary community spanning economics~\citep{foster1997calibrated,foster1999regret,hart2000simple,hart2005adaptive}, online learning~\citep{hannan1957approximation,freund1999adaptive,cesa2006prediction}, control and dynamical systems~\citep{hofbauer1998evolutionary,fox2013population,bailey2018multiplicative}. 
In fact, game dynamics remain generally difficult to analyze and characterize, and most positive results in the literature are limited in scope. For example, it is known that a very specific class of dynamics satisfying \emph{vanishing (external) regret}---initially introduced by~\cite{hannan1957approximation}\footnote{We discuss the implications of stronger notions of regret, such as internal/swap regret~\citep{blum2007external,greenwald2003general}, in Section~\ref{sec:related-work}.}---played by one player, combined with best-response play by the other, converges to a Nash equilibrium in the special case of zero-sum games. When both players use no-regret strategies, the dynamics are only guaranteed to converge to a broader notion of equilibrium, a \emph{coarse correlated equilibrium}~\citep{cesa2006prediction} (CCE), in general-sum games. Importantly, these convergence guarantees hold only in a \emph{time-average} sense: i.e. the average of the strategies over time converges. 

This provides no guarantee for the actual play at any given round.
In fact, guarantees for convergence in a \emph{last-iterate} sense---where the strategy at each individual time step $t$ eventually approaches equilibrium----are far more limited. Most results in this direction are negative, showing non-convergence or even chaotic behavior in broad classes of games. Our object of focus, the \emph{Exponential Weights} dynamic~\citep{cesa1997use} (commonly called the replicator dynamics in continuous time as surveyed in~\cite{hofbauer1998evolutionary}), has been shown to be oscillatory for simple zero/constant-sum games~\citep{bailey2018multiplicative,cheung2018multiplicative} and more generally can even be chaotic, both for zero-sum and general-sum games~\citep{cheung2019vortices,palaiopanos2017multiplicative,chotibut2021family,cheung2020chaos,cheung2021evolution}. While an optimistic variant of exponential weights (and/or, relatedly, an extragradient method) is commonly employed to mitigate the oscillations in the zero-sum case~\citep{daskalakis2019last,cai2022finite}, a fix for non-zero-sum games remains elusive.
In fact, there exist general-sum finite games for which \emph{any} uncoupled, ``first-order" dynamic would not converge~\citep{hart2003uncoupled}. 

In this paper, we establish a new set of surprisingly positive last-iterate convergence results for a previously overlooked class of games: two-player symmetric games with two actions per player, i.e. \emph{$2 \times 2$ symmetric games}.
Formally, these are normal-form games with 2 players, 2 actions, and where the payoff matrices of players $1$ and $2$, respectively $A$ and $B$, satisfy\footnote{Note that this is a zero-sum game when $A$ is a skew-symmetric matrix.} $B = A^\top$. 
Symmetric games describe many practical scenarios of interest. For example, applications such as road traffic and packet routing on networks can be modeled as congestion games---where strategies correspond to path choice, and ``traffic" affects all participants equally. These games, in their atomic form, are known to correspond to symmetric games~\citep{rosenthal1973class}. Symmetry also arises in Bertrand competition games, which inspire the canonical setting of multi-agent performative prediction~\citep{narang2023multiplayer} as well as our practically motivated case study to follow.

We center our study on the \emph{Exponential Weights} (EW) algorithm, one of the most celebrated no-regret learning rules~\citep{cesa1997use}. EW is widely used in online learning applications (both full-information and bandit learning) due to its simplicity and optimal (worst-case) regret guarantees. Importantly, we identify a significantly stronger form of convergence compared to the related literature: we are able to show not only that our dynamics i) converge in a \emph{last-iterate} (as opposed to the standard time-average) sense in $2\times 2$ symmetric games, but also that ii) they converge to specific Nash equilibria of our game, as opposed to the more general class of CCEs.

\paragraph{Our contributions:} We conduct a systematic and comprehensive study of the exponential weights dynamic in the simple but non-trivial case of $2 \times 2$ symmetric games defined above.
Note that even this simple setting admits a multiplicity of Nash and correlated equilibria --- see Table~\ref{Table:NEs} for a complete characterization.
In particular, we show the following results:
\begin{itemize}
\vspace{-2mm}
    \item Section  \ref{Section:convergecnce} shows that, surprisingly, the exponential weights dynamic always converges to either a pure or a strictly mixed NE for \emph{any} ``non-degenerate" initialization and on \emph{any} ``non-degenerate" $2 \times 2$ symmetric game\footnote{See Assumption~\ref{as:nondegenerate} for formal assumptions. A ``degenerate" initialization would be one for which the dynamics never move, and a ``degenerate" game would be one for which any tuple of strategies constitutes a NE.}. In the cases where convergence to a pure NE is achieved, the step size/learning rate used in exponential weights can be any positive constant.
    When convergence to a strictly mixed NE is achieved, the step size needs to be smaller than a positive constant that depends on the parameters of the game. We obtain these results by conducting an intricate, first-principles, case-by-case analysis of the exponential weights dynamic. 
\vspace{-2mm}
    \item In Section~\ref{sec:applications}, we interpret our results in the context of multi-agent performative prediction between $2$ agents. We design a stylized, but canonical model of mortgage competition between two banks (or ``learners"), where the strategies consist of interest rates and credit score thresholds, and the utilities of the banks depend on the distribution of the customers who choose each of the banks, as well as the probability of repaying their loans.
    Interestingly, our results in Section~\ref{Section:convergecnce} imply that convergence to \emph{asymmetric} pure NE can happen (despite the fact that a symmetric NE, either pure or mixed, always exists for symmetric games).
\end{itemize}
We next set up the formal framework for $2 \times 2$ symmetric games, and present our main results in Sections~\ref{Section:convergecnce} and~\ref{sec:applications}.
We provide a detailed treatment of related work in Section~\ref{sec:related-work}.



\section{Preliminaries}\label{sec:prelim}

We begin by introducing the basic setting of $2$-player, $2$-action symmetric games.
We then briefly review basic equilibrium concepts, and present (for completeness) the Nash and correlated equilibrium landscape of this setting.

\paragraph{Two-player-two-action symmetric games:} We study a general two-player, two-action symmetric game.  
Each player $i \in \{1,2\}$ selects one of two available actions (i.e. pure strategies), which we denote by $\theta_1$ and $\theta_2$.  
The corresponding payoffs are summarized in Table~\ref{tab:utilities111}. 
Note that the $2 \times 2$ utility matrix $A$ is parameterized by fixed constants $a, b, c, d \in \mathbb{R}$. We also define
\begin{align}\label{eq:eps-defs}
\textstyle \epsilon_1 := a - b \text{ and } \epsilon_2 = c - d
\end{align}
as two shorthand parameters---we will see that the signs of these parameters critically dictate both the equilibrium landscape and the behavior of exponential weights.

For any pair of indices $m,n \in \{1,2\}$, let $u_i(\theta_m, \theta_n)$ denote the payoff of player $i\in\{1,2\}$ when player 1 chooses $\theta_m$ and player 2 chooses $\theta_n$. 
Note that by symmetry of the game, we have $u_1(\theta_m, \theta_n) \;=\; u_2(\theta_n, \theta_m)$ for all $m,n \in \{1,2\}$.
We further define a pair of mixed strategies $(p_1,p_2)$ where $p_i := (p_{i,1},p_{i,2})$ denotes a probability distribution (on the $2$-dimensional simplex denoted by $\Delta_2$) over pure strategies $\theta_1$ and $\theta_2$ respectively for player $i \in \{1,2\}$. Next, we recall the following basic equilibrium concepts.

\begin{table}[t]
\centering
\scalebox{0.9}{ 
\begin{tabular}{c}
\textbf{Player 1 utility matrix $A$} \\
\begin{tabular}{|c|c|c|}
\Xhline{1.2pt}
\diagbox{\textbf{Player 2}}{\textbf{Player 1}} & \textbf{$\theta_1$} & \textbf{$\theta_2$} \\
\Xhline{1.2pt}
\textbf{$\theta_1$} & $a$ & $b$ \\
\hline
\textbf{$\theta_2$} & $c$ & $d$ \\
\Xhline{1.2pt}
\end{tabular}
\end{tabular}
\hspace{2cm}
\begin{tabular}{c}
\textbf{Player 2 utility matrix $B = A^\top$} \\
\begin{tabular}{|c|c|c|}
\Xhline{1.2pt}
\diagbox{\textbf{Player 2}}{\textbf{Player 1}} & \textbf{$\theta_1$} & \textbf{$\theta_2$} \\
\Xhline{1.2pt}
\textbf{$\theta_1$} & $a$ & $c$ \\
\hline
\textbf{$\theta_2$} & $b$ & $d$ \\
\Xhline{1.2pt}
\end{tabular}
\end{tabular}
} 
\caption{Utility matrices of Player 1 and Player 2}
\label{tab:utilities111}
\end{table}

\begin{definition}[Nash equilibrium]
A pair of strategies \((\theta_{m^*}, \theta_{n^*})\) constitutes a \emph{pure Nash equilibrium} (pure NE) iff the following conditions hold::
\[
u_1(\theta_{m^*}, \theta_{n^*}) \geq \max_{\theta'\in\{\theta_1,\theta_2\}} u_1(\theta', \theta_{n^*}), \quad \text{and} \quad u_2(\theta_{m^*}, \theta_{n^*}) \geq \max_{\theta' \in \{\theta_1,\theta_2\}} u_2(\theta_{m^*}, \theta').
\]
\noindent A pair of mixed strategies \((p_1^*, p_2^*)\) constitutes a \emph{mixed Nash equilibrium} (mixed NE) iff the following two conditions hold:  
\begin{align*}
&\E_{M \sim p^*_1, N \sim p^*_2}[u_1(\theta_M, \theta_N)] \geq \max_{p_1 \in \Delta_2} \E_{M' \sim p_1, N \sim p^*_2}[u_1(\theta_{M'},\theta_N)],\\
&\E_{M \sim p^*_1, N \sim p^*_2}[u_2(\theta_M, \theta_N)] \geq \max_{p_2 \in \Delta_2} \E_{M \sim p^*_1, N' \in p_2}[u_2(\theta_M, \theta_N)].
\end{align*}
\end{definition}

In Nash equilibria, players pick their actions simultaneously and independently of each other. 
While our central focus in this paper is convergence to NE, we also briefly consider the broader concept of \emph{correlated equilibria} (CE) of a game.
This concept was introduced by~\cite{aumann1987correlated} and, informally, uses a signaling device that recommends some action $(\theta_M,\theta_N)$ to the players where $(M,N)$ is now drawn from a \emph{joint distribution} (for example, think of a traffic signal at an intersection).
The signal induces implicit correlation between the players.
At a correlated equilibrium, no player benefits from deviating from the recommended action in expectation. We provide a detailed review of the definition of CE, as well as characterize the set of CE for our game, in Appendix \ref{section:Characterization of CE}.    

\begin{table}[t]
\centering
\scalebox{0.9}{%
\begin{tabular}{|c|c|c|c|}
\hline
Condition & Pure NE & Strict Mixed-NE & CE \\
\hline
$\epsilon_1<0,\epsilon_2<0$ & $(\theta_2,\theta_2)$ & $-$ & $(\theta_2,\theta_2)$ \\
\hline
$\epsilon_1>0, \epsilon_2>0$ & $(\theta_1,\theta_1)$ & $-$ & $(\theta_1,\theta_1)$\\
\hline
$\epsilon_1<0,\epsilon_2>0$ & $(\theta_1,\theta_2)$, $(\theta_2,\theta_1)$ & $\left(p_{\text{SE}},p_{\text{SE}}\right)$ & Given in \eqref{eqn:CEl0g0xx}    \\
\hline
$\epsilon_1>0,\epsilon_2<0$ & $(\theta_1,\theta_1)$, $(\theta_2,\theta_2)$  & $\left(p_{\text{SE}},p_{\text{SE}}\right)$ &  Given in \eqref{eqn:CEg0l0xx} \\
\hline
$\epsilon_1=0,\epsilon_2<0$ & $(\theta_1,\theta_1)$, $(\theta_2,\theta_2)$  & $-$ & Given in \eqref{eqn:2133121}\\
\hline
$\epsilon_1=0,\epsilon_2>0$ & $(\theta_1,\theta_1)$, $(\theta_2,\theta_1)$, $(\theta_1,\theta_2)$  & $(\theta_1,p), (p,\theta_1)$ & Given in \eqref{eqn:ddfdgjndauih}\\
\hline
\end{tabular}
}
\caption{NE under different conditions. Recall that $\epsilon_1=a-b,$ while $\epsilon_2=c-d$.  $p\in\Delta_2$ stands for any mixed strategy in the $2$-dimensional simplex. We define $p_{\text{SE}}=\left(\frac{|\epsilon_2|}{|\epsilon_1|+|\epsilon_2|},\frac{|\epsilon_1|}{|\epsilon_1|+|\epsilon_2|}
\right)$.}
\label{Table:NEs}
\vspace{-0.5cm}
\end{table}

\paragraph{Landscape of Nash and correlated equilibria:}

We now present the NE landscape of our canonical setting of symmetric $2 \times 2$ games in Table \ref{Table:NEs} under different conditions, dictated by the signs of $\epsilon_1 := a - b$ and $\epsilon_2:= c - d$.
(Proving that the corresponding entries in the table are correct is a simple exercise in undergraduate game theory that we leave to the reader.)
We make a few observations about the landscape of NEs:
\begin{itemize}
\vspace{-2mm}
\item When the signs of $\epsilon_1$ and $\epsilon_2$ match (both strictly positive or negative), there is a unique, \emph{symmetric}, pure NE. It is also easy to verify that these cases have a dominated strategy for both players (e.g. in the case where $\epsilon_1, \epsilon_2 < 0$, $\theta_1$ is dominated by $\theta_2$ for both players). Moreover, as we will subsequently see the set of CEs is also a singleton set comprising of the unique pure NE. Section~\ref{Section:convergecnce} will demonstrate that these are the ``easy" cases for convergence to NE.
\vspace{-2mm}
\item The cases where the signs of $\epsilon_1$ and $\epsilon_2$ \emph{do not} match (including cases where $\epsilon_1 = 0$ or $\epsilon_2 = 0$) are much more interesting as they admit a multiplicity of pure NE as well as potentially \emph{strictly} mixed NE. The set of CEs are also more complex in these cases, rendering a generic no-regret convergence analysis inadequate for our purposes. Proving convergence for these more complex cases is the main novelty in our analysis.
\vspace{-2mm}
\item We do not consider the case of $\epsilon_1 = \epsilon_2 = 0$ as this corresponds to a degenerate game where any tuple of mixed strategies would be a valid NE; see Assumption~\ref{as:nondegenerate} for more details.
\end{itemize}

\section{Last-iterate Convergence of Exponential Weights} 
\label{Section:convergecnce}
\begin{center}
\scalebox{0.9}{%
\begin{minipage}{\linewidth}
\begin{algorithm}[H] 
\caption{Online Learning Dynamic through Exponential Weights}
\begin{algorithmic}
\label{alg:Hedge}
\STATE \textbf{Input}: Step size $\eta>0$, Initialization $p_1^{(1)},p_2^{(1)}\in\Delta_2$
\FOR{$t=1,2,3\dots$}
\STATE \emph{Decision:} Each player $i\in\{1,2\}$ plays $\theta_j$ with probability $p_{i,j}^{(t)}$ for $j\in\{1,2\}$;
\STATE \emph{Update Rule:} Players update their probability distributions to $p_{1}^{(t+1)},~p_{2}^{(t+1)}$ as follows:
\begin{align*}
&p^{(t+1)}_{1,j}\propto p^{(t)}_{1,j} \cdot \exp\left(\eta \sum_{k=1}^2 p^{(t)}_{2,k} \cdot u_1\left(\theta_j,\theta_k\right)\right),~~\forall j\in\{1,2\};\\
&p^{(t+1)}_{2,j}\propto p^{(t)}_{2,j} \cdot \exp\left(\eta \sum_{k=1}^2 p^{(t)}_{1,k} \cdot u_2\left(\theta_k,\theta_j\right)\right),~~\forall j\in\{1,2\}.
\vspace{-1cm}
\end{align*}
\ENDFOR
\end{algorithmic}
\end{algorithm}
\end{minipage}%
}
\end{center}

We now characterize the last-iterate convergence of the exponential weights algorithm, formally described in Algorithm~\ref{alg:Hedge}.
We use $p_i^{(t)} \in \Delta_2$ to denote the mixed strategy used by player $i \in \{1,2\}$ at time step $t \geq 1$.
Thus, when both players play exponential weights, it induces a sequence of mixed strategy tuples $\{p_1^{(t)},p_2^{(t)}\}_{t \geq 1}$.
Our hope is that this sequence will eventually converge to some NE.

The most salient step in the exponential weights algorithm is the update from $t \to t+1$, which proceeds as follows. At each time step $t$, each player $i \in \{1,2\}$ randomly picks an action according to their current probability distribution $p_i^{(t)}$. Then, each player updates their weight on each action in a manner that is exponentially proportional to the expected utility derived by playing that action (which in turn depends on the mixed strategy of the other player). For example, considering round $t$ and candidate action $\theta_j$, player $1$ obtains an expected utility of $\sum_{k=1}^2 p^{(t)}_{2,k} \cdot u_1\left(\theta_j,\theta_k\right)$ since player $2$'s probability distribution is given by $p_2^{(t)}$. The speed of this update is controlled by a \emph{constant} step size $\eta>0$.
(As of now, our proofs, particularly in the more complex cases, do require both players to use the same step size; relaxing this requirement is an important future direction.)

Recall that we define the shorthand parameters $\epsilon_1 := a - b, \epsilon_2 := c - d$.
Let us define the shorthand functional $
\Delta_i^{(t)} = p^{(t)}_{i,1}\epsilon_1 + p^{(t)}_{i,2}\epsilon_2$.
Before diving into the details, we first make the following very mild \emph{non-degeneracy} assumption about the games and the algorithm's initialization:

\begin{assu}\label{as:nondegenerate}
Throughout the paper, we assume: 1) $|\epsilon_1|+|\epsilon_2|>0$; 2) $p^{(1)}_i$ is not a pure strategy for either player $i\in\{1,2\}$; 3) the initial functionals $\Delta_1^{(1)}$ and $\Delta_2^{(1)}$ are both non-zero. 
\end{assu}

The assumptions above only exclude trivial cases that are straightforward to analyze. First, note that when $\epsilon_1=\epsilon_2=0$, the algorithm does not update---since all strategy profiles in this case are a mixed NE, the algorithm in fact ``converges'' (in a single time step) to some NE. 
Regarding 2), if both $p^{(1)}_1$ and $p^{(1)}_2$ are pure strategies, then the algorithm does not update; if, for example, $p^{(1)}_1$ is a pure strategy but $p^{(1)}_2$ is not, then $p^{(t)}_1$ remains the same for all $t$, while $p^{(t)}_2$ converges to the corresponding best response---clearly, the resulting tuple need not be a NE. 

Finally and regarding 3), if both $\Delta_1^{(1)}=\Delta_2^{(1)}=0$, the algorithm does not update; if only one of them is $0$, then by round $t=2$ both become non-zero and automatically satisfy Assumption~\ref{as:nondegenerate} thereafter. This is formalized in the following lemma, which is proved in Appendix \ref{append:proofoflemma111sdewq}. 

\begin{lemma}
\label{lem:::verysimpleobsercation}
Suppose $|\epsilon_1|+|\epsilon_2|>0$, and $p^{(1)}_1$ and $p^{(1)}_2$ are not pure strategies. We have the following: 
\begin{itemize}
    \item If $\Delta_1^{(1)}=\Delta_2^{(1)}=0,$ then $p^{(t)}_i$ remains the same for all $t\geq 1$ and $i\in\{1,2\}$.
    \item  If (without loss of generality) $\Delta_1^{(1)}=0$ and $\Delta_2^{(1)}\not=0$, then $\Delta_1^{(2)}, \Delta_2^{(2)}\not=0$. 
\end{itemize}
\end{lemma}

\subsection{Main result: convergence of the last iterate}

Our main result, Theorem 1 (Parts I-IV), fully characterizes the convergence of Algorithm \ref{alg:Hedge}. 
Our results in a nutshell are informally summarized in Table~\ref{Table:NEconvergence}.

\begin{table}[t]
\centering
\scalebox{0.88}{
\begin{tabular}{|c|c|c|c|}
\hline
Utility Condition & Initialization & Nature of result & Requirement on $\eta$ \\
\hline
r1. $\epsilon_1<0,\epsilon_2<0$ & Any & Exponential convergence to pure NE & None \\
\hline
r2. $\epsilon_1>0, \epsilon_2>0$ & Any & Exponential convergence to pure NE & None\\
\hline
r3. $\epsilon_1<0,\epsilon_2>0$ & Opposite-sign & Exponential convergence to one of pure NEs & None   \\
\hline
r4. $\epsilon_1<0,\epsilon_2>0$ & Same sign & Asymptotic convergence to one of pure NEs  & None    \\
\hline
r5. $\epsilon_1<0,\epsilon_2>0$ & Identical & Asymptotic convergence to strictly mixed NE & $< \frac{8}{|\epsilon_1| + |\epsilon_2|}$    \\
\hline
r6. $\epsilon_1>0,\epsilon_2<0$ & Same sign  & Exponential convergence to one of the pure NEs &  None \\
\hline
r7. $\epsilon_1>0,\epsilon_2<0$ & Opposite signs & Asymptotic convergence to one of the NEs & $< \frac{8}{|\epsilon_1| + |\epsilon_2|}$ \\
\hline
r8. $\epsilon_1=0,\epsilon_2<0$ & Any  & Exponential convergence to pure NE & None \\
\hline
r9. $\epsilon_1=0,\epsilon_2>0$ & Identical  & Asymptotic convergence to symmetric pure NE & None \\
\hline
r10. $\epsilon_1=0,\epsilon_2>0$ & Non-identical  & Asymptotic convergence to set of mixed NE & None\\
\hline
\end{tabular}
}
\caption{Summary of our convergence results. This table is best read along with Table~\ref{Table:NEs}, which summarizes the NE characterization for each case.}
\label{Table:NEconvergence}
\vspace{-0.4cm}
\end{table}

We start with the most \emph{naive} case, where $\epsilon_1$ and $\epsilon_2$ have the same sign.
This case is naive because it implies the existence of a common dominating action for both players, which also determines the unique pure NE.
The following result, which is proved in Appendix \ref{appdnex:profoofpar1}, shows that in this case the exponential weights dynamic converges to the corresponding pure NE at an exponentially fast rate.

\begin{theoremfour}{Part I}
\label{thm:same-sign}
Suppose $\epsilon_1$,  $\epsilon_2$ are non-zero and have the same sign. Then, Algorithm \ref{alg:Hedge} converges to a) the pure NE  $(\theta_2,\theta_2)$ in the case where $\epsilon_1, \epsilon_2 < 0$, and b) the pure NE $(\theta_1,\theta_1)$ in the case where $\epsilon_1, \epsilon_2 > 0$, at an exponential rate for any step size $\eta > 0$ (see Theorem~\ref{thm:partone-formal} in Appendix~\ref{appdnex:profoofpar1} for the explicit constant in the rate).
\end{theoremfour}

Next, we consider the more interesting and non-trivial cases where $\epsilon_1$ and  $\epsilon_2$ are non-zero and have opposite signs. 
We begin with the case where $\epsilon_1 < 0$ and $\epsilon_2 > 0$.
Recall from Table~\ref{Table:NEs} that this case admits two pure \emph{asymmetric} NEs $(\theta_1,\theta_2)$ and $(\theta_2,\theta_1)$, and there is also a strictly mixed NE given by $(p_{\text{SE}},p_{\text{SE}})$, where $p_{\text{SE}}=\left(\frac{|\epsilon_2|}{|\epsilon_1|+|\epsilon_2|},\frac{|\epsilon_1|}{|\epsilon_1|+|\epsilon_2|}
\right)$.  
The following result, which is proved in Appendix~\ref{appendix"thm:e1l0e2g0}, shows a rich variety of convergence behaviors of the exponential weights dynamic depending on the initialization. We also illustrate the results through simulation in Figure \ref{figure1111111111}.

\begin{figure}
  \centering

  \begin{subfigure}{0.45\textwidth}
    \includegraphics[width=\linewidth]{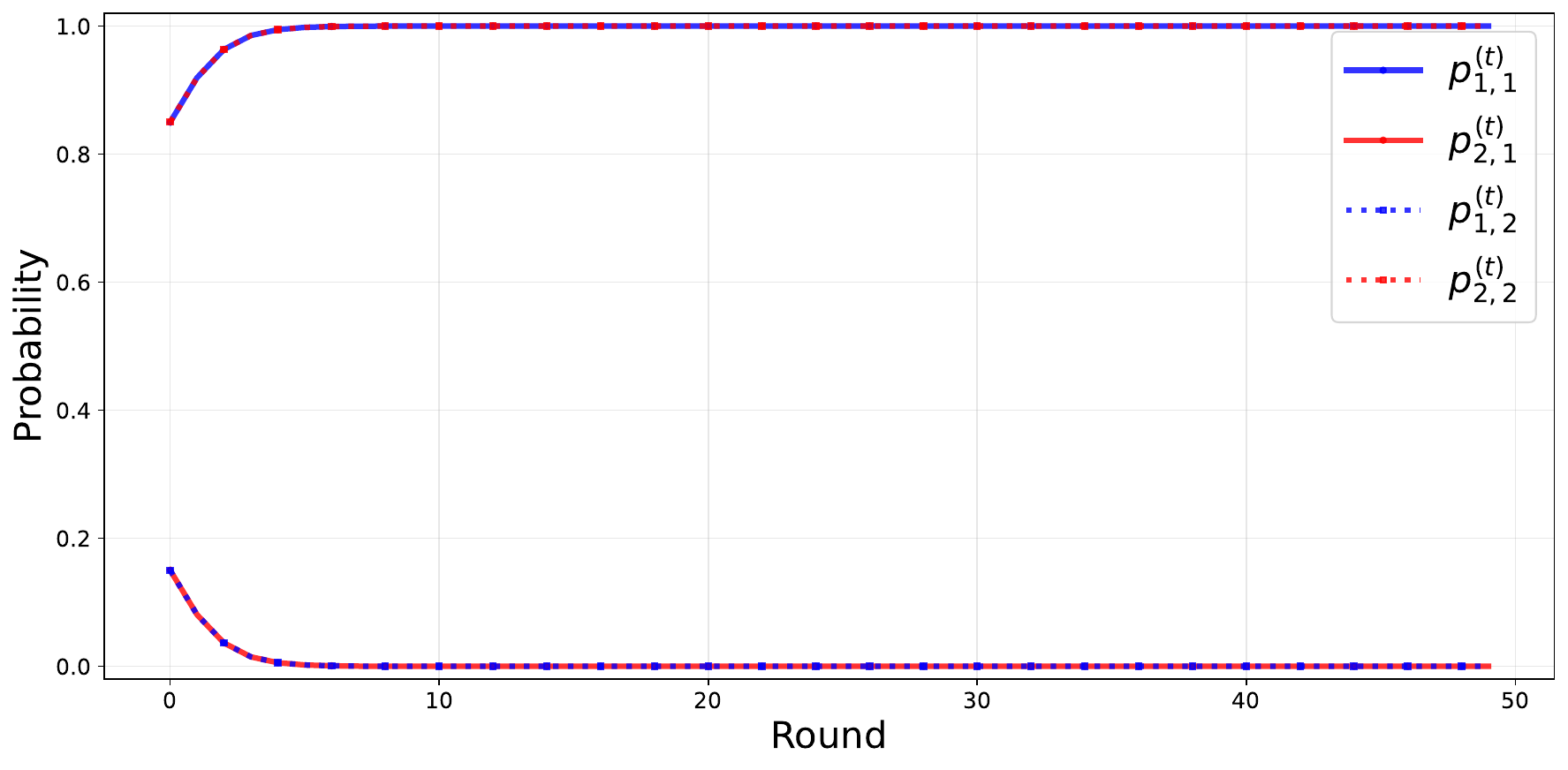}
    \caption*{r3. Opposite-sign initialization}
  \end{subfigure}
  \begin{subfigure}{0.45\textwidth}
    \includegraphics[width=\linewidth]{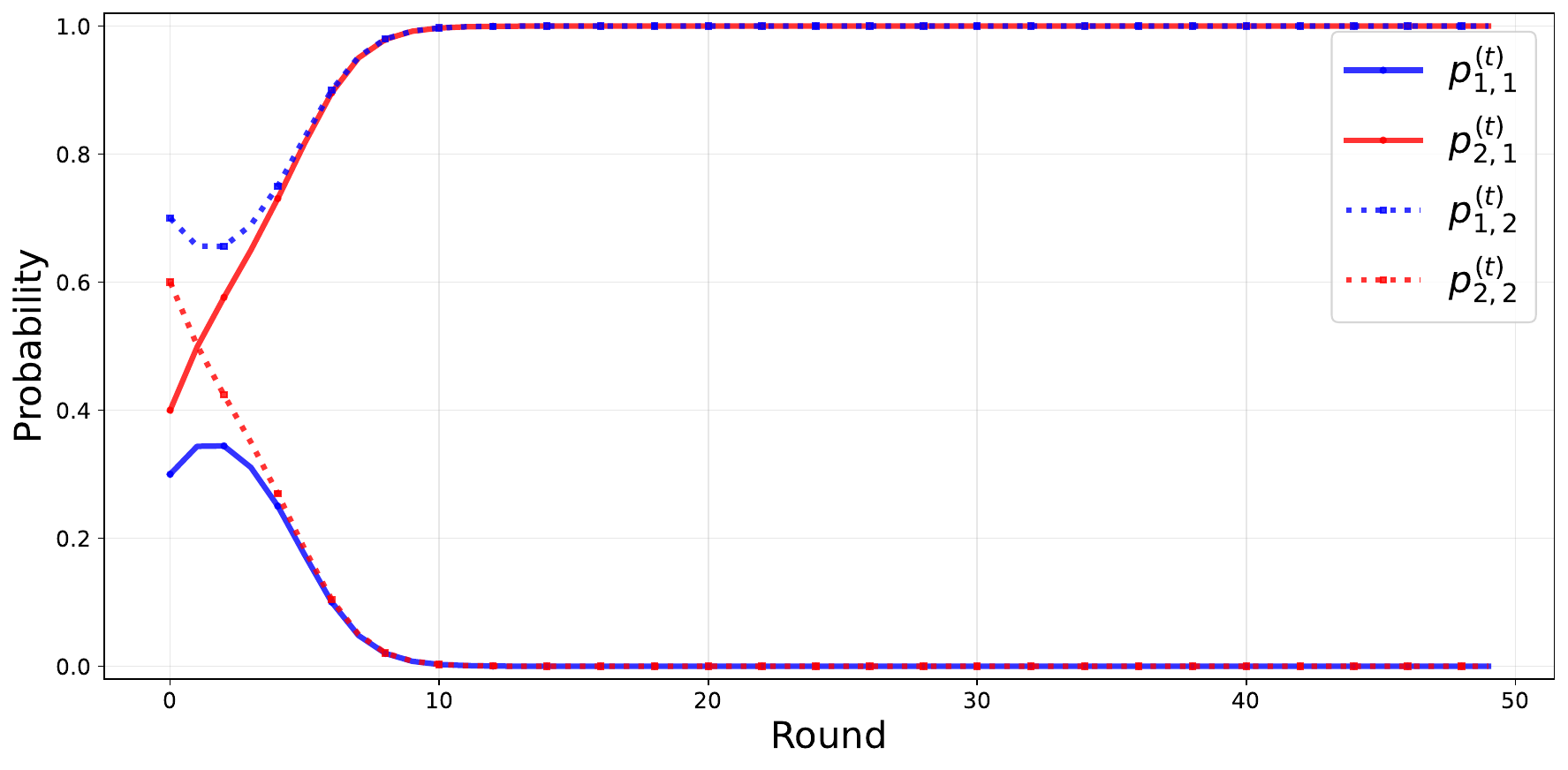}
    \caption*{r4. Same-sign initialization}
  \end{subfigure}
  \begin{subfigure}{0.45\textwidth}
    \includegraphics[width=\linewidth]{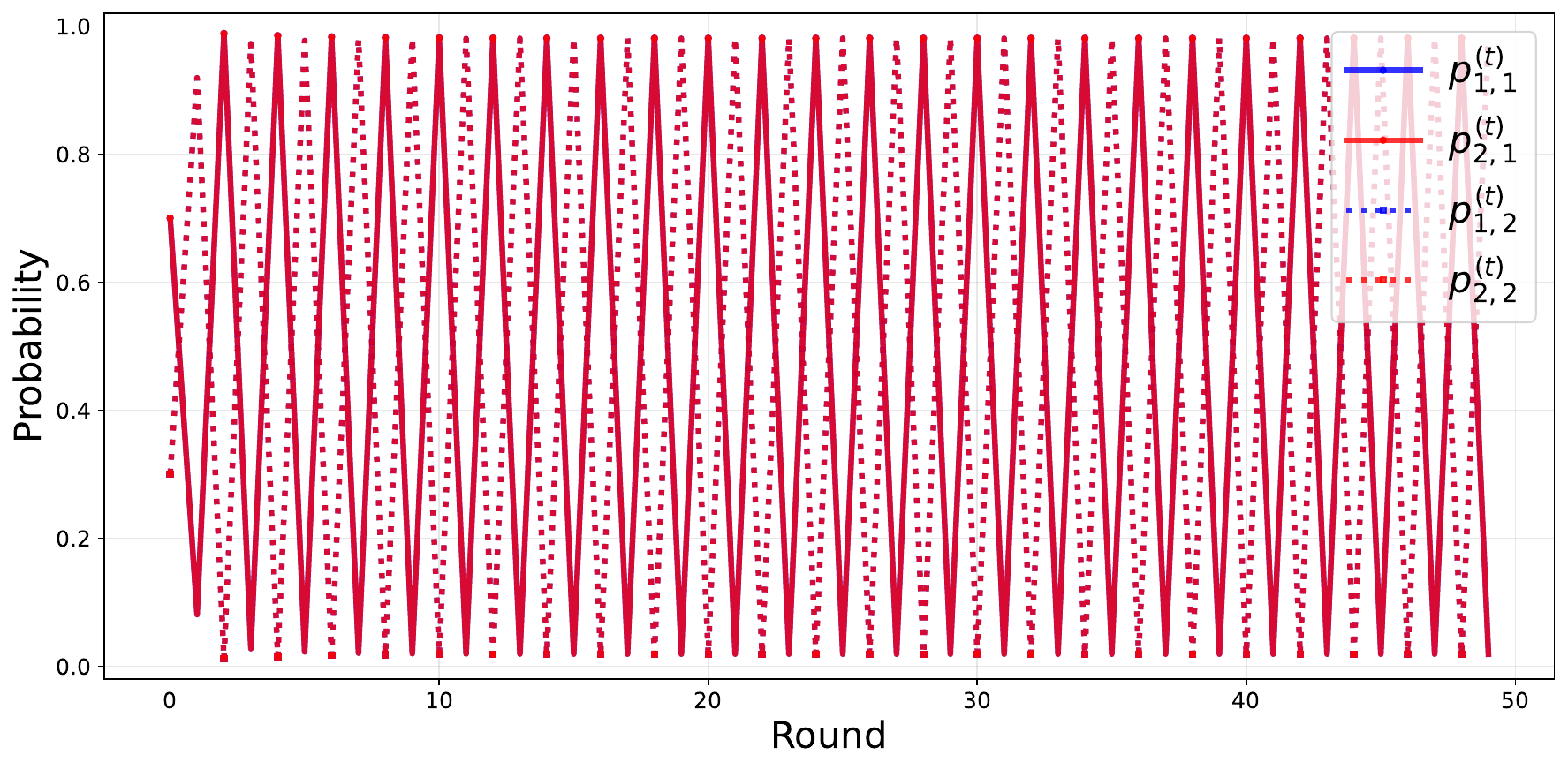}
    \caption*{r5. Identical initialization, with $\eta>\frac{8}{|\epsilon_1|+|\epsilon_2|}$}
  \end{subfigure}
  \begin{subfigure}{0.45\textwidth}
    \includegraphics[width=\linewidth]{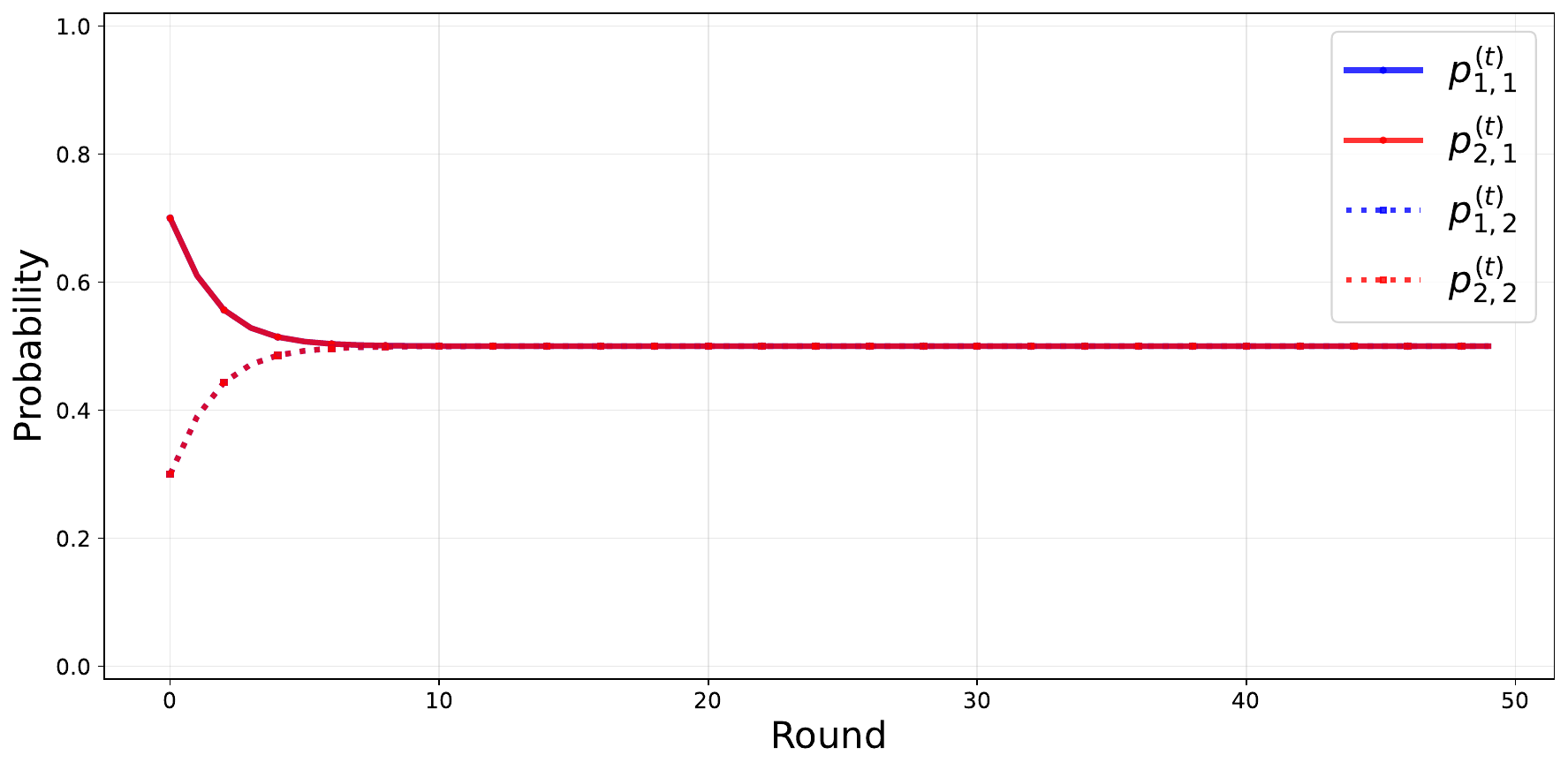}
    \caption*{r5. Identical initialization, with $\eta<\frac{8}{|\epsilon_1|+|\epsilon_2|}$}
  \end{subfigure}

  \caption{Simulation results for Theorem \ref{thm:e1l0e2g0}. The row in Table~\ref{Table:NEconvergence} to which each case corresponds is marked in the sub-captions. Corresponding simulations for Parts I, III and IV are in the appendices due to space limitations.}
  \vspace{-0.4cm}
  \label{figure1111111111}
\end{figure}

\begin{theoremfour}{Part II}
\label{thm:e1l0e2g0}  Suppose $\epsilon_1<0$, and $\epsilon_2>0$. Then:
\begin{itemize}
    \item \vspace{-2mm} (Opposite-sign initialization) If $\Delta_1^{(1)}$ and $\Delta_2^{(1)}$ have different signs, then Algorithm \ref{alg:Hedge} converges to one of the pure NEs $(\theta_1,\theta_2)$ or $(\theta_2,\theta_1)$ at an exponential rate (see Theorem~\ref{thm:parttwo-easycase-formal} in Appendix~\ref{appendix"thm:e1l0e2g0} for the explicit constant in the rate);
    \item \vspace{-2mm} (Same-sign but not identical initialization) If $\Delta_1^{(1)}$ and $\Delta_2^{(1)}$ have the same sign and $\Delta_{1}^{(1)}\not=\Delta_{2}^{(1)}$, then Algorithm \ref{alg:Hedge} converges asymptotically to one of the pure NEs $(\theta_1,\theta_2)$ or $(\theta_2,\theta_1)$;
    \item \vspace{-2mm}(Identical initialization) If $\Delta_1^{(1)}=\Delta_2^{(1)}$, and $\eta\in\left(0,\frac{8}{|\epsilon_1|+|\epsilon_2|}\right)$, then Algorithm \ref{alg:Hedge} converges to the strictly mixed NE $(\frac{\epsilon_2}{|\epsilon_1|+\epsilon_2},\frac{|\epsilon_1|}{|\epsilon_1|+\epsilon_2})$ asymptotically.  
\end{itemize}
\end{theoremfour}
Interestingly, Theorem~\ref{thm:e1l0e2g0} shows that the algorithm always converges to some pure NE when the initialization functionals $\Delta^{(1)}_1$ and $\Delta^{(1)}_2$ differ, but requires a sufficiently small step size $\eta \in \left(0, \frac{8}{|\epsilon_1| + |\epsilon_2|}\right)$ if the initialization is identical.
Note that this requirement on the step size is, however, relatively mild: if we normalized the utilities to be in $[-1,1]$, we have $|\epsilon_1|, |\epsilon_2| \leq 2$, and so the condition reduces to $\eta < 2$. For $\eta > \frac{8}{|\epsilon_1| + |\epsilon_2|}$, we also provide a simple example where the algorithm oscillates between two pairs of mixed strategies that are not mixed NE for such prohibitively large values of $\eta$. Moreover, in the special case of ($2 \times 2$) congestion games, \cite[Section 4]{palaiopanos2017multiplicative} provides counterexamples for specific large values of $\eta$, which we show exceed the threshold ($\eta > \frac{8}{|\epsilon_1| + |\epsilon_2|}$). { Finally, \cite{chotibut2021family} establishes last-iterate convergence of the multiplicative weights algorithm for $2 \times 2$ congestion games, which constitute a special case of this class of games with $\epsilon_1 = -\epsilon_2$. Therefore, our results recover their conclusion for congestion games.}

Next, we consider the case where $\epsilon_1>0$ and $\epsilon_2<0$. Recall that in this case the pure NEs are $(\theta_1,\theta_1)$ and $(\theta_2,\theta_2)$. There also exists a strictly mixed NE in the form of $(p_{\text{SE}},p_{\text{SE}})$.
The following result, which is proved in Appendix \ref{appendix:thm:theoppositcase}, again shows a rich variety of convergence behaviors.

\begin{theoremfour}{Part III}
\label{thm:theoppositcase} 
Suppose $\epsilon_1>0$, and $\epsilon_2<0$. Then: 
\begin{itemize}
    \item \vspace{-2mm} (Same Sign initialization) If $\Delta_{1}^{(1)}$ and $\Delta_{2}^{(1)}$ have the same sign, then Algorithm \ref{alg:Hedge} converges to one of the pure NEs at an exponential rate (see Theorem~\ref{thm:partthree-easycase-formal} in Appendix~\ref{appendix:thm:theoppositcase} for the explicit constant in the rate);
    \item \vspace{-2mm} (Opposite Sign initialization) If  $\Delta_{1}^{(1)}$ and $\Delta_{2}^{(1)}$ have different signs, and $\eta\in\left(0,\frac{8}{|\epsilon_1|+|\epsilon_2|}\right)$, then Algorithm \ref{alg:Hedge} converges to either a pure NE or  strictly mixed-NE asymptotically. 
\end{itemize}
\end{theoremfour}
The above conclusion is similar, but not identical to its counterpart in Theorem \ref{thm:e1l0e2g0}. On one hand, when $\Delta^{(1)}_{1}$ and $\Delta^{(1)}_{2}$ have the same sign, the algorithm converges to a pure NE. On the other hand, when $\Delta^{(1)}_{1}$ and $\Delta^{(1)}_{2}$ have different signs, the algorithm can converge either to the strictly mixed NE or pure NEs, but requires a sufficiently small step size (analogous to the case of identical initialization in Part II).
Interestingly, the proofs of Theorems~\ref{thm:theoppositcase} and~\ref{thm:e1l0e2g0}, proceed differently, despite the seeming symmetry in the cases $(\epsilon_1 < 0, \epsilon_2 > 0)$ and $(\epsilon_1 > 0, \epsilon_2 < 0)$.

Finally, we consider the case where one of $\epsilon_1$ and  $\epsilon_2$ is equal to zero. 
This may seem like a corner case (and is indeed a set of measure zero in the set of all $2 \times 2$ symmetric games), but we can handle it in our framework.
Without loss of generality, we assume $\epsilon_1=0$. (Note that the other case ($\epsilon_2=0$) is purely symmetric to this case, by just switching $\theta_1$ and $\theta_2$.) Recall that in this case, the pure NEs are $(\theta_1,\theta_1)$, $(\theta_1,\theta_2)$ and $(\theta_2,\theta_1)$, while the mixed NEs are $(\theta_1,p)$ and $(p,\theta_1)$ for any $p\in\Delta_2$. 
The following result, which is proved in Appendix \ref{appendix:thm:forcasewithonezero}, shows convergence to different classes of NEs depending on the initialization (and does not possess any requirements on the step size).

\begin{theoremfour}{Part IV}
\label{thm:forcasewithonezero}
Suppose $\epsilon_1=0$. Then:
\begin{itemize}
  \item \vspace{-2mm} If $\epsilon_2<0$, then Algorithm \ref{alg:Hedge} converges to the pure NE $(\theta_2,\theta_2)$ at an exponential rate;
  \item \vspace{-2mm} If $\epsilon_2>0$ and $\Delta_1^{(1)}=\Delta_2^{(1)}$, then Algorithm \ref{alg:Hedge} converges to the pure NE $(\theta_1,\theta_1)$ asymptotically;
  \item \vspace{-2mm} If $\epsilon_2>0$ and $\Delta_1^{(1)}\not=\Delta_2^{(1)}$, then Algorithm \ref{alg:Hedge} converges to a mixed NE. Moreover, there exist examples where the algorithm converges to a strictly mixed NE of the form $(\theta_1,[A,1-A])$, where $A\in(0,1)$ is some constant.    
\end{itemize}
\end{theoremfour}

\paragraph{Reconciliation of Theorem 1 with results showing oscillation/chaos:} It is worth briefly clarifying, given the overwhelmingly negative results for exponential weights dynamics in games, why we are able to obtain convergence results of such a sweeping nature for our setting.
The zero-sum games that induce oscillations are ones for which the \emph{only} NE is strictly mixed~\citep{bailey2018multiplicative,cheung2018multiplicative} (the Matching Pennies game is a common example).
As Table~\ref{Table:NEs} shows, this never happens in a symmetric $2 \times 2$ game.
Moreover,~\cite{cheung2019vortices} provide a sufficient and necessary condition for increasing volume of the flow of exponential weights (and thereby chaotic behavior) which, in the case of $2 \times 2$ symmetric games, can be verified to \emph{never} occur (see~\cite[Appendix G]{cheung2019vortices} for the detailed condition).

\subsection{Proof sketch of Theorem 1}\label{sec:proof-sketch}

\begin{figure}[H]
    \centering
    \scalebox{0.88}{
    \begin{minipage}{\linewidth}
        \begin{subfigure}{0.45\linewidth}
            \centering
            \includegraphics[width=\linewidth]{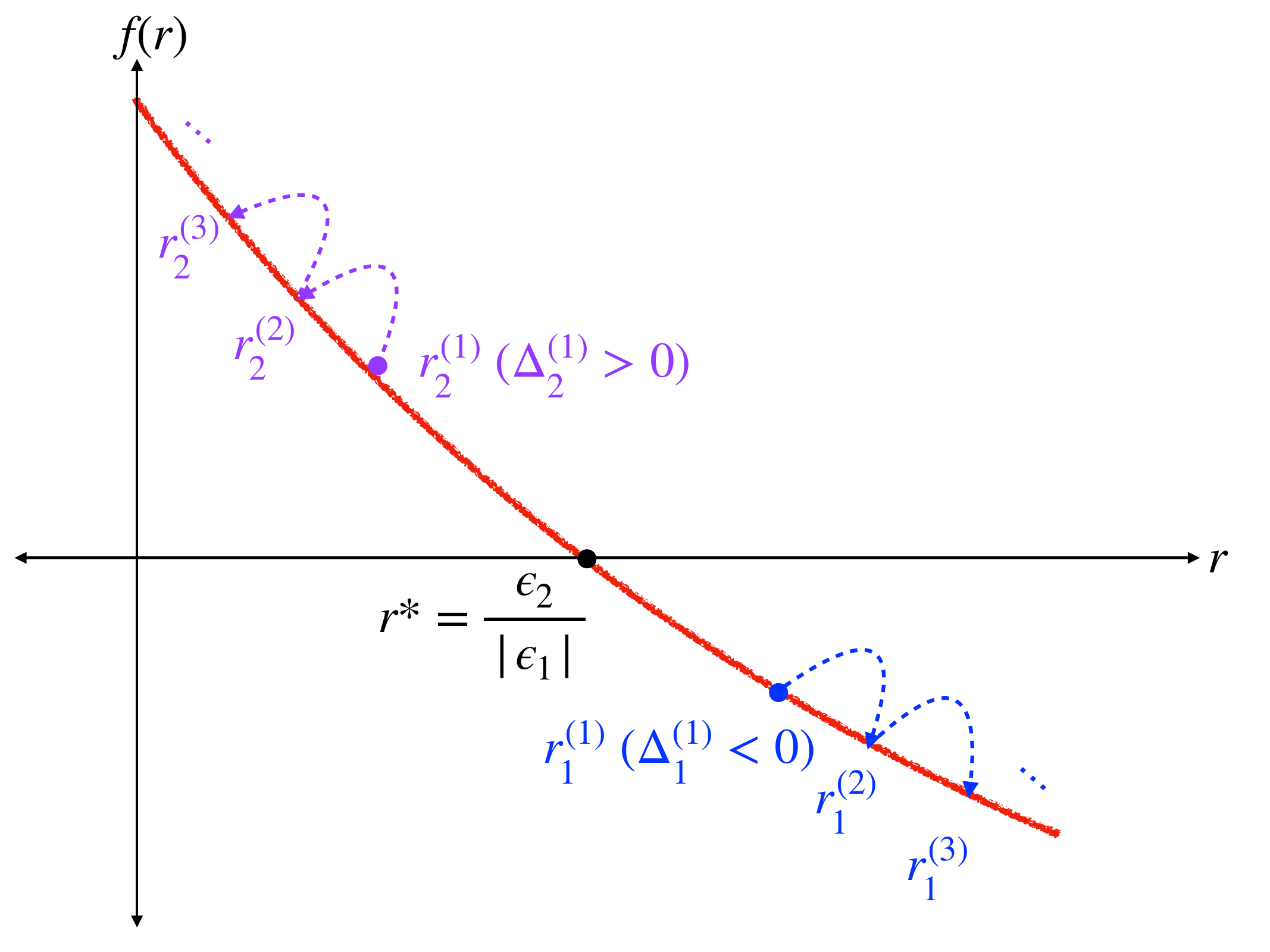}
            \caption{Opposite-sign initialization (Condition 1): 
            \\leads to ratios moving away from each other.}\label{fig:condition1sketch}
        \end{subfigure}
        \begin{subfigure}{0.45\linewidth}
            \centering
            \includegraphics[width=\linewidth]{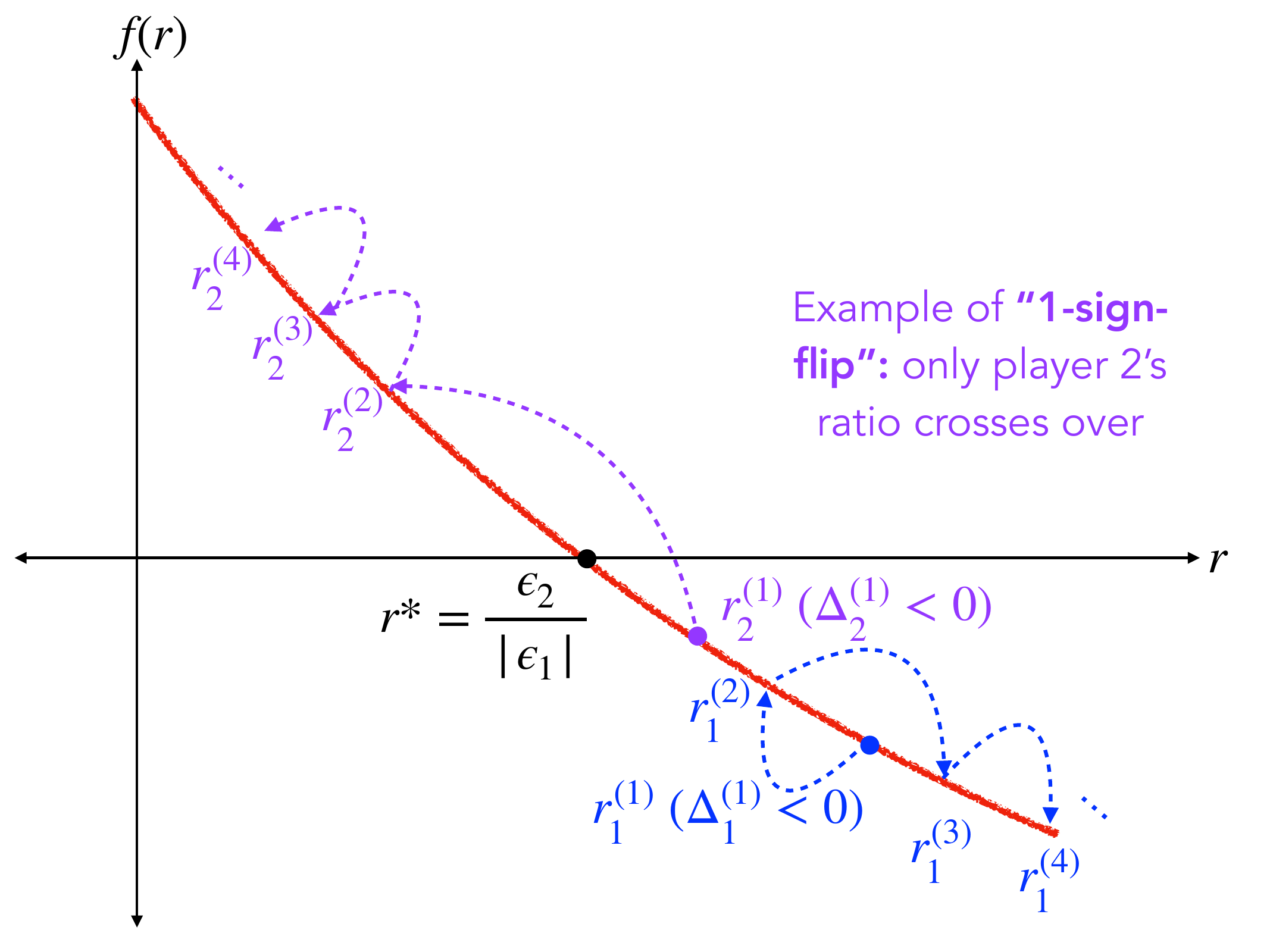}
            \caption{Same-sign initialization with a ``1-sign-flip": 
            \\reduces to Condition 1.}\label{fig:1signflip}
        \end{subfigure}
        
        \begin{subfigure}{0.45\linewidth}
            \centering
            \includegraphics[width=\linewidth]{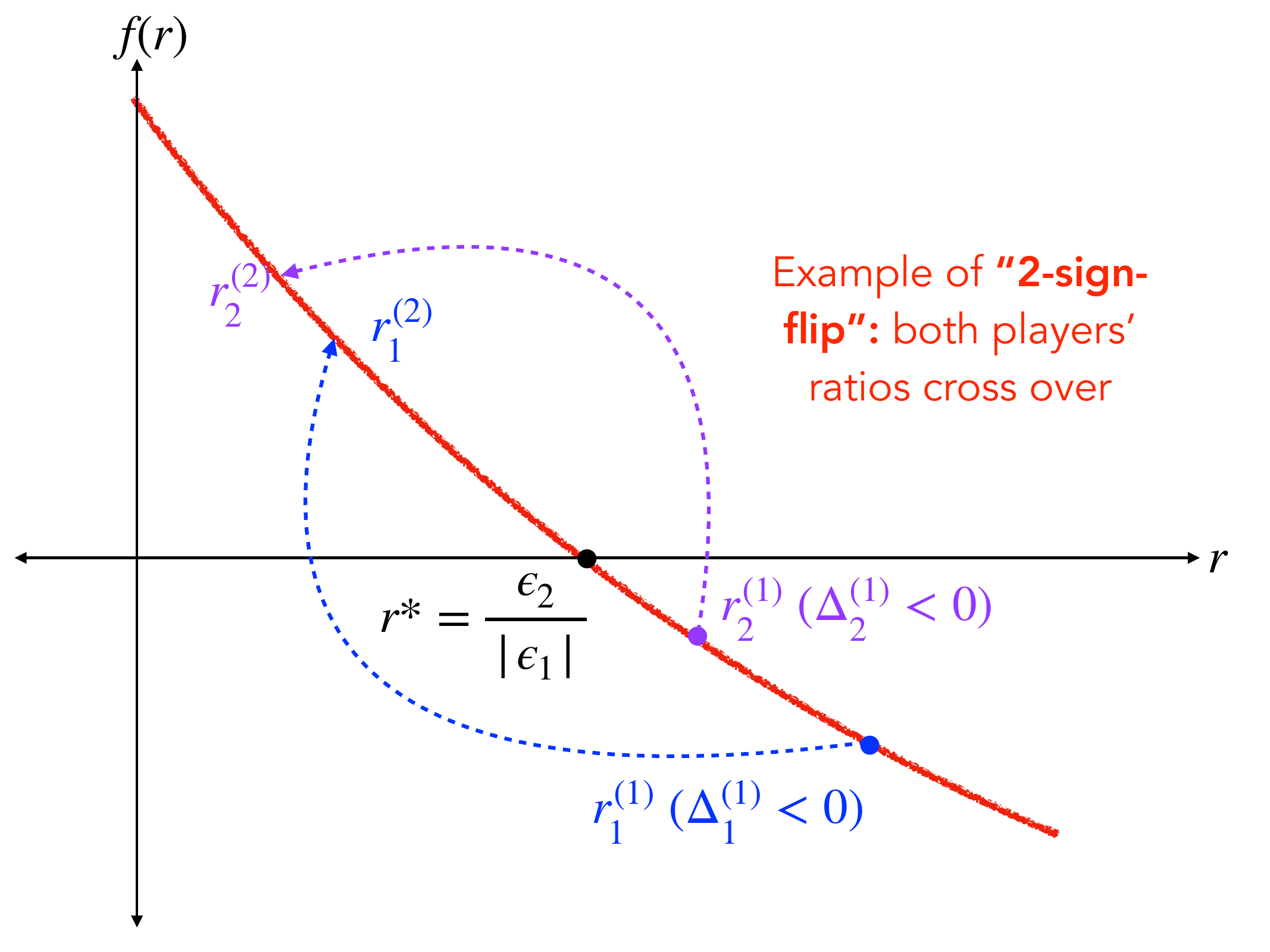}
            \caption{Same-sign initialization with a ``2-sign-flip": 
            \\happens finitely often.}\label{fig:2signflip}
        \end{subfigure}
        \begin{subfigure}{0.45\linewidth}
            \centering
            \includegraphics[width=\linewidth]{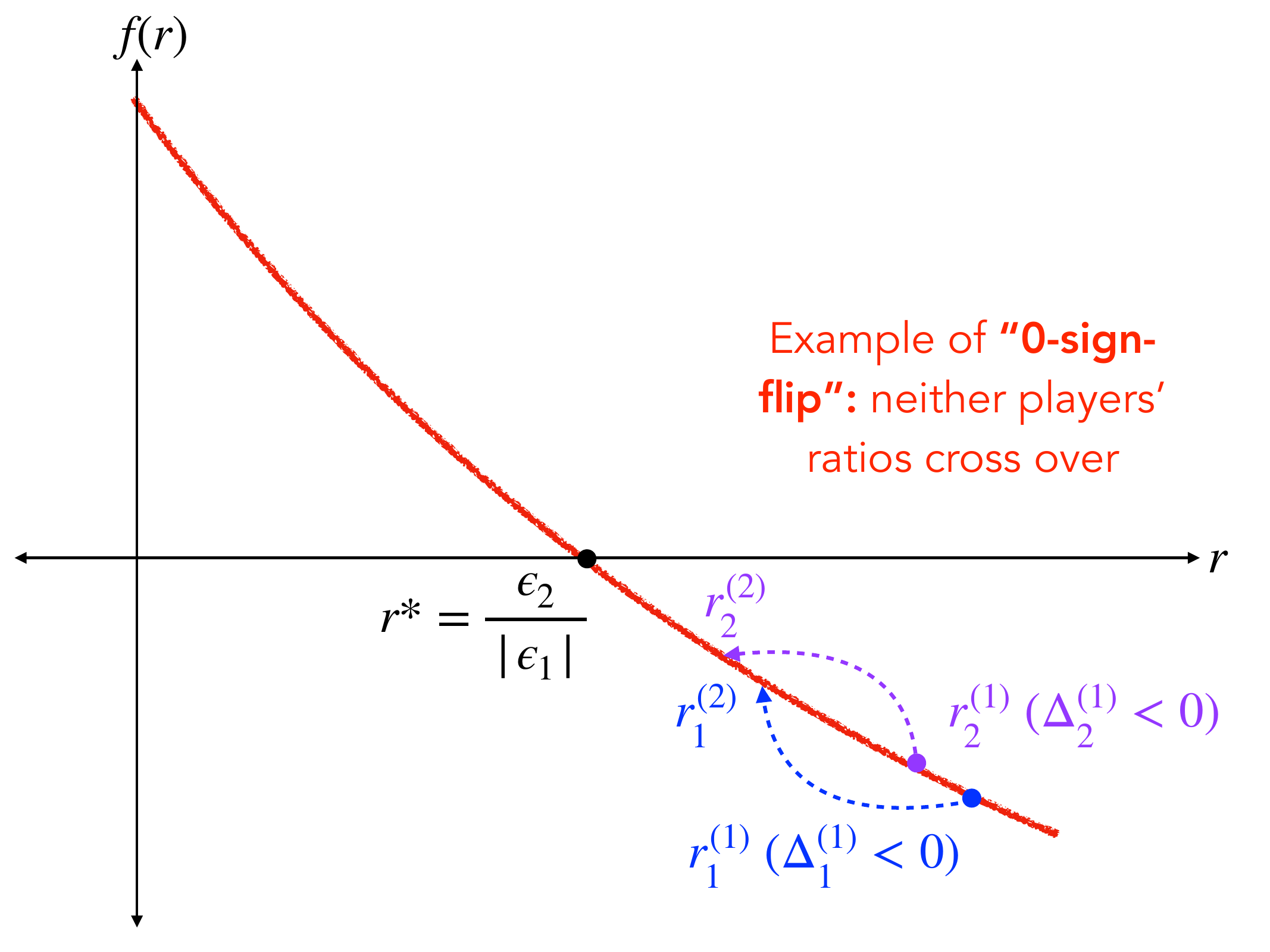}
            \caption{Same-sign initialization with a ``0-sign-flip": 
            \\happens finitely often.}\label{fig:0signflip}
        \end{subfigure}
    \end{minipage}
    }
    \caption{Illustrations of how the signs of ratio/initialization parameters $\Delta_i^{(1)}$, and their potential ``flips", affect analysis in the proof of Theorem~\ref{thm:e1l0e2g0}. $r^*$ is the unique root of $f(\cdot)$.}\label{fig:proofsketch}
    \vspace{-0.4cm}
\end{figure}

In this section, we provide some brief intuition for the rich variety of behavior that we observe across initialization and game environment.
An initial observation --- that simplifies the proof for all the cases --- is that we can parameterize any mixed strategy in the $2$-dimensional simplex through a single positive real number, i.e. the ratio $r_i := p_{i,1}/p_{i,2}$ for player $i$.
Since $p_{i,1} + p_{i,2} = 1$, $r_i$ uniquely identifies a mixed strategy $(p_{i,1},p_{i,2})$ for player $i$.
It is then easy to see that, for a given player $i\in\{1,2\}$ and the other player $j\in\{1,2\}$ (where $j\not=i$), the update for player $i$ can be written as a function of solely $(r_i,r_j)$, as
\[
r^{(t+1)}_i =r^{(t)}_i \cdot \exp(\eta f(r_j^{(t)})),
\]
where $f(r^{(t)}_j)=\frac{\epsilon_1r^{(t)}_j+\epsilon_2}{1+r^{(t)}_j}=\Delta_{j}^{(t)}.$
An initial question one might ask at this stage is what the \emph{fixed points} of the above mapping would be, i.e. when $\bar{r} = \bar{r} \cdot \exp(\eta f(\bar{r}))$.
It is easy to see that fixed points correspond either to $\bar{r} = 0$ (i.e. the pure strategy $\theta_1$), $\bar{r} = \infty$ (i.e. the pure strategy $\theta_2$) or, in the cases where the the signs of $\epsilon_1$ and $\epsilon_2$ do not match (Parts II and III) $\bar{r} = r^* = -\frac{\epsilon_2}{\epsilon_1}$ (i.e. the mixed strategy $p_{\text{SE}}$, which corresponds to the unique root of $f(r)$).
The name of the game is then to show that each of $\{r_i\}_{i=1}^2$ converges to one of the above fixed points, which in turn corresponds to either a pure/strictly mixed NE in each of the relevant cases.

\paragraph{Proof sketch of Theorem~\ref{thm:same-sign}} The cases where $\epsilon_1$ and $\epsilon_2$ have the same sign (Part I) are easy.
To see this, consider the case where $\epsilon_1, \epsilon_2 > 0$. 
In this case, $f(\cdot)$ is positive,  leading to monotonically increasing updates $\{r_i^{(t)}\}_{t \geq 1}$ in $t$ for each $i \in \{1,2\}$, that will converge exponentially fast (at a rate depending on the initial value $f(r_i^{(1)})$) to the pure NE $(\theta_1,\theta_1)$. (Similarly, if $\epsilon_1, \epsilon_2 < 0$, $f(\cdot)$ is instead negative, leading to exponential convergence to the pure NE $(\theta_2,\theta_2)$.)


\paragraph{Proof sketch of Theorem~\ref{thm:e1l0e2g0}:} The analysis of Parts II-IV is far trickier and heavily depends on the initialization of the players' strategies, controlled by $\Delta_1^{(1)}$ and $\Delta_2^{(1)}$. 
Due to space limitations we illustrate our analysis for the proof of Theorem~\ref{thm:e1l0e2g0}, noting that Parts III and IV utilize similar (but not identical) ideas.
A key initial observation is that $f(r)$ in Part II is monotonically \emph{decreasing}, meaning that the sign of $\text{sign}(\Delta_i^{(t)}) = - \text{sign}(r_i^{(t)} - r^*)$ for all $t \geq 1$ (see Figure~\ref{fig:proofsketch} for an illustration).
The analysis is then divided into three conditions:
\begin{enumerate}
\item \vspace{-2mm} Here as well, there is an ``easy" initialization in which $\Delta_1^{(1)}$ and $\Delta_2^{(1)}$ have opposite signs --- this type of initialization turns out to lead to similar monotonic behavior as in Part I.
To see this, consider the case where $\Delta_1^{(1)} < 0$ and $\Delta_2^{(1)} > 0$ (Condition 1 in Appendix~\ref{appendix"thm:e1l0e2g0}).
Because $r_1^{(t+1)} = r_1^{(t)} \exp(\eta \Delta_2^{(t)})$, the update leads to the ratio $r_1$ increasing which further \emph{decreases} $\Delta_1^{(t)}$.
Similarly, the update leads to the ratio $r_2$ \emph{decreasing} which further \emph{increases} $\Delta_2^{(t)}$.
As illustrated in Figure~\ref{fig:condition1sketch}, in this case the parameters $\Delta_i^{(t)}$ move further and further away from each other, eventually converging to the asymmetric pure NE $(\theta_1,\theta_2)$.
\item \vspace{-2mm} The more difficult case is when the signs of $\Delta_1^{(1)}$ and $\Delta_2^{(1)}$ \emph{match} at initialization.
In this case the above monotonicity argument no longer directly applies, due to the fact that the sign of $r_i^{(1)} - r^*$ for each player $i \in \{1,2\}$ may change unpredictably over time.
Clearly, if one of the players flips sign (what we call a ``1-sign-flip"), the analysis reduces to the easy initialization condition above; see Figure~\ref{fig:1signflip} for an illustration.
Remarkably, we are able to show that the other two bad events of a ``0-sign-flip", meaning the signs stay the same from round $t \to t+1$, and ``2-sign-flip", meaning the signs both change from round $t \to t + 1$, can only happen \emph{finitely many times}. (See Figures~\ref{fig:2signflip} and~\ref{fig:0signflip} for an illustration of these bad events). In fact, Lemma~\ref{lemmmm:noffilop} shows via a potential-based argument that a ``2-sign-flip" can happen at most a constant number of times (for an explicit constant that depends on the parameters of the game). 
\item \vspace{-2mm} Interestingly, the above analysis critically uses the fact that $r_1^{(1)} \neq r_2^{(2)}$, which only occurs at unequal initialization.
The case of identical initialization needs to be handled separately and turns out to correspond to analyzing a $1$-dimensional dynamical system on the variable $u^{(t)} := \ln \left(\frac{r^{(t)}}{r^*}\right)$ (where we dropped the index $i$ as the ratios are identical).
We show that the map $u \to T(u)$ is \emph{contractive} for the specified threshold on $\eta$ and appeal to Banach's fixed point theorem to show convergence to the unique fixed point, given by $\bar{u} = 0 \implies \bar{r} = r^*$, which corresponds to the unique strictly mixed NE.
\end{enumerate}

\vspace{-4mm}
\section{Applications to  multi-agent performative prediction}\label{sec:applications}

\label{sectionMPPP}

The recently proposed framework of Performative Prediction~\citep{perdomo2020performative} studies interactions in which machine learning models are trained on data that is not i.i.d. In particular, it focuses on scenarios where, at each time step, the observed data distribution depends on the previously deployed model or decision. This dependence may arise because individuals strategically adapt their behavior in response to the decision rule, or because the deployed policy itself shapes the underlying environment from which new data is generated.~\cite{narang2023multiplayer} extended Performative Prediction to multi-agent settings, where $n$ different agents simultaneously deploy machine learning models and decision rules. They showed that simple dynamics (specifically, retraining models at each time step and gradient descent) converge in a last-iterate sense to a stable solution. However, these results i) do not cover our exponential weights dynamic, and ii) the stable solutions obtained may not coincide with equilibria of the actual ensuing game without additional assumptions such as strong monotonicity (which does not hold for finite normal-form games).
\vspace{-2mm}
\paragraph{A special case of a lending game between two banks:} In this section, we formulate a new model for Performative Prediction in a special but practically motivated case: that of a game between two banks, or learners, which need to decide on application thresholds and interest rates for lendees. Here, the banks aim to approve loan applications from customers. Each bank \(i \in \{1, 2\}\) selects two parameters \( (\tau_i, \gamma_i):= \theta_i\), where $\tau_i$ is the minimum credit score threshold for giving a loan to a customer and $\gamma_i$ is the interest rate offered to approved customers.

Let \(y_j \in [0,1]\) be the normalized credit score of customer $i$, which is drawn from a distribution \(D_y\). A customer with a  credit score \(y_j\) is approved by Bank $i$ if only if \(y_j \geq \tau_i\). For simplicity, we assume that \(y_j\) directly represents the customer’s probability of repaying the loan. That is, letting $Y_j \in \{0,1\}$ be a random variable controlling whether the customer $j$ repays their loans, with $0$ meaning default and $1$ repayment, we have $y_j = \mathbb{P}[Y_j = 1]$. 
A customer's (normalized) credit score can be interpreted as a noisy observation of $y_j$. In this case, if the customer chooses Bank $i$ with parameter $\theta_i=(\tau_i,\gamma_i)$, the expected reward for this bank would be $(1+\gamma_i)\cdot y_j-(1-y_i)$: if the customer repays the loan, with probability $y_j$, the return is $1+\gamma_i$; if not, the return is 0.  

For given parameter choices \(\theta_1 = (\tau_1, \gamma_1)\) and \(\theta_2 = (\tau_2, \gamma_2)\), the allocation of customers and rewards for banks are determined by the following conditions:
\begin{enumerate}
    \item \textbf{Exclusive allocation by threshold}: If \(\tau_1 \leq y_j < \tau_2\), the customer goes to Bank 1, as the score qualifies for Bank 1 but not Bank 2. Conversely, if \(\tau_2 \leq y < \tau_1\), the customer chooses Bank 2.
    \item \textbf{Thresholds met by both banks}: If \(\tau_1, \tau_2 \leq y_j\) and \(\gamma_1 < \gamma_2\), the customer selects Bank 1 due to the lower interest rate. Conversely, if \(\gamma_1 > \gamma_2\), the customer chooses Bank 2. If there is a tie and \(\gamma_1 = \gamma_2\), the customer goes to each bank with probability $0.5$.
    \item \textbf{Rejection by both banks}: If \(y_j < \tau_1\) and \(y_j < \tau_2\), the customer is rejected by both banks.
\end{enumerate}

In particular, when considering two banks that learn their strategies over time, this is an exact example of performative prediction: the consumer data that each bank sees is a function of both of their deployed decision rules, since consumers strategically select which bank they get a loan from. 

\paragraph{A 2x2 bank game:} We focus on a simple case where there are two choices for each parameter: \(0 \leq \tau_{\ell} < \tau_{h} \leq 1\) and \(0 \leq \gamma_{\ell} < \gamma_{h} \leq 1\). We consider two possible actions per player: either $(\tau_l,\gamma_h)$, or $(\tau_h,\gamma_l)$; the rationale is that a bank with a lower threshold, i.e. which allows customers with worse credit, sets a higher rate with said customers\footnote{It is not hard to show that this is without loss of generality, as actions $(\tau_l,\gamma_l)$ and $(\tau_h,\gamma_h)$ are dominated under reasonable choices of parameters. See Appendix~\ref{sbusectionthebankgame} for a formal proof of this fact. }. Let 
\begin{equation}
\label{eqn:utility-onebank}
    h_{D_y}(\gamma,\tau_a,\tau_b)=\int_{\tau_a}^{\tau_b}[(2+\gamma)y-1]p(y)dy.
\end{equation}

It is not hard to see that a) this is a symmetric $2 \times 2$ game that fits into our framework, and moreover, b) the payoff matrix of player 1 is then given by Table~\ref{tab:my-table}, where the values of $a, b, c$ and $d$ are explicitly defined.
Therefore, Theorem 1 directly applies to this Bank Game and implies convergence to either pure symmetric NE, pure asymmetric NE or strictly mixed NE depending on the parameters of the Bank Game (i.e. the values of $\gamma_{\ell}, \gamma_h, \tau_{\ell}$ and $\tau_h$, as well as the distribution on customers $D_y$).
Our experiments section (Appendix~\ref{sec: exp-results}) provides simulations showing convergence under different values of $\epsilon_1$ and $\epsilon_2$.  
Overall, Theorem 1 implies rich consequences for lenders and lendees in this application: symmetric NE correspond to homogeneous bank behavior, but asymmetric pure NE correspond to heterogeneous behavior (e.g. one bank picking a high interest rate and low threshold, while the other bank does the opposite).
Our results show that either type of behavior can be realized depending on the parameters of the game and initialization conditions.

\begin{table}[t]
\centering
\scalebox{0.9}{
\begin{tabular}{c}\\
\begin{tabular}{|c|c|c|}
\Xhline{1.2pt}
\diagbox{\textbf{Bank 2}}{\textbf{Bank 1}} & \textbf{$\theta_1=(\tau_{\ell},\gamma_h)$} & \textbf{$\theta_2=(\tau_{h},\gamma_{\ell})$} \\
\Xhline{1.2pt}
\textbf{$\theta_1=(\tau_{\ell},\gamma_h)$} & $a=\frac{1}{2}h(\gamma_{h},\gamma_{\ell},1)$ & $b=h(\gamma_{\ell},\tau_h,1)$ \\
\hline
\textbf{$\theta_2=(\tau_{h},\gamma_{\ell})$} & $c=h(\gamma_{h},\tau_{\ell},\tau_{h})$ & $d=\frac{1}{2}h(\gamma_{\ell},\tau_{h},1)$ \\
\Xhline{1.2pt}
\end{tabular}
\end{tabular}
} 
\caption{Utility matrix of Bank 1, which exactly corresponds to the utility matrix of Player 1 in Table \ref{tab:utilities111}. Similar to Table \ref{tab:utilities111}, the utility matrix of Bank 2 is symmetric to that of Bank 1.}
\label{tab:my-table}
\vspace{-0.4cm}
\end{table}

\section{Related work}\label{sec:related-work}

We organize our discussion of related work under three verticals below.

\paragraph{Black-box results for ``no-regret" dynamics:} The design of algorithms that achieve vanishing external regret dates back to~\citep{hannan1957approximation,blackwell1956analog}, as well as the interpretation of such algorithms as solving repeated zero-sum games.
A modern non-asymptotic result by~\cite{freund1999adaptive} shows that the time-average of the payoff of two agents playing \emph{any} no-external-regret algorithms against each other will converge to the (unique) value of the zero-sum game. 
These results do not extend to arbitrary general-sum games; in fact, ``uncoupled" dynamics that only use first-order information can be shown to never converge to NE for certain games~\citep{hart2003uncoupled}, and ``higher-order" dynamics are instead needed~\citep{toonsi2023higher}.
Instead, a folklore result (see, e.g.~\citep{cesa2006prediction}) shows that the time-average of the joint strategy profile of the agents will converge to the polytope of \emph{coarse correlated equilibria} (CCE).
Such a result still allows for oscillation within the CCE polytope\footnote{It \emph{is} possible to achieve convergence to a specific equilibrium point within the CCE polytope for certain games if one player plays a no-regret algorithm and the other enacts a series of ``threats", or ``grim-trigger" strategies that penalize both players if they deviate from a specific equilibrium (see, e.g.~\citep{collina2023efficient}). Such dynamics are asymmetric and highly specialized (due to the threat-based nature of the strategies), and do not constitute the standard modus operandi in multi-agent learning.} (except in the special case where it is a singleton set, meaning that the set of all CCE coincides with the set of all NE).
When the players' algorithms achieve the stronger guarantee of vanishing \emph{internal/swap regret}~\citep{blum2007external,greenwald2003general} (see also the ``adaptive heuristics" defined in the economics literature~\citep{hart2005adaptive}, as well as the notion of calibrated learning~\citep{foster1997calibrated}), the time-averaged convergence is specialized to the polytope of Aumann's \emph{correlated equilibria} (CE)~\citep{aumann1987correlated}.
This result also allows for oscillations within the CE polytope.
Table~\ref{Table:NEs} showed that both the CE polytope and the CCE polytope are non-trivial in general even for $2 \times 2$ symmetric games.
Achieving convergence to a specific NE then necessitated a specialized study of a specific dynamic, which we chose to be exponential weights in this paper.
It is also worth noting that we analyzed exponential weights with a \emph{universally constant} step size, while the no-external-regret manifestation of exponential weights requires a step size that either uniformly decays with rounds $t$ or depends inversely on the horizon of learning $T$~\citep{cesa1997use} --- so, strictly speaking, our dynamic is not even no-external-regret.

Interestingly, the \emph{last iterate} of no-swap-regret dynamics was recently shown to converge to a specific symmetric Nash equilibrium for \emph{zero-sum symmetric games} when the players initialize identically~\citep{leme2024convergence}. 
This result is complementary to our work, which studies a specific dynamic in a different setting (still symmetric, but $2 \times 2$ and often not zero-sum).

\paragraph{Oscillation and chaos in last-iterate of game dynamics:} The \emph{last iterate} of game dynamics has seen intense recent study, motivated both by natural desiderata of behavioral stability~\citep{hofbauer1998evolutionary} and by applications in machine learning such as the training process of generative adversarial networks~\citep{daskalakis2018training}.
Results on the last-iterate, particularly for the standard exponential weights dynamic, are predominantly negative, frequently showing oscillations or some sort of chaotic behavior.
Even for zero-sum games,~\cite{bailey2018multiplicative,cheung2018multiplicative} showed that multiplicative weights with a constant step size oscillates towards the boundaries of pure strategies when the game only has strictly mixed NEs.~\cite{cheung2019vortices} further showed the advent of chaotic behavior for all algorithms within the ``Follow-the-Regularized-Leader" (FTRL) family (which includes exponential weights) for a broad class of games, including graphical constant-sum games, certain symmetric games and arbitrary $2 \times 2$ general-sum (not necessarily symmetric) games.
~\cite{palaiopanos2017multiplicative} also discovers a difference between the exponential weights dynamic and the ``linear" version of the multiplicative weights update, with the former leading to chaotic behavior with large enough step sizes, in congestion games.
The papers that provide a general framework for oscillations/chaos conduct a volume analysis showing that if the volume of the ``flow" of the dynamics (from some initialization set of non-zero Lebesgue measure) increases, it implies chaotic behavior --- but the converse does not hold.
Negative results exist even for the weaker requirement of \emph{local asymptotic stability} near any NE --- in particular, any mixed NE is shown to be unstable under the family of FTRL dynamics both in continuous and discrete time~\citep{vlatakis2020no,giannou2021survival}.
As we discussed in Sections~\ref{Section:convergecnce} and~\ref{sec:applications}, our positive convergence results for $2 \times 2$ symmetric games can be readily reconciled with these results.
We showed that the mixed NE is only approached through specific directions (involving identical initialization), which turns out to be the only stable direction of deviation. 
Moreover, many of the oscillatory situations (e.g. $2 \times 2$ congestion games~\citep{palaiopanos2017multiplicative} turned out to only occur under a sufficiently large step size.

\paragraph{Positive convergence results for specific dynamics:} 
A plethora of last-iterate convergence results have been shown for zero-sum games (most relevant to our setting are the ones for finite games, i.e. min-max optimization constrained to the probability simplices~\citep{daskalakis2019last,cai2022finite}) under modifications of exponential weights to include optimism and/or extragradient method, as well as a sufficiently small or decaying step size.
In contrast, we allow for non-zero-sum games but require symmetry and the game to be $2 \times 2$, and we consider the original exponential weights dynamic; moreover, in many cases we can handle an arbitrarily large step size.
Positive \emph{average-iterate} convergence results have also been developed for the fictitious-play dynamic~\citep{brown1951iterative,robinson1951iterative} (which corresponds to ``Follow-the-Leader in online learning) for zero-sum games with a diagonal payoff matrix~\citep{abernethy2021fast} as well as zero-sum symmetric games of a generalized Rock-Paper-Scissors variety~\citep{lazarsfeld2025fast}. 
It is an intriguing open question whether fictitious play dynamics would converge in our setting, whether in the average or last iterate. Specific to congestion games (with any finite number of players/pure strategies),~\cite{kleinberg2009multiplicative} showed that \emph{randomized} variants of exponential weights can converge \emph{with high probability} for a sufficiently small step size. {\cite{chotibut2021family} established the last-iterate convergence of the multiplicative weights algorithm for $2 \times 2$ congestion games. We recover their conclusion in
this special setting as a consequence of Theorem \ref{thm:e1l0e2g0}.}

A recent yet relatively extensive line of work studies performative prediction~(\cite{perdomo2020performative,miller2021outside,brown2022performative,bechavod2021gaming,mendler2022anticipating,jagadeesan2022regret,zrnic2021leads,shavit2020causal} to only name a few), which looks at single-agent machine learning settings where the training data on a given round depends on the previously deployed model(s)---capturing situations where deployed models induce strategic behavior or distribution shifts. The case of performative prediction in \emph{multi-agent settings}, however, has received significantly less attention. In particular,~\cite{narang2023multiplayer} study convergence of learning dynamics in specially structured games induced by these machine-learning-oriented problems. In these works, multiple players iteratively train their own machine learning models, and each player’s data distribution---and hence their utility---depends on the set of models deployed on a given round. This literature shows that both a form of \emph{repeated empirical risk minimization} as well as gradient descent (when used by all players in the game) can converge in a last-iterate sense to a stable outcome. However, two limitations remain: (i) said dynamics can converge to a ``sub-optimal'' point, which is not an equilibrium of the underlying game between learners, and (ii) convergence guarantees rely on strong assumptions both on the players' loss functions --- especially strong monotonicity, which does not hold for normal-form games, and on the sensitivity of outcomes to players’ actions.

The works~\citep{vlatakis2020no,giannou2021survival} show that for any general-sum game, a particular NE is locally asymptotically stable under FTRL dynamics if and only if it is pure. Note that these results only imply \emph{local convergence} of the dynamic, i.e. if it is initialized sufficiently close to a pure Nash Equilibrium. On the other hand, we utilize the special structure of $2 \times 2$ symmetric games to provide a \emph{global convergence} result, i.e. for a broader set of initializations that could be arbitrarily far from any NE.

\section{Discussion}\label{sec:discussion}

Overall, our results paint a more positive than expected picture for the exponential weights dynamic in the special case of $2 \times 2$ symmetric games.
The most pressing open question is whether this story extends in any way to symmetric games with more than $2$ actions.
Preliminary simulations suggest convergence behavior in many (but not all) cases, but characterizing the ensuing dynamical systems is an important future direction.
Even in the symmetric $2 \times 2$ case, interesting open questions remain --- such as whether we can prove fast non-asymptotic rates of convergence for all the cases in Theorem 1 (not just the easy ones), and whether we can handle unequal step sizes for different players.

\paragraph{Acknowledgment}{ GW was supported by the Apple Scholars in AI/ML PhD fellowship by Apple and ARC-ACO fellowship provided by Georgia Tech. VM
was supported by the NSF (through award CCF-2239151 and award IIS-2212182), an Adobe
Data Science Research Award and an Amazon Research Award. KA and JZ
were supported by the NSF (through award IIS-2504990 and award IIS-2336236).}

\bibliographystyle{unsrt}
\bibliography{name}

\appendix
\newpage

\section{Notation and Preliminaries for proofs}\label{sec:notation-proofs}
In this section, we introduce some notation that will help simplify the exposition of the proof of Theorem 1. In all that follows, recall that we defined $\epsilon_1 := a - b$ and $\epsilon_2 := c - d$. Specifically, $\epsilon_1$ denotes the change in utility of player $1$ when their strategy changes from $\theta_1$ to $\theta_2$ when player $2$ plays $\theta_1$, and $\epsilon_2$ denotes the change in utility of player $1$ when their strategy changes from $\theta_1$ to $\theta_2$ when player $2$ plays $\theta_2$. Also, for each player $i\in\{1,2\}$, recall that we defined the functional
\begin{equation}
    \label{eqn:defnofdelta}
    \Delta^{(t)}_{i} = p^{(t)}_{i,1}\epsilon_1+ p^{(t)}_{i,2}\epsilon_2.
\end{equation} 
We now introduce our new notation below:
\begin{itemize}
\item For each player $i \in \{1,2\}$, let $r_i^{(t)}=\frac{p^{(t+1)}_{i,1}}{p^{(t+1)}_{i,2}}$ be the ratio of probabilities allocated to pure strategy $\theta_1$ over pure strategy $\theta_2$. 
Then, for any pair of players $i \neq j \in \{1,2\}$, we have
\begin{equation}\label{eq:rt-deltat-update}
r_i^{(t+1)} = \frac{p^{(t)}_{i,1}}{p^{(t)}_{i,2}}\exp\left(\eta (p^{(t)}_{j,1}a + p^{(t)}_{j,2}c - p^{(t)}_{j,1}b - p^{(t)}_{j,2}d)  \right) =r_i^{(t)}\exp(\eta\Delta_j^{(t)}).
\end{equation}

\noindent Note that, since $p^{(t)}_{i,1}+p^{(t)}_{i,2}=1$ and $p^{(t)}_{i,1}=p^{(t)}_{i,2}r^{(t)}_i$, we have $p^{(t)}_{i,2}=\frac{1}{1+r^{(t)}_i}$, and $p^{(t)}_{i,1}=\frac{r_i^{(t)}}{1+r^{(t)}_i}$.
Clearly, the mapping from $r \to p$ is a bijection.
\item Recall that we defined the function $f(r) := \frac{\epsilon_1r+\epsilon_2}{1+r}$ in Section~\ref{sec:proof-sketch}. The function $f(r)$ allows us to rewrite our update step in the following simplified form, solely in terms of $r$.
Namely, for any pair of players $i \neq j \in\{1,2\}$, we have
\begin{equation}\label{eq:rt-update}
r^{(t+1)}_i =r^{(t)}_i\exp(\eta f(r^{(t)}_j)),
\end{equation}
noting that $f(r_j^{(t)})=\frac{\epsilon_1r_j^{(t)}+\epsilon_2}{1+r_j^{(t)}}= p_{j,1}^{(t)} \epsilon_1  + p_{j,2}^{(t)} \epsilon_2  = \Delta_j^{(t)}$. Further, it is easy to verify that $f'(r)=\frac{\epsilon_1-\epsilon_2}{(1+r)^2}$. 
\item Finally, let $r^*=\frac{|\epsilon_2|}{|\epsilon_1|}$. If $\epsilon_1\cdot\epsilon_2<0$ (which is the case for Parts II and III of Theorem 1), then $f(r^*)=0$. 
\end{itemize}

\section{Proof of Lemma \ref{lem:::verysimpleobsercation}}
\label{append:proofoflemma111sdewq}

In this section, we prove Lemma~\ref{lem:::verysimpleobsercation}.
We first consider the first bullet point of Lemma~\ref{lem:::verysimpleobsercation}, i.e. what happens when $\Delta_1^{(1)} = \Delta_2^{(1)} = 0$.
From Equation~\eqref{eq:rt-deltat-update}, it is easy to see that the ratio $r^{(t)}_i$ of each player $i \in \{1,2\}$ does not change for all $t \geq 1$. Since the map $r \to p$ is a bijection, the mixed strategies also do not change.\\ 

Next, we consider the second bullet point of Lemma~\ref{lem:::verysimpleobsercation}, i.e. the case when $\Delta_1^{(1)} = 0$ but $\Delta_2^{(1)} \neq 0$.
Since $\Delta_1^{(1)}=0$, Equation~\eqref{eq:rt-deltat-update} tells us that $r^{(2)}_2=r^{(1)}_2$, meaning that $\Delta_2^{(2)}=f(r^{(2)}_2)=f(r^{(1)}_2)=\Delta^{(1)}_{2}\not=0$. Moreover, if $\Delta_1^{(1)}=0$, the parameters of the game must satisfy $\epsilon_1\not=\epsilon_2$. (This is because, if we instead had $\epsilon_1=\epsilon_2=\epsilon$, then we would obtain $f(r)=\epsilon$ for all values of $r$, meaning that $0=\Delta_{1}^{(1)}=f(r_{1}^{(1)})=\epsilon$, leading to $\epsilon_1=\epsilon_2=0$. This is a contradiction with the assumption in the lemma that $|\epsilon_1| + |\epsilon_2| > 0$.) 

Because $\epsilon_1 \neq \epsilon_2$, the sign of $f'(r)$ is solely determined by the sign of $\epsilon_1 - \epsilon_2$ for all values of $r$, implying that the function $f$ is strictly monotone.
On the other hand, since $\Delta_2^{(1)}\not=0$, Equation~\eqref{eq:rt-deltat-update} again implies that $r^{(2)}_1\not=r^{(1)}_1$. Since 
$f$ is strictly monotone and $\Delta_{1}^{(1)}=f(r^{(1)}_1)=0$,  we have $\Delta_{1}^{(2)}=f(r^{(2)}_1)\not=f(r^{(1)}_1)=0$.
This ultimately yields $\Delta_1^{(2)} \neq 0$, which completes the proof of the lemma.

\section{Characterization of correlated equilibria (CEs)}
\label{section:Characterization of CE}

In this section, we first review the definition of correlated equilibrium (CE), and then  characterize the set of CEs for our game. 
Formally, we define $\nu = (\nu_{1,1},\nu_{1,2},\nu_{2,1},\nu_{2,2}) \in \Delta_4$ as a joint probability distribution over pairs of strategies $\{(\theta_m,\theta_n): (m,n) \in \{1,2\} \times \{1,2\}\}$.

\begin{definition}[Correlated Equilibrium (CE)] 
\label{defn:CE}
A joint probability distribution $\nu$ \text{over} pairs of strategies $\{(\theta_p,\theta_q): (p,q) \in \{1,2\} \times \{1,2\}\}$ is a correlated equilibrium iff the following inequalities hold:
\[
\sum_{(m,n) \in \{1,2\} \times \{1,2\}} \nu_{m,n} \, u_1(\theta_m,\theta_n)
\geq 
\sum_{(m,n) \in \{1,2\} \times \{1,2\}} \nu_{m,n} u_1(\theta_{m'},\theta_n),
~~\forall m' \in \{1,2\},
\]
\[
\sum_{(m,n) \in \{1,2\} \times \{1,2\}} \nu_{m,n} \, u_2(\theta_m,\theta_n)
\geq 
\sum_{(m,n) \in \{1,2\} \times \{1,2\}} \nu_{m,n} \, u_2(\theta_m,\theta_{n'}),
~~\forall n' \in \{1,2\}.
\]
\end{definition}

In general, the set of CEs is a polytope (in our case, a subset of $\Delta_4$) and is bigger than the set of NEs, which may be disjoint and non-convex. 
We summarized the landscape of CEs along with the landscape of NEs in Table~\ref{Table:NEs}; we provide the calculation of the CEs here for completeness.

Based on Definition \ref{defn:CE}, the condition for a strategy profile to be a CE can be written as:
\[
\nu_{1,1} \epsilon_1\geq \nu_{1,2} (-\epsilon_2), \ \ \nu_{2,1} (-\epsilon_1)\geq \nu_{2,2} \epsilon_2,\ \ \nu_{1,1} \epsilon_1\geq \nu_{2,1} (-\epsilon_2),\ \ \nu_{1,2} (-\epsilon_1)\geq \nu_{2,2} \epsilon_2.\]
When $\epsilon_1<0$ and $\epsilon_2>0$, this condition reduces to
\begin{equation}
\label{eqn:CEl0g0xx}
    \max\left\{\frac{|\epsilon_1|}{\epsilon_2} \nu_{1,1}, \frac{\epsilon_2}{|\epsilon_1|} \nu_{2,2}\right\}\leq \min\{\nu_{1,2}, \nu_{2,1}\}.
\end{equation}
Similarly, when $\epsilon_1>0$ and $\epsilon_2<0$, the condition reduces to
\begin{equation}
\label{eqn:CEg0l0xx}
    \max\left\{\frac{|\epsilon_1|}{\epsilon_2} \nu_{2,1}, \frac{\epsilon_2}{|\epsilon_1|} \nu_{1,2} \right\}\leq \min\{\nu_{1,1},\nu_{2,2}\}.
\end{equation}
When $\epsilon_1=0$ and $\epsilon_2<0$, we have 
\begin{equation}
\label{eqn:2133121}
    \nu_{1,2}=\nu_{2,1}=0.
\end{equation}
When $\epsilon_1=0$, $\epsilon_2>0$, we have 
\begin{equation}
\label{eqn:ddfdgjndauih}
    \nu_{2,2}=0.
\end{equation}
We note that the case with $\epsilon_2 = 0$ and $\epsilon_1 \neq 0$ is exactly symmetric to the case with $\epsilon_1 = 0$ and $\epsilon_2 \neq 0$, simply by switching $\theta_1$ and $\theta_2$. For simplicity, we omit this case. 

Observe that in all of the instances, excluding the ones where $\text{sign}(\epsilon_1) = \text{sign}(\epsilon_2)$, the set of CEs forms a non-trivial polytope, and there exist infinitely many points that are CE but not mixed or pure NEs.
This means that even average-iterate convergence to a specific NE would \emph{not} follow by directly plugging in regret bounds, let alone last-iterate convergence.

\section{Proof of Theorem \ref{thm:same-sign}} 
\label{appdnex:profoofpar1}

\begin{figure}
  \centering

  \begin{subfigure}{0.45\textwidth}
    \includegraphics[width=\linewidth]{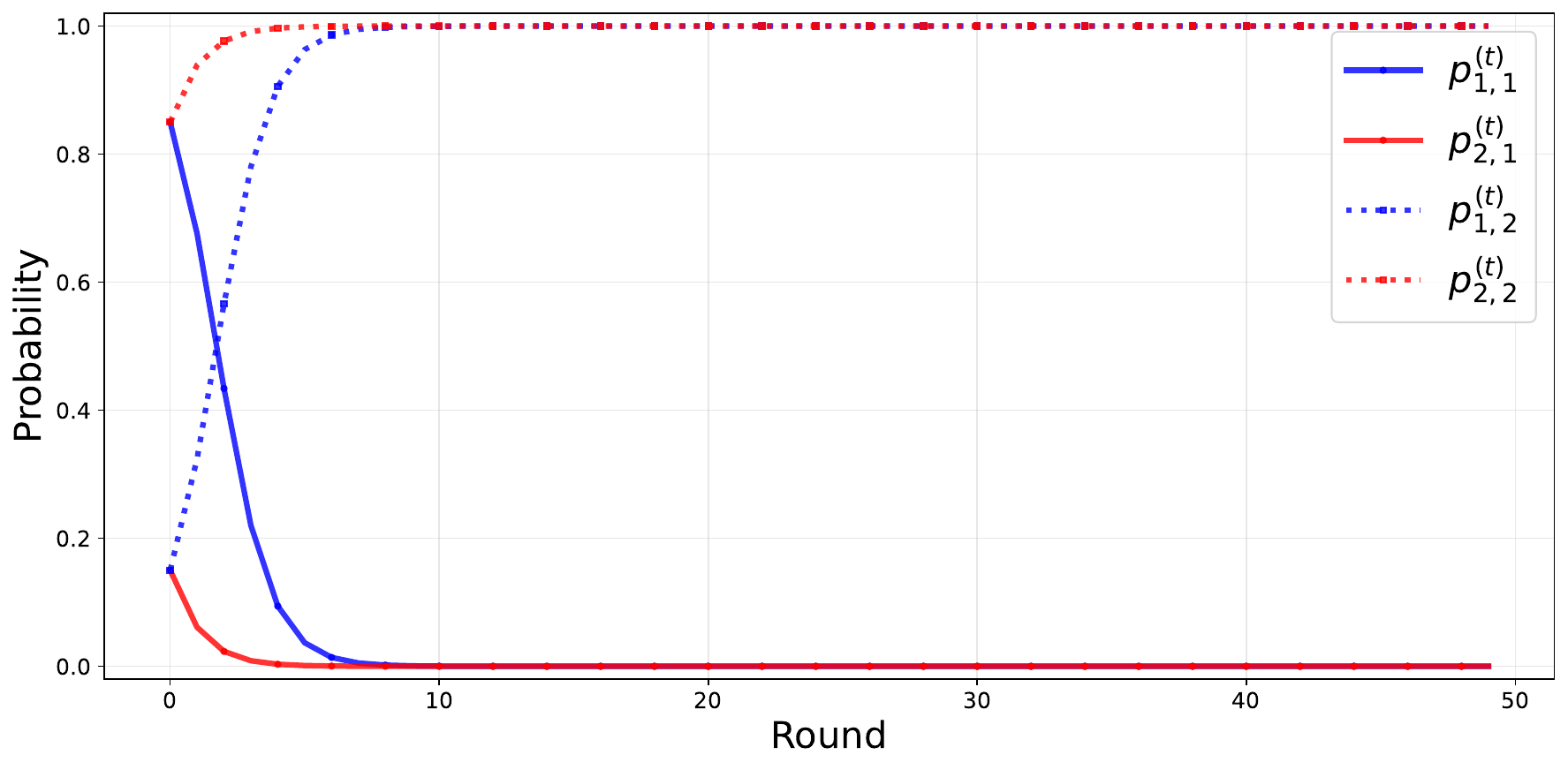}
    \caption*{r1.  Opposite-sign Initialization}
  \end{subfigure}
  \begin{subfigure}{0.45\textwidth}
    \includegraphics[width=\linewidth]{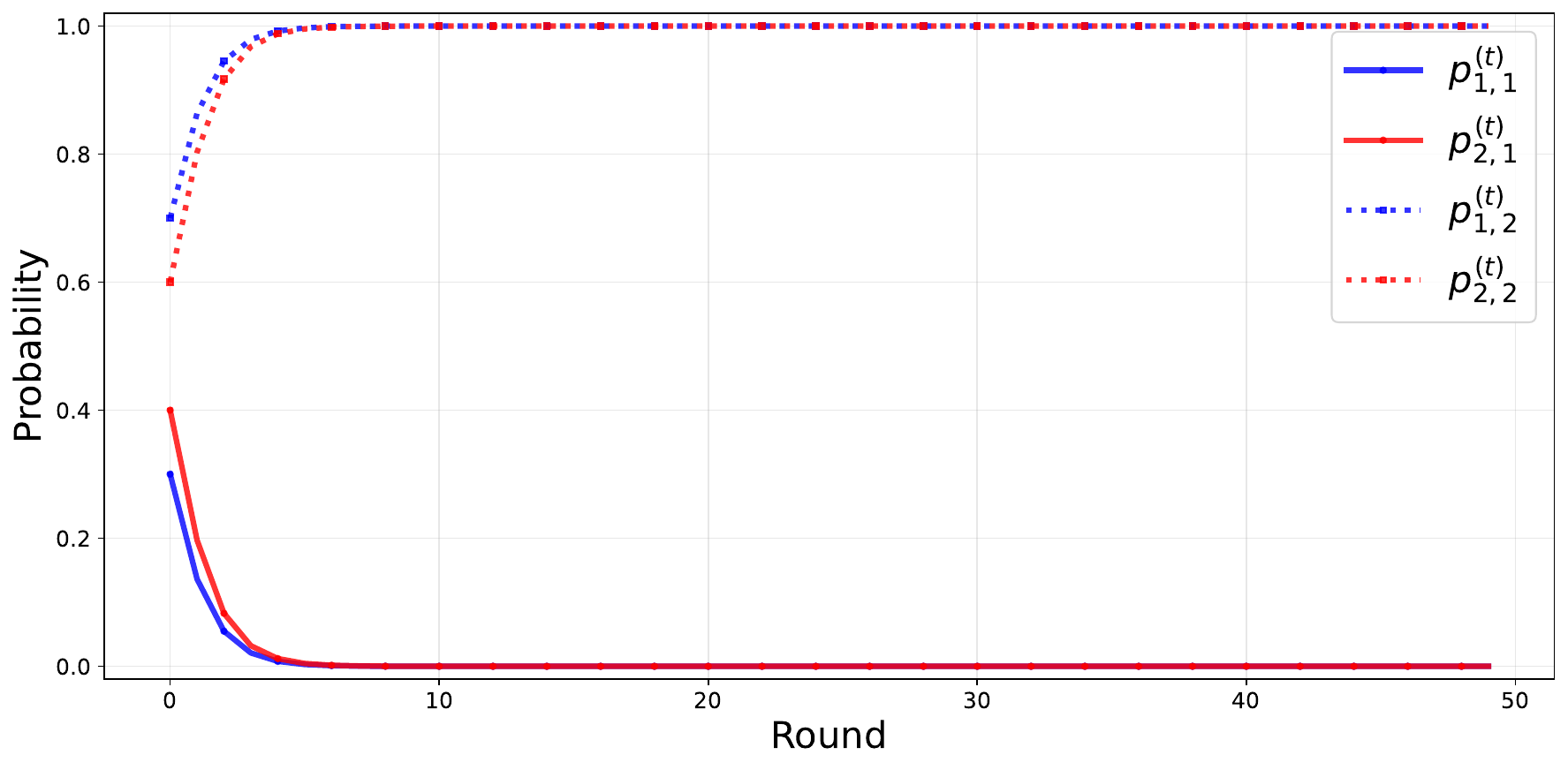}
    \caption*{r1. Same-sign Initialization}
  \end{subfigure}

  \begin{subfigure}{0.45\textwidth}
    \includegraphics[width=\linewidth]{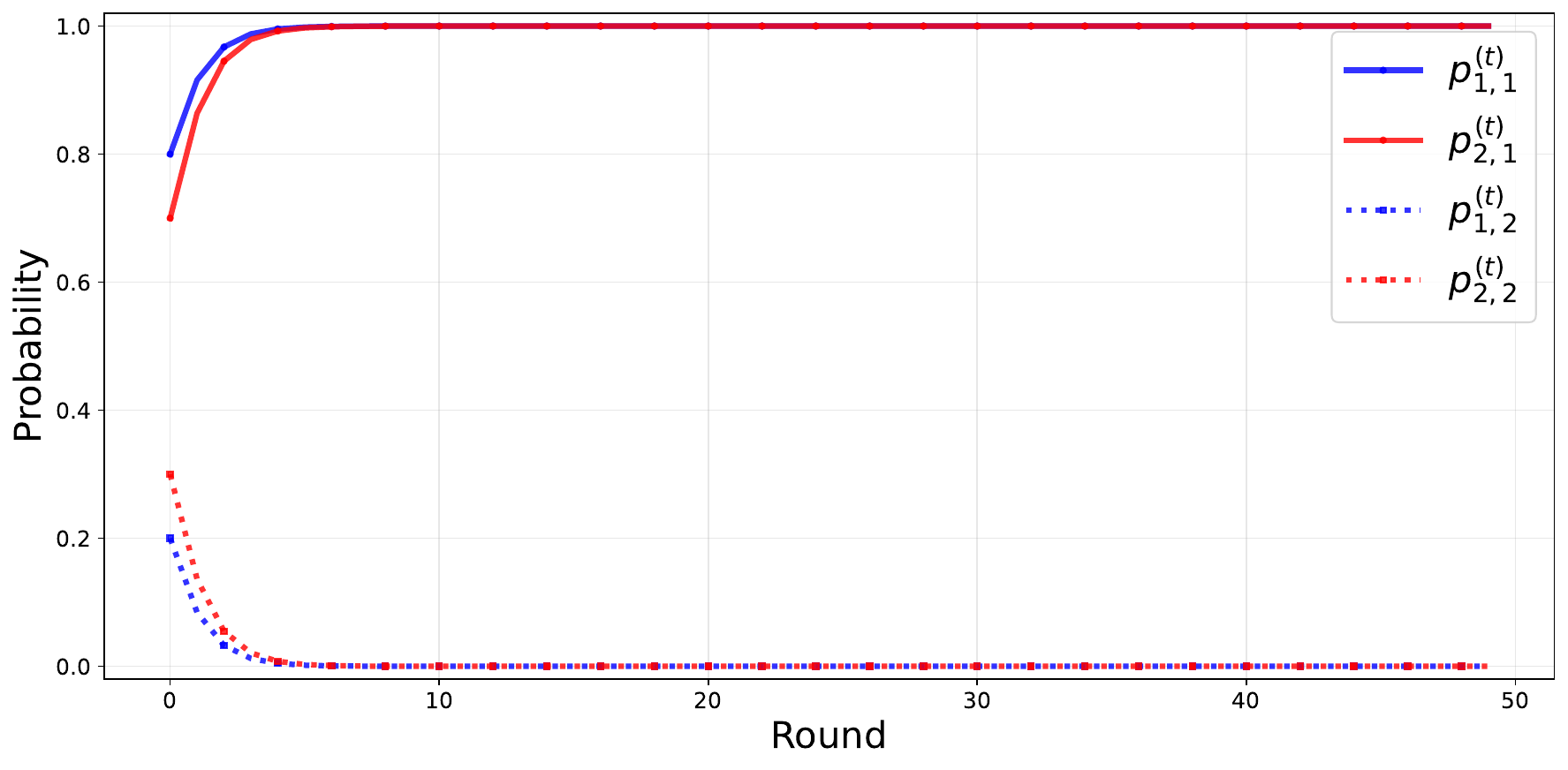}
    \caption*{r2. Same-sign Initialization}
  \end{subfigure}
  \begin{subfigure}{0.45\textwidth}
    \includegraphics[width=\linewidth]{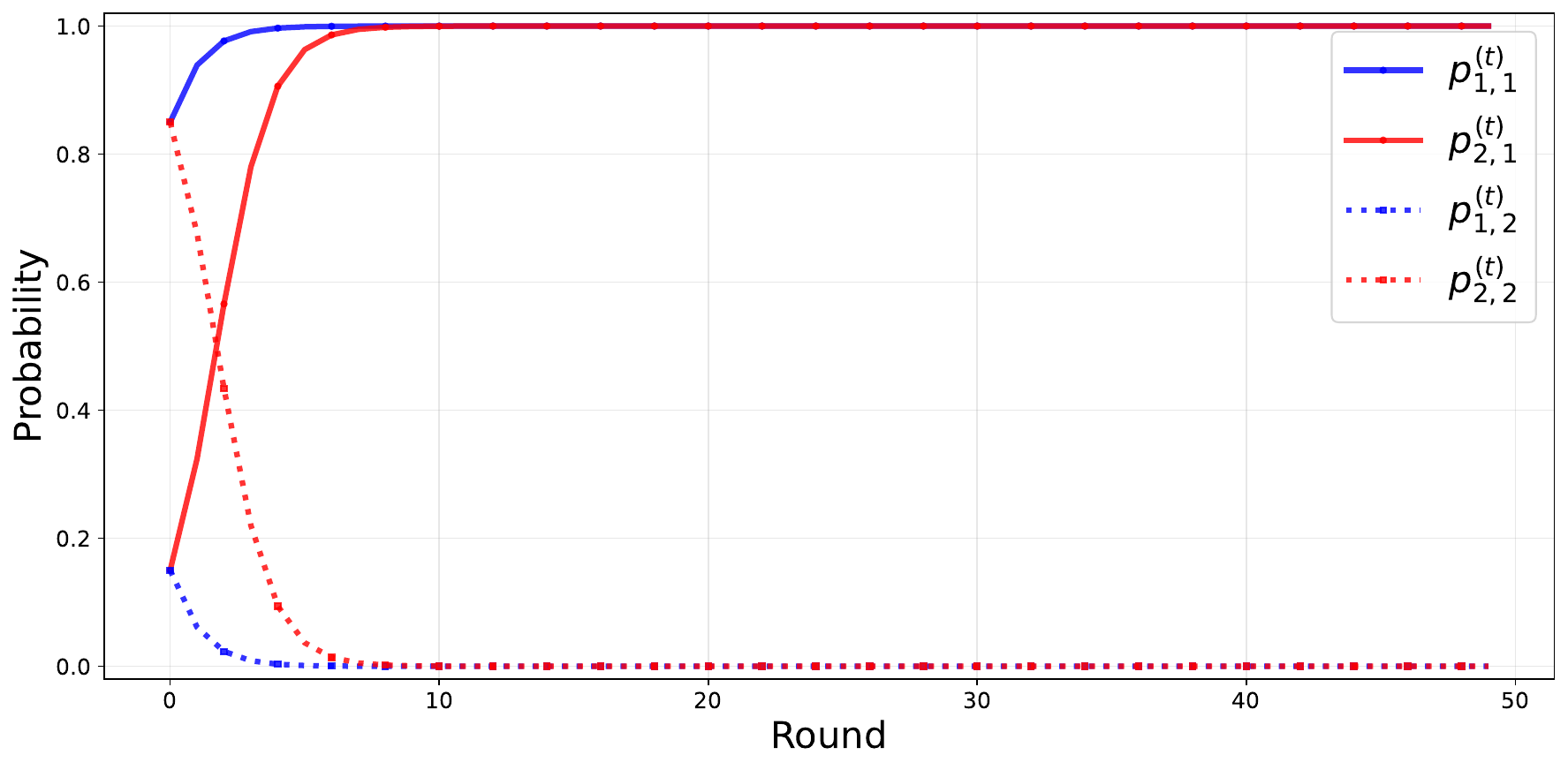}
    \caption*{r2. Opposite-sign initialization}
  \end{subfigure}

  \caption{Simulation results for Theorem \ref{thm:same-sign}. The row in Table~\ref{Table:NEconvergence} to which each case corresponds is marked in the sub-captions.}
  \vspace{-0.4cm}
  \label{figure:part 1}
\end{figure}

In this section, we prove Theorem~\ref{thm:same-sign}.
Without loss of generality, we consider the case where $\epsilon_1,\epsilon_2<0$. (Note that the other case is symmetric: it suffices to switch $\theta^{(1)}$ and $\theta^{(2)}$ to obtain the result.) Moreover, since $\epsilon_1,\epsilon_2<0$, Equation~\eqref{eqn:defnofdelta}, tells us that for each player $i\in\{1,2\}$,  we have $\Delta^{(t)}_i < 0$ for all $t$. This implies that the sequence $\{r^{(t)}_i\}_{t=1}^{\infty}$ is strictly decreasing. We first provide a formal convergence result for this case.
\begin{theorem}\label{thm:partone-formal}
Suppose Assumption \ref{as:nondegenerate} holds, and the game satisfies $\epsilon_1,\epsilon_2<0$. Then, Algorithm \ref{alg:Hedge} with any step size $\eta>0$ satisfies the following:
\begin{itemize}
    \item If $|\epsilon_1|<|\epsilon_2|$, then for each player pair $i \neq j \in\{1,2\}$, we have $p_{i,1}^{(t+1)}\leq r^{(1)}_i\exp(\eta t\Delta_{j}^{(1)})$, where $\Delta_{j}^{(1)}<0$. This implies that the algorithm converges to the symmetric pure NE $(\theta_2,\theta_2)$ exponentially fast;
    \item If $|\epsilon_1|\geq |\epsilon_2|$, then for each player pair $i \neq j \in\{1,2\}$, we have $p_{i,1}^{(t+1)}\leq r^{(1)}_i\exp(\eta t\epsilon_2)$, where $\epsilon_2<0$. This implies that the algorithm converges to the symmetric pure NE $(\theta_2,\theta_2)$ exponentially fast.
\end{itemize}
\end{theorem}
We now provide the proof of Theorem~\ref{thm:partone-formal}, considering each condition listed above. 

\paragraph{Condition  1: $|\epsilon_2|>|\epsilon_1|$.} Under this condition, we have $f'(r)>0$, meaning that $f(\cdot)$ is strictly increasing. Since $\{r^{(t)}_i\}_{t=1}^{\infty}$ is
decreasing, we know $r^{(t)}_i < r^{(1)}_i$, so $f(r^{(t)}_i) < f(r^{(1)}_i)$, and as a consequence, we have
$$r^{(t+1)}_i=r_i^{(1)}\exp\left(\eta \sum_{k=1}^tf\left(r_j^{(k)}\right)\right)\leq r_i^{(1)}\exp(\eta t\Delta_{j}^{(1)}).$$
Therefore, we have $p^{(t+1)}_{i,1}\leq p^{(t+1)}_{i,2}r^{(1)}_i\exp(\eta t\Delta_{j}^{(1)})\leq r^{(1)}_i\exp(\eta t\Delta_{j}^{(1)})$. 
Note that $\Delta_{j}^{(1)}<0$ is a strictly negative constant that depends on the algorithm's initialization. 
In conclusion, $p^{(t+1)}_{i,1}\leq r^{(1)}_i\exp(\eta t\Delta_{j}^{(1)})$, while $p^{(t+1)}_{i,2}\geq 1-r^{(1)}_i\exp(\eta t\Delta_{j}^{(1)})$.
Therefore, the algorithm converges to the symmetric pure NE $(\theta_2,\theta_2).$

\paragraph{Condition II: $|\epsilon_2|\leq|\epsilon_1|$.} In this case, $f'(r)\leq 0$, so $f$ is non-increasing.  Since $r \geq 0$ (i.e. the ratio is always non-negative),  we have $f(r)\leq f(0)=\epsilon_2<0$.
So 
$$r^{(t+1)}_i=r_i^{(1)}\exp\left(\eta \sum_{k=1}^tf(r_j^{(k)})\right)\leq r_i^{(1)}\exp(\eta t\epsilon_2),$$
which gives us $p^{(t+1)}_{i,1}\leq p^{(t+1)}_{i,2}r^{(1)}_i \exp(\eta t\epsilon_2)\leq r^{(1)}_i\exp(\eta t\epsilon_2)$. 
Note that $\epsilon_2<0$ is a strictly negative constant. 
In conclusion, we have $p^{(t+1)}_{i,1}\leq r^{(1)}_i\exp(\eta t\epsilon_2)$, while $p^{(t+1)}_{i,2}\geq 1-r^{(1)}_i\exp(\eta t\epsilon_2)$.
Therefore, the algorithm again converges to the symmetric pure NE $(\theta_2,\theta_2).$ 

\section{Proof of Theorem \ref{thm:e1l0e2g0}}
\label{appendix"thm:e1l0e2g0}

In this section, we prove Theorem~\ref{thm:e1l0e2g0}.
In this case, we have $\epsilon_1<0$ and $\epsilon_2>0$. 
This implies that $f'(r)<0$, i.e. $f$ is monotonically decreasing. Moreover, we define $r^*=\frac{\epsilon_2}{|\epsilon_1|}$ to be the unique root of the function $f(\cdot)$. We consider different initialization conditions one by one.

\subsection{Exponential Convergence With Opposite-Sign Initialization}
We first provide a formal description of exponential convergence in the ``easy" cases where the signs of the initialization functionals $\Delta_1^{(1)}$ and $\Delta_2^{(1)}$ are opposite.
\begin{theorem}\label{thm:parttwo-easycase-formal}
  Suppose Assumption \ref{as:nondegenerate} holds, and $\epsilon_1<0,\epsilon_2>0$. Then, Algorithm \ref{alg:Hedge} with any step size $\eta > 0$ satisfies the following:
  \begin{itemize}
      \item If $\Delta_1^{(1)}<0$, $\Delta_2^{(1)}>0$, then we have $p^{(t)}_{1,2}\leq \frac{1}{r_1^{(1)}}\exp(-\eta t\Delta_2^{(1)})$,
while 
$p^{(t)}_{2,1}\leq r^{(1)}_2\exp(\eta t\Delta_1^{(1)})$. This implies that the algorithm converges to the asymmetric pure NE $(\theta_1,\theta_2)$ exponentially fast;
\item If $\Delta_1^{(1)}>0,$, $\Delta_2^{(1)}<0$, then we have $p_{1,1}^{(t)}\leq r_{1}^{(1)}\exp(\eta t\Delta_2^{(1)}),$ while $p_{2,2}^{(t)}\leq \frac{1}{r_{2}^{(1)}}\exp(-\eta t\Delta_1^{(1)})$. This implies that the algorithm converges to the asymmetric pure NE $(\theta_2,\theta_1)$  exponentially fast.
  \end{itemize}
\end{theorem}
We now provide the proof of Theorem~\ref{thm:parttwo-easycase-formal}, considering each condition listed above.

\paragraph{Condition 1: $\Delta_1^{(1)}<0$, $\Delta_2^{(1)}>0$ (Basic Case 1).} In this case, it can be easily shown by induction that the ratio of player $1$, $\{r_1^{(t)}\}_{t=1}^{\infty}$, is monotonically increasing, while the ratio of player $2$, $\{r_2^{(t)}\}_{t=1}^{\infty}$ is 
monotonically decreasing. 
To see this, first note that $f(r_1^{(1)})=\Delta_1^{(1)}<0$, and $f(r_2^{(1)})=\Delta_2^{(1)}>0$.
Equation~\eqref{eq:rt-update} then tells us that $r^{(2)}_1=r^{(1)}_1\exp(\eta f(r_{2}^{(1)}))>r_1^{(1)}$, while $r^{(2)}_2=r^{(1)}_2\exp(\eta f(r^{(1)}_1))<r^{(1)}_2$. Since $f$ is decreasing, $f(r_1^{(2)})$ becomes smaller (i.e. more negative) while  $f(r_2^{(2)})$  becomes larger (i.e. more positive). Therefore, we have
$$r_1^{(t+1)}=r_1^{(1)}\exp\left( \eta \sum_{k=1}^tf(r_2^{(k)})\right) \geq r_1^{(1)}\exp(\eta t f(r_2^{(1)}))=r_1^{(1)}\exp(\eta t\Delta_2^{(1)}),$$
and
$$r_2^{(t+1)}=r_2^{(1)}\exp\left( \eta \sum_{k=1}^tf(r_1^{(k)})\right) \leq r_2^{(1)}\exp(\eta t f(r_1^{(1)}))=r_2^{(1)}\exp(\eta t\Delta_1^{(1)}).$$
From the above, we get $p^{(t)}_{1,2}\leq \frac{1}{r_1^{(1)}}\exp(-\eta t\Delta_2^{(1)})$,
while 
$p^{(t)}_{2,1}\leq r^{(1)}_2\exp(\eta t\Delta_1^{(1)})$. 
Noting that $\Delta_1^{(1)} < 0$ and $\Delta_2^{(1)} > 0$ are constants, this proves the desired exponential rate of convergence to the asymmetric pure NE $(\theta_1,\theta_2)$.

\paragraph{Condition 2: $\Delta_1^{(1)}>0$, $\Delta_2^{(1)}<0$ (Basic Case 2).} This case is exactly symmetric to Condition 1. Similar to Condition 1, we can show (using induction and based on the fact that $f$ is decreasing) that the ratio of player $1$, $r_1^{(t)}$, is decreasing while the ratio of player $2$, $r_2^{(t)}$, is increasing. As a consequence, we have $p_{1,1}^{(t)}\leq r_{1}^{(1)}\exp(\eta t\Delta_2^{(1)}),$ while $p_{2,2}^{(t)}\leq \frac{1}{r_{2}^{(1)}}\exp(-\eta t\Delta_1^{(1)})$. 
Noting again that $\Delta_1^{(1)} > 0$ and $\Delta_2^{(1)} < 0$ are again constants, this proves the desired exponential rate of convergence to the asymmetric pure NE $(\theta_2,\theta_1)$.
\\

This argument further implies that, if there exists a finite round $t_0$ at which $\Delta_1^{(t)}$ and $\Delta_2^{(t)}$ have different signs, the algorithm thereafter converges exponentially fast to the asymmetric pure NE.

\subsection{Asymptotic Convergence in Same-Sign Initialization}
Next, we consider the more challenging initialization in which $\Delta_1^{(1)}$, $\Delta_2^{(1)}$ have the same sign, and $\Delta_1^{(1)}\not=\Delta_2^{(1)}.$

\paragraph{Condition 3: $\Delta_1^{(1)}$, $\Delta_2^{(1)}$ have the same sign, and $\Delta_1^{(1)}\not=\Delta_2^{(1)}.$} 

Note that since $f$ is strictly monotone and $f(r)=\Delta$, the map $r \to \Delta = f(r)$ is one-to-one. Thus, $\Delta_1^{(1)}\not=\Delta_2^{(1)}$ implies that $r_1^{(1)}\not =r_2^{(1)}$. Assume without loss of generality that $r_i^{(1)}<r_j^{(1)}$, for a pair of players $i \neq j\in\{1,2\}$. In this case, we will show that $\lim_{t\rightarrow\infty }r_j^{(t)}=\infty$ and $\lim_{t\rightarrow\infty }r_i^{(t)}=0$, that is, player $j$ converges to $\theta_1$ while player $i$ converges to $\theta_2$.  

Recall that $\epsilon_1<0$, $\epsilon_2>0$ implies that $f'<0$, i.e. $f(\cdot)$ is monotonically decreasing. Moreover, since $f(r^*)=0$, we have $\text{Sign}(r_i^{(t)}-r^*)=-\text{Sign}(\Delta_i^{(t)})$ for both $i$ and $j$. To show convergence to the pure NE, it is sufficient to show that $r_i^{(t)}-r^*$ and $r_j^{(t)}-r^*$ have opposite signs after a constant, finite number of time steps,  as the analysis then reduces to one of Basic Case 1 or Basic Case 2 above. 

We first show the following fundamental result, which indicates that the relation $r_j^{(t)}>r_i^{(t)}$ is preserved for all $t$. This lemma will be critical for our analysis.

\begin{lemma}
\label{lemma:rqgrp}
Assume that $r_j^{(1)}>r_i^{(1)}$. Then, we have $r_j^{(t)}>r_i^{(t)}$ for all $t\geq 1$. Moreover, we have $\ln\frac{r^{(t)}_j}{r^{(t)}_i}\geq \ln\frac{r^{(1)}_j}{r^{(1)}_i}>0$.  
\end{lemma}
\begin{proof} We prove this lemma by induction. 
The relation holds for the base case ($t = 1$) by assumption, so we only need to prove the inductive step.
Suppose that $r_j^{(t)}>r_i^{(t)}$. Then, we have
$$ \ln r_j^{(t+1)} - \ln r_i^{(t+1)}=\ln\frac{r^{(t+1)}_j}{r^{(t+1)}_i} = \ln \frac{r^{(t)}_j}{r^{(t)}_i}+\eta[f(r_i^{(t)})-f(r_j^{(t)})]\geq \ln r_j^{(t)} - \ln r_i^{(t)}>0,$$
where the first inequality is based on the induction hypothesis, the positivity of $\eta$, and the fact that $f(\cdot)$ is monotonically decreasing.
This completes the proof of the lemma.
\end{proof}
The above lemma implies that $\lim_{t\rightarrow\infty} \frac{r_{j}^{(t)}}{r_{i}^{(t)}}=\infty$. However, it does not directly prove convergence of the algorithm, as it is possible for both the ratios $r_j^{(t)}$ and 
$r_i^{(t)}$ to oscillate, while \emph{their} ratio diverges to infinity.  We first rule out the following corner cases, where one of $r_j^{(t+1)}-r^*$ and  $r_i^{(t+1)}-r^*$ is equal to zero (note that based on Lemma \ref{lemma:rqgrp} they cannot be simultaneously equal to zero):
\begin{itemize}
    \item $r_i^{(t+1)}=r^*<r_j^{(t+1)}$: In this case, we have $r_j^{(t+2)}=r_j^{(t+1)}>r^*$, while $r_i^{(t+2)}=r_i^{(t+1)}\exp(\eta f(r_j^{(t+1)}))<r_i^{(t+1)}=0$.  Thus,  $r_j^{(t+2)}-r^*$ and  $r_i^{(t+2)}-r^*$ have different signs at round $t+2$.
    \item $r_i^{(t+1)}<r^*=r_j^{(t+1)}$: A similar argument yields $r_i^{(t+2)}<r^*$, while $r_j^{(t+2)}>r^*$. Again, $r_j^{(t+2)}-r^*$ and  $r_i^{(t+2)}-r^*$ have different signs at round $t+2$.
\end{itemize}
In summary, these corner cases, if they occur, reduce the problem to one of the Basic Cases above, after which the algorithm would converge to the pure NE. 

We now provide the argument when such corner cases do not manifest. 
In other words, we consider $r_j^{(1)}-r^*$ and $r_i^{(1)}-r^*$ to both be non-zero and have the same sign. Suppose that $r_j^{(t)}-r^*$ and $r_i^{(t)}-r^*$ have the same sign at round $t$. Then from round $t$ to round $t+1$, 
the sign of  $r_j^{(t)}-r^*$ and $r_i^{(t)}-r^*$ can change in only three ways (excluding the corner cases):
\begin{itemize}
    \item Event 1 (``2-sign-flip"): $r_j^{(t+1)}-r^*$ and $r_i^{(t+1)}-r^*$ still have the same sign, but different from   sign of  $r_j^{(t)}-r^*$ and $r_i^{(t)}-r^*$;
    \item Event 2 (``0-sign-flip"):  $r_j^{(t+1)}-r^*$ and $r_i^{(t+1)}-r^*$ still have the same sign, and also the same as the sign of $r_j^{(t)}-r^*$ and $r_i^{(t)}-r^*$;
    \item Event 3 (``1-sign-flip"): $r_j^{(t+1)}-r^*$ and $r_i^{(t+1)}-r^*$ have different signs. Based on Lemma \ref{lemma:rqgrp}, we must then have $r_j^{(t+1)}>r^*>r_i^{(t+1)}$. 
\end{itemize}

Note that, once Event 3 happens,  the problem again reduces to the Basic Cases.  Therefore, it is sufficient to show that the first two ``bad events" only happen finitely many times, and Event 3 must eventually happen. To do so, we state the following key lemma, which shows that 2-sign flips can only happen a finite number of times. The proof of Lemma~\ref{lemmmm:noffilop} is deferred to Appendix \ref{Sectionprofejjrhrhueihi}.

\begin{lemma}
\label{lemmmm:noffilop}
Suppose $\epsilon_1<0$, $\epsilon_2>0$. Let $n$ be the number of times that a ``$2$-sign flip" (Event 1) happens. Define the universal constants
\begin{align*}
\beta&=\eta\cdot\max\{-\epsilon_1,\epsilon_2\}, \\
C&=\eta(\epsilon_2-\epsilon_1)\min\left\{\frac{r^{*}\exp(-\beta)}{(1+r^{*}\exp(-\beta))^2},\frac{r^{*}\exp(\beta)}{(1+r^{*}\exp(\beta))^2}\right\}>0.
\end{align*}
Further, let $W_1=\ln\frac{r^{(1)}_j}{r^{(1)}_i}>0$. Then, we have
\[
n\leq\frac{1}{\ln(1+C)}\ln\frac{2\beta}{W_1}.
\]
\end{lemma}
To show that Event 2 (``0-sign-flip") happens finitely often, we \emph{disprove} the following claim.
We can now finish the proof by \emph{disproving} the following claim.  
\begin{claim}\label{clm:0signflip}
Let $t_0$ be the smallest round such that all ``2-sign-flips" have finished (as per Lemma~\ref{lemmmm:noffilop}, $t_0$ will be finite).
Moreover, suppose that $r_j^{(t_0)} - r^*$ and $r_i^{(t_0)} - r^*$ still have the same sign.
Then, both $r_j^{(t)}-r^*$ and $r_i^{(t)}-r^*$ do not change their sign for all $t>t_0$.
\end{claim}
Suppose that the claim above is true. Lemma~\ref{lemma:rqgrp} tells us that $r_j^{(t)}>r_i^{(t)}$ for all $t$. First, suppose that $r_j^{(t)}>r_i^{(t)}>r^*$ for all $t>t_0$. Since $f(\cdot)$ is strictly decreasing, we have for any player $m \in\{1,2\}$ that $f(r_m^{(t)})<f(r^*)=0$.
Equation~\eqref{eq:rt-update} would then give us
\[
r_m^{(t+1)}=r_m^{(t)}\exp\bigl(\eta f(r_n^{(t)})\bigr),
\]
for any player pair $m \neq n \in \{1,2\}$.
This implies that the ratio $\{r_m^{(t)}\}_{t \geq 1}$ would be decreasing for each player $m \in \{1,2\}$.
Moreover, $r_m^{(t)}\geq 0$ since it is a probability ratio. Therefore, based on the monotone convergence theorem, we have $\lim_{t\rightarrow \infty}r^{(t)}_m=\ell_m$ for some constant $\ell_m\geq 0$. Further, since $f(\cdot)$ is continuous and monotone, we have
\[
 f(\ell_m)=f \left(\lim_{t\rightarrow\infty }r^{(t)}_m\right)=\lim_{t\rightarrow\infty }f(r^{(t)}_m)
   =\frac{1}{\eta}\lim_{t\rightarrow\infty }\bigl(\ln r^{(t+1)}_n -\ln r^{(t)}_n\bigr)=0,
\]
i.e., $f(\ell_m)=0$. Since $f$ is strictly monotone and $f(r^*)=0$, we have $\ell_m=r^*$ for both players $m\in\{1,2\}$. Thus, we have
\[
 \lim\limits_{t\rightarrow \infty} \ln\frac{r_j^{(t)}}{r_i^{(t)}} 
   =  \lim\limits_{t\rightarrow \infty} \ln{r_j^{(t)}}-\ln{r_i^{(t)}}=0.
\]
However, Lemma~\ref{lemma:rqgrp} implies that $\ln{r_j^{(t)}}-\ln{r_i^{(t)}}\geq W_1>0$, which contradicts the above equation. \\

We similarly handle the other possible situation, where $r_i^{(t)}<r_j^{(t)}<r^*$ for all $t>t_0$. Since $f(\cdot)$ is strictly decreasing, we have for any player $m \in\{1,2\}$ that $f(r_m^{(t)})>f(r^*)=0$.
Equation~\eqref{eq:rt-update} would then give us
\[
r_m^{(t+1)}=r_m^{(t)}\exp \bigl(\eta f(r_n^{(t)})\bigr)
\]
for any player pair $m \neq n \in \{1,2\}$.
This implies that the ratio $\{r_m^{(t)}\}_{t \geq 1}$ would be increasing for each player $m \in \{1,2\}$.
Moreover, under the claim that the signs do not change for either player, we would have $r_m^{(t)}\leq r^*$. 
Therefore, based on the monotone convergence theorem, we have $\lim_{t\rightarrow \infty}r^{(t)}_m=\ell_m$ for some constant $\ell_i\leq r^*$. Further, since $f$ is continuous and monotone, we have
\[
 f(\ell_m)=f \left(\lim_{t\rightarrow\infty }r^{(t)}_m\right)
 =\lim_{t\rightarrow\infty }f(r^{(t)}_m)
 =\frac{1}{\eta}\lim_{t\rightarrow\infty }\bigl(\ln r^{(t+1)}_n -\ln r^{(t)}_n\bigr)=0,
\]
i.e., $f(\ell_m)=0$. Since $f$ is monotone and $f(r^*)=0$, we have $\ell_m=r^*$ for both players $m\in\{1,2\}$. Thus, we have
\[
 \lim\limits_{t\rightarrow \infty} \ln\frac{r_j^{(t)}}{r_i^{(t)}} 
 =  \lim\limits_{t\rightarrow \infty} \ln{r_j^{(t)}}-\ln{r_i^{(t)}}=0.
\]
However, Lemma~\ref{lemma:rqgrp} implies that $\ln{r_j^{(t)}}-\ln{r_i^{(t)}}\geq W_1>0$, which contradicts the above equation.

This completes the \emph{disproof} of Claim~\ref{clm:0signflip}.
Combined with Lemma~\ref{lem:::verysimpleobsercation}, this means that a $1$-sign-flip (Event 3) must happen at some finite round, which completes the proof of Theorem~\ref{thm:e1l0e2g0} in this case.

\subsection{Convergence and Divergence in the Identical Initialization Case}
Finally, we consider the case of equal initialization, where $\Delta_1^{(1)}=\Delta_2^{(1)}$.

\paragraph{Condition 4: $\Delta_1^{(1)}=\Delta_2^{(1)}$, $\eta<\frac{8}{|\epsilon_1|+|\epsilon_2|}.$}
For convenience, we define $u_i^{(t)}=\ln\frac{r_{i}^{(t)}}{r^*}.$ We can then define a $1$-dimensional dynamical system for each player, given by
$$ u_i^{(t+1)}= \ln\frac{r^{(t+1)}_i}{r^*}=u_i^{(t)}+\eta\,f(r^*\exp(u_{j}^{(t)})). $$ 
Since $\Delta_1^{(1)}=\Delta_2^{(1)}$ and both players use the same update rule~\eqref{eq:rt-update}, we observe that $r^{(t)}_1=r^{(t)}_2$ for all $t \geq 1$, and the same holds for $u_i^{(t)}$. Therefore, we can denote a common variable $u^{(t)} := u_i^{(t)}$ and write its update in the following form:
$$ u^{(t+1)}=T(u^{(t)}), $$
where we define
$$ T(u)= u+\eta\,\frac{\epsilon_1r^*\exp(u)+\epsilon_2}{1+r^*\exp(u)}.$$
We now have 
$$T'(u) = 1+ \eta\,\frac{(\epsilon_1 - \epsilon_2)r^* \exp(u)}{(1 + r^* \exp(u))^2}= 1-\eta\,\frac{\Gamma r^* \exp(u)}{(1 + r^* \exp(u))^2}.$$
where we denote $\Gamma=|\epsilon_1|+\epsilon_2>0$ as shorthand. Also note that $T(0)=0$ is a fixed point.
Because $T$ is a continuous function, we can use the mean value theorem to write
$$ |T(u)-T(0)|=|T'(u')||u-0|$$
for any value $u$, where $u'$ is some point between $u$ and $0$. Under our assumption on the step size, we have $\eta\Gamma\in(0,8)$. In this case, based on the fact that $\max_{z\geq 0}\frac{z}{(1+z)^2}=\frac{1}{4}$, it can be easily seen that $|T'(u')|<1$ for all $u'$. Putting these together, we see that $T$ is a \emph{contraction map}.
Applying Banach's fixed point theorem then implies that $u^{(t)} \to 0$.

\paragraph{Condition 5: $\Delta_1^{(1)}=\Delta_2^{(1)}$, $\eta>\frac{8}{|\epsilon_1|+|\epsilon_2|}.$} Here, we construct a simple example in which the algorithm oscillates. For simplicity, let $\eta=1$ and $|\epsilon_1|=\epsilon_2$.
This means that $r^*=1$, and the map $T$ simplifies to
$$T(u)=u+\epsilon_1\frac{\exp(u)-1}{1+\exp(u)}=u+\epsilon_1 \text{tanh}\left(\frac{u}{2}\right).$$
Since $\text{tanh}(u)$ and $u$ are both odd functions in $u$, $T$ is also an odd function. Formally, we have $T(-u)=-T(u)$. Let $a>0$ be any constant, and $$\epsilon_1=-2a \text{coth}\left(\frac{a}{2}\right), \ \ \epsilon_2=2a \text{coth}\left(\frac{a}{2}\right).$$
Note that $a\text{coth}(\frac{a}{2})> 2$, so $\eta(|\epsilon_1|+|\epsilon_2|)>8$ in this case. Therefore, we have
 \[
T(a) = a + \varepsilon_1 \tanh \left(\tfrac{a}{2}\right)
= a - 2 a \coth \left(\tfrac{a}{2}\right)\tanh \left(\tfrac{a}{2}\right)
= a - 2 a \,(\tanh \cdot \coth) = -a,
\]
and because $T$ is an odd function we have $T(-a)= -T(a) = -(-a) = a$. Therefore, the algorithm will oscillate if initialized at $a$ or $-a$. 
Note that such a game can be defined for any positive constant $a$.

\subsection{Proof of Lemma \ref{lemmmm:noffilop} } 
\label{Sectionprofejjrhrhueihi}

We complete this section by providing the proof of Lemma~\ref{lemmmm:noffilop}.
Recall that, without loss of generality, we index a player pair $i \neq j \in \{1,2\}$ such that $r_j^{(1)}>r_i^{(1)}$.  We also define $u_m^{(t)}= \ln\frac{r^{(t)}_m}{r^*}$ for each player $m \in \{1,2\}$.
Since $r_i^{(1)}<r_j^{(1)}$, we have $u^{(1)}_i<u^{(1)}_j$. Moreover, for each player $m\in\{1,2\}$, $u^{(t)}_m$ and $r^{(t)}_m-r^*$ have the same sign. Reproducing the main argument in the proof of Lemma~\ref{lemma:rqgrp} below, we also have
\begin{equation}
u^{(t+1)}_j- u^{(t+1)}_i =\ln\frac{r_j^{(t+1)}}{r_i^{(t+1)}} = \ln \frac{r_j^{(t)}}{r_i^{(t)}} + \eta[f(r_i^{(t)})-f(r_j^{(t)})]> \ln \frac{r_j^{(t)}}{r_i^{(t)}}= u^{(t)}_j -u^{(t)}_i >0.
\end{equation}
Next, we show that, if Event 1 (i.e. ``2-sign-flip") happens at round $t$, then  $|u^{(t)}_m|$ must be small for both players $m \in \{1,2\}$.

\begin{lemma}
\label{lem:uusdnsdji}
Suppose $u^{(t)}_i$ and $u^{(t)}_j$ have the same sign, and the sign of both $u^{(t+1)}_i$ and $u^{(t+1)}_j$ flip. Then, we must have $\max\{|u^{(t)}_i|,|u_{j}^{(t)}|\}\leq \beta = \eta\cdot \max\{-\epsilon_1,\epsilon_2\}$.
\end{lemma}

\begin{proof}\ First, note that $f$ is monotonically decreasing, $f(0)=\epsilon_2>0$, and $f(\infty)=\epsilon_1<0$. 
This means that $\epsilon_1 \leq f(r) \leq \epsilon_2$ for any ratio $r$.

Next, consider the case where $u^{(t)}_m<0$, and $u^{(t+1)}_m>0$ for both players $m \in \{1,2\}$.
We would like to show that in this case that $\max\{-u^{(t)}_1, -u^{(t)}_2\}\leq \beta.$ We will prove this statement by contradiction. Since we have $- u_i^{(t)} > - u_j^{(t)}$, we suppose instead that $-u_i^{(t)}>\beta>\eta\epsilon_2$, i.e., $u_i^{(t)}<-\beta<-\eta\epsilon_2.$ Then, we have 
$$ u_i^{(t+1)} = u_i^{(t)} + \eta f(r^*\exp(u^{(t)}_j))<-\eta\epsilon_2+\eta\epsilon_2=0,$$
where the inequality above substitutes our supposition $u_i^{(t)} < -\eta \epsilon_2$ and uses the fact that $f(r) \leq \epsilon_2$ for any value of $r$.
Overall, this yields $u_i^{(t+1)} < 0$, which contradicts our assumption of ``2-sign-flip" which would instead yield $u^{(t+1)}_i>0$. 
Similarly, consider the case where $u^{(t)}_m>0$, and $u^{(t+1)}_m<0$ for both players $m\in\{1,2\}$. We will again prove the statement by contradiction. Suppose instead that $u_j^{(t)}>\beta> - \eta\epsilon_1$. Then, we have
$$u^{(t+1)}_j = u_j^{(t)}+\eta f(r^*\exp(u^{(t)}_i))>-\eta\epsilon_1+\eta\epsilon_1=0, $$
where the inequality above substitutes our supposition $u_j^{(t)} > - \eta \epsilon_1$ and uses the fact that $f(r) \geq \epsilon_1$ for any value of $r$.
Overall, this yields $u_j^{(t+1)} > 0$, which contradicts our assumption of ``2-sign-flip" which would instead yield $u_j^{(t+1)} < 0$.
This completes the proof of this lemma.
\end{proof}

Armed with Lemma~\ref{lem:uusdnsdji}, we are ready to prove Lemma~\ref{lemmmm:noffilop}. Define the potential function $W_t=u_j^{(t)}-u_i^{(t)}$. We know from Lemma~\ref{lemma:rqgrp} that $W_t$ is monotonically increasing. Then, we have 
$$W_{t+1} =u_j^{(t+1)}-u_i^{(t+1)}= u_j^{(t)}-u_i^{(t)} +\eta(g(u_i^{(t)})-g(u_j^{(t)}))=(1-\eta g'(\delta_t))(u_j^{(t)}-u_i^{(t)}), $$
where we defined $g(u):=f(r^*\exp(u))$ as shorthand.
As it is easy to see that $g(u)$ is monotone and continuous, the last equality above is based on the mean value theorem (where $\delta^{(t)}\in[u^{(t)}_i,u^{(t)}_j]$). Note that for any $|u|\leq \beta$, we have
$$ 1-\eta g'(u) = 1+\eta\frac{(\epsilon_2-\epsilon_1)r^*\exp(u)}{(1+r^*\exp(u))^2}\geq 1+ \eta(\epsilon_2-\epsilon_1)\min\left\{\frac{r^{*}\exp(-\beta)}{(1+r^{*}\exp(-\beta))^2},\frac{r^{*}\exp(\beta)}{(1+r^{*}\exp(\beta))^2}\right\}. $$
Let $\{t_1,\ldots,t_n\}$ denote the rounds on which ``2-sign-flip" happens.
Using the monotonically increasing property of $\{W_t\}_{t \geq 1}$, we have
$$W_{t_n}\geq (1+C)W_{t_{n-1}} \ldots \geq (1+C)^n W_1,$$
where $n$ is the total number of rounds where Event 1 happens, and 
\[
C=\eta(\epsilon_2-\epsilon_1)\min\left\{\frac{r^{*}\exp(-\beta)}{(1+r^{*}\exp(-\beta))^2},\frac{r^{*}\exp(\beta)}{(1+r^{*}\exp(\beta))^2}\right\}>0.
\] 
On the other hand, Lemma~\ref{lem:uusdnsdji} tells us that $W_{t_n}\leq 2\beta$.
Combining the upper and lower bounds on $W_{t_n}$ completes the proof of Lemma~\ref{lemmmm:noffilop}.

\section{Proof of Theorem \ref{thm:theoppositcase}}
\label{appendix:thm:theoppositcase}

\begin{figure}
  \centering

  \begin{subfigure}{0.45\textwidth}
    \includegraphics[width=\linewidth]{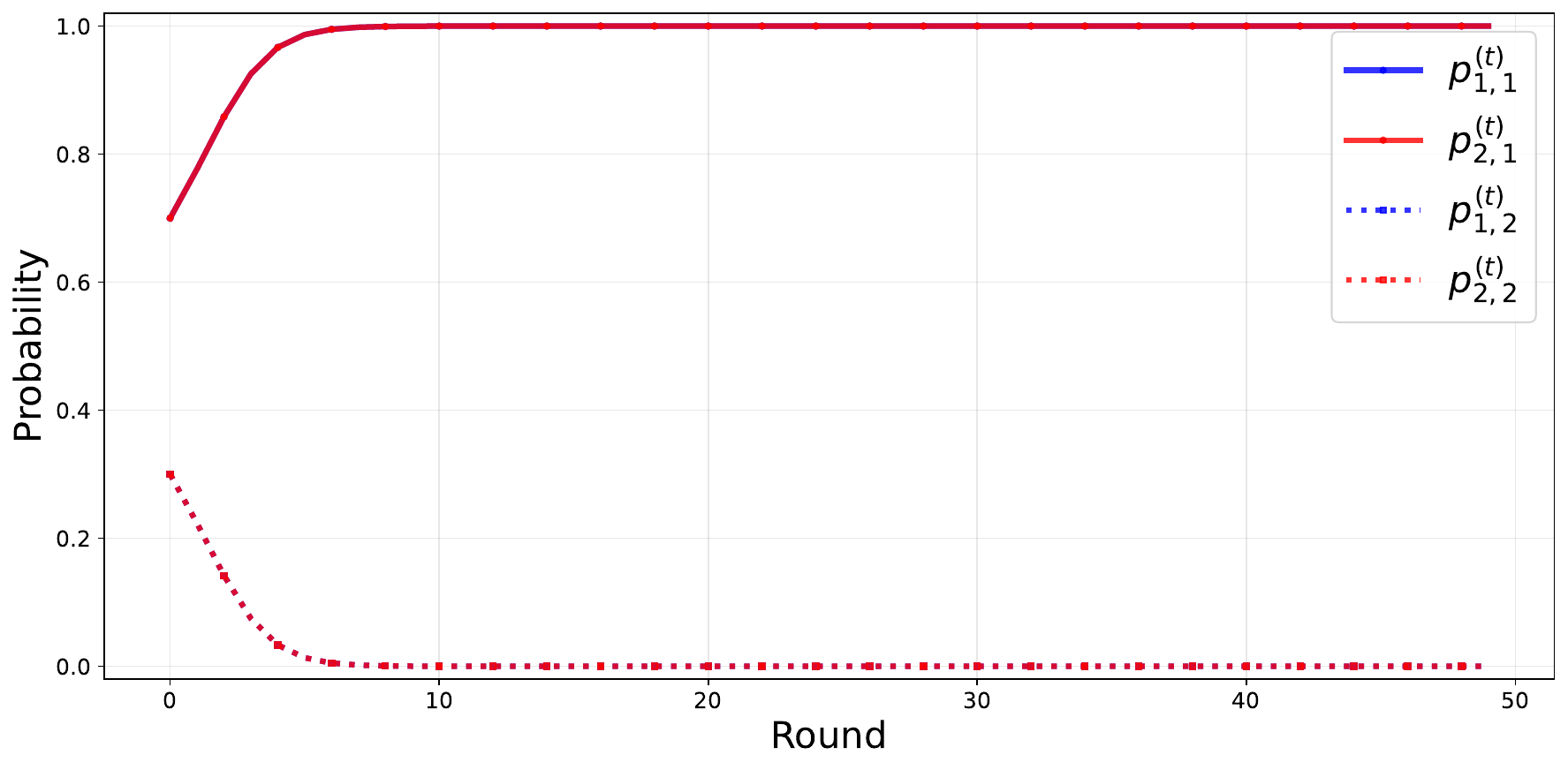}
    \caption*{r6. Identical initialization}
  \end{subfigure}
  \begin{subfigure}{0.45\textwidth}
    \includegraphics[width=\linewidth]{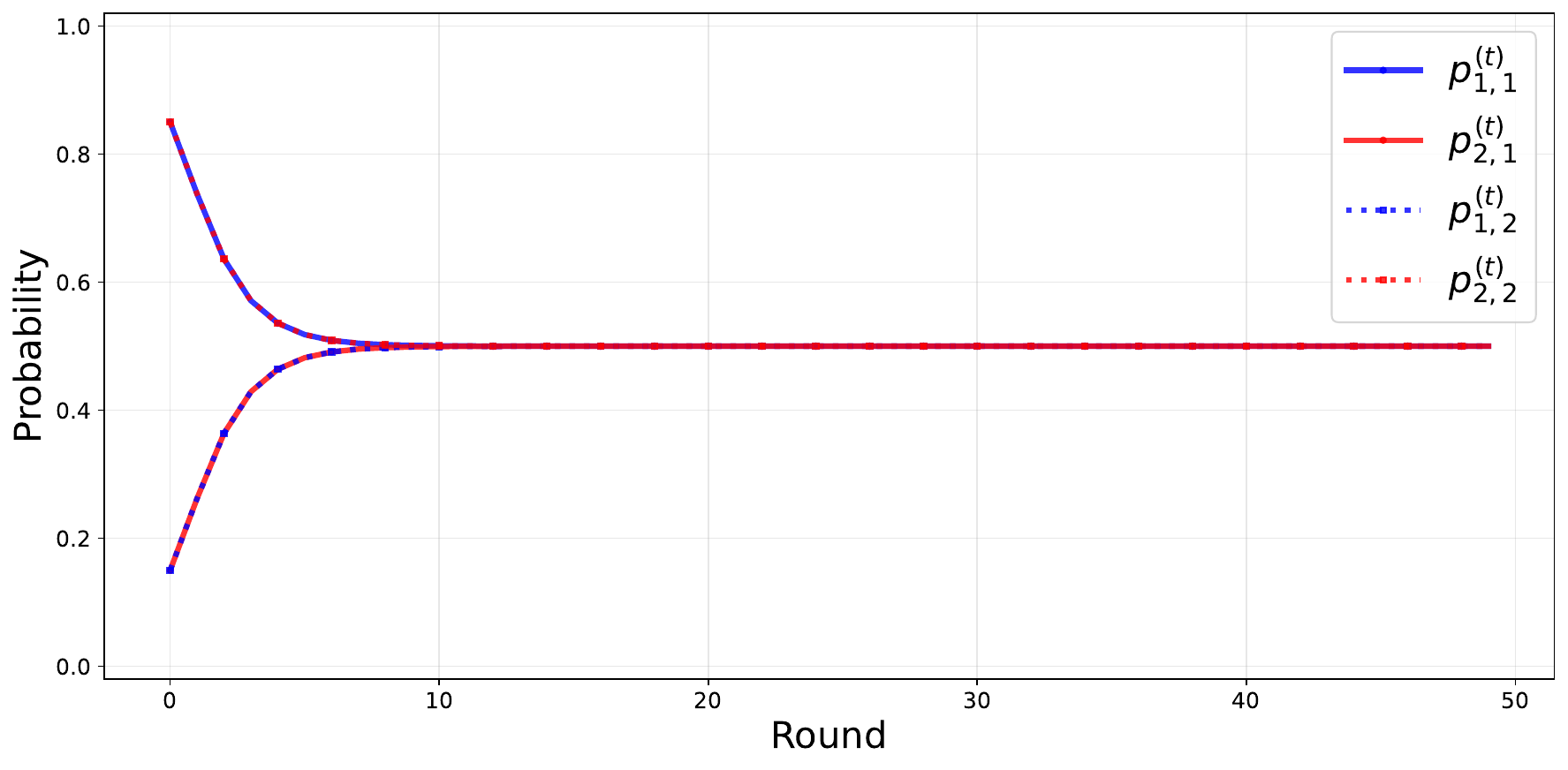}
    \caption*{r7. Opposite-sign initialization with $\eta>\frac{8}{|\epsilon_1|+|\epsilon_2|}$}
  \end{subfigure}

  \begin{subfigure}{0.45\textwidth}
    \includegraphics[width=\linewidth]{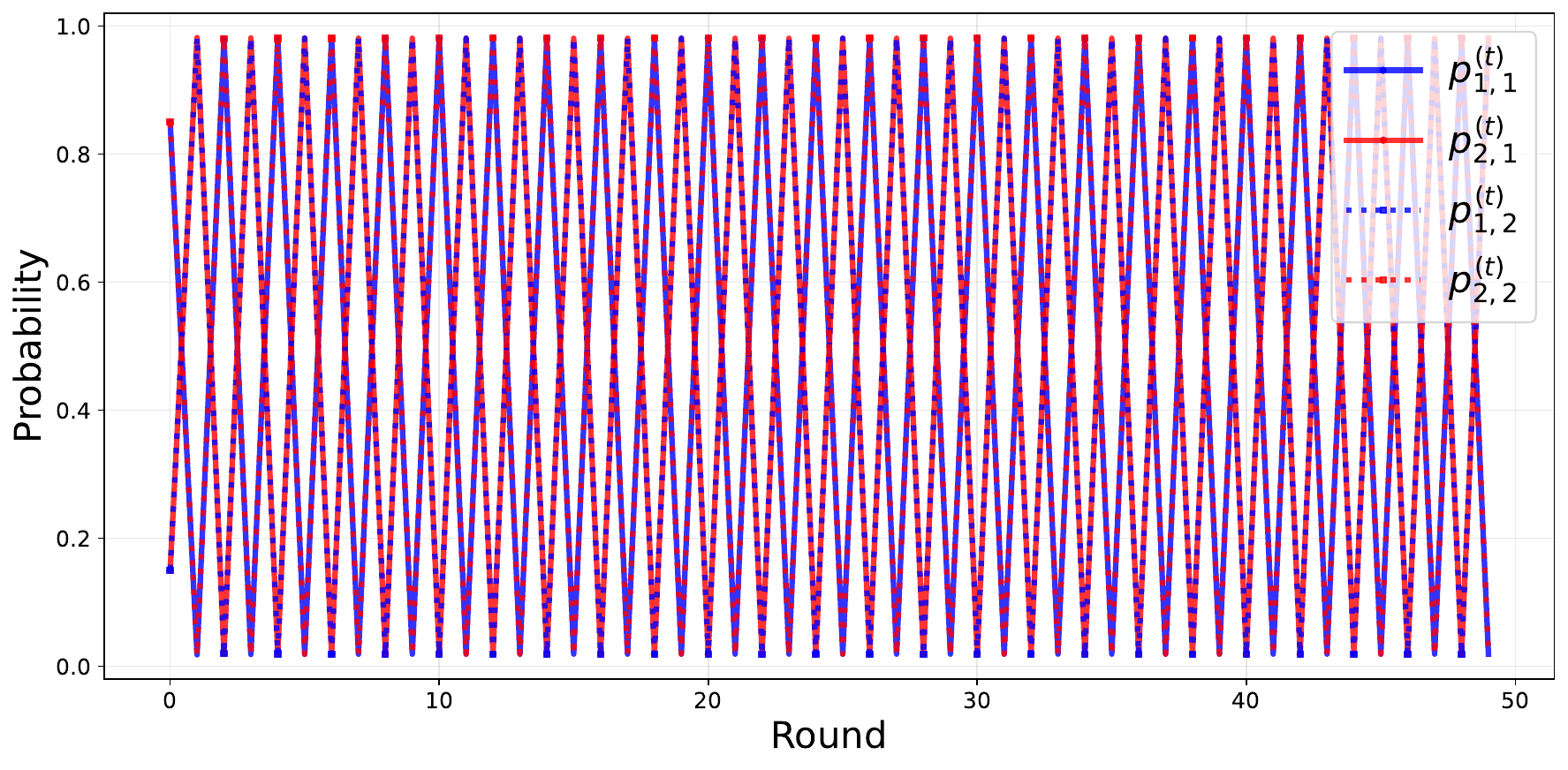}
    \caption*{r7. Opposite-sign with $\eta>\frac{8}{|\epsilon_1|+|\epsilon_2|}$}
  \end{subfigure}
  \begin{subfigure}{0.45\textwidth}
    \includegraphics[width=\linewidth]{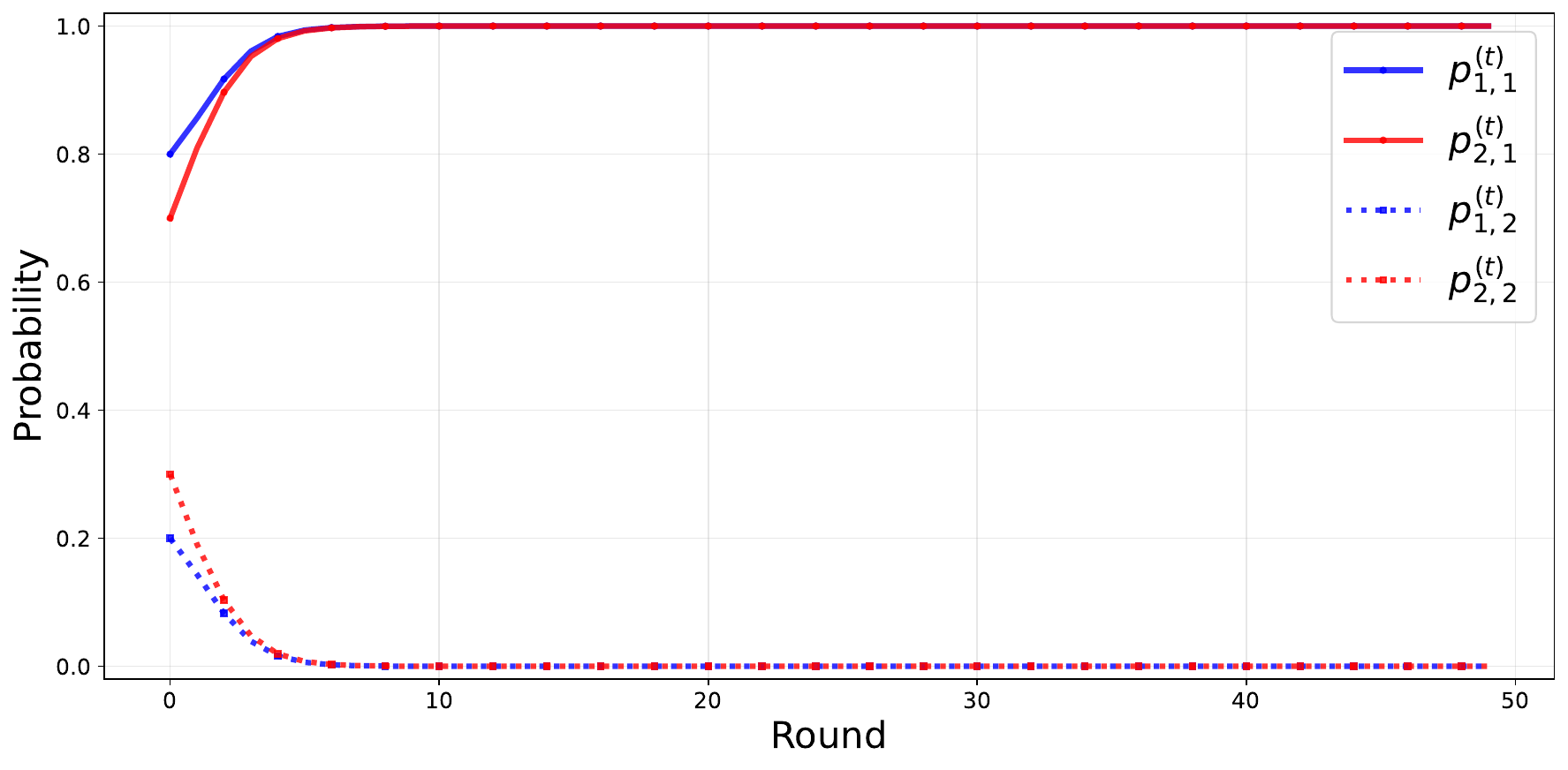}
    \caption*{r6. Same-sign initialization}
  \end{subfigure}

  \caption{Simulation results for Theorem \ref{thm:theoppositcase}. The row in Table~\ref{Table:NEconvergence} to which each case corresponds is marked in the sub-captions.}
  \vspace{-0.4cm}
  \label{figure:part 3}
\end{figure}

In this section, we prove Theorem~\ref{thm:theoppositcase}. In this case, we have $\epsilon_1>0$ and $\epsilon_2<0$. The corresponding simulation results are given in Figure~\ref{figure:part 3} (we did not include them in the main paper due to space limitations). We first recall/introduce some basic conclusions. Recall that $f'(r)=\frac{\epsilon_1-\epsilon_2}{(1+r)^2}>0$, which means that $f(\cdot)$ is strictly increasing for the case where $\epsilon_1 > 0$ and $\epsilon_2 < 0$. Moreover, $r^*=\frac{|\epsilon_2|}{\epsilon_1}$ is the unique root of $f(\cdot)$.
Recall that for both players $i\in\{1,2\}$, we can write the following dynamical system
$$ u_{i}^{(t+1)}= \ln\frac{r^{(t+1)}_i}{r^*}= u_i^{(t)} + g(u_{j}^{(t)}),$$
where we defined $g(u)=\eta f(r^*\exp(u))$ as shorthand.
Note that, since $f$ is strictly increasing and $f(r^*)=0$, if $r-r^*>0$, $\Delta=f(r)>0$, and also $u=\ln\frac{r}{r^*}>0$. On the other hand, if $r-r^*<0$, $\Delta=f(r)<0$, and also $u=\ln\frac{r}{r^*}<0$. Thus, we have $\text{Sign}(\Delta)=\text{Sign}(r-r^*)=\text{Sign}(u)$. Next, we have
$$g'(u)=\frac{\eta(\epsilon_1-\epsilon_2)r^*\exp(u)}{(1+r^*\exp(u))^2}.$$
For all $u\in\R$, we have $g'(u)\leq \frac{\eta(\epsilon_1-\epsilon_2)}{4}.$ Therefore, we have $g'(u)\in(0,2)$ when $\eta\in\left(0,\frac{8}{|\epsilon_1|+|\epsilon_2|}\right)$. Next, we provide our convergence result for different choices of initialization. 

\subsection{Exponential Convergence with Same Sign Initialization}
We first provide a formal description of exponential convergence in the ``easy" cases where the signs of the initialization functionals $\Delta_1^{(1)}$ and $\Delta_2^{(1)}$ are the same.

\begin{theorem}\label{thm:partthree-easycase-formal}
Suppose Assumption \ref{as:nondegenerate} holds, $\epsilon_1>0$ and $\epsilon_2<0$. Then, Algorithm \ref{alg:Hedge} with any step size $\eta > 0$ satisfies the following:
\begin{itemize}
    \item If $\Delta_1^{(1)}>0$, $\Delta_2^{(1)}>0$, then we have $p^{(t+1)}_{i,2}\leq \frac{1}{r^{(1)}_i}\exp\left(-\eta \Delta_{j}^{(1)}\right)$ for any player pair $i \neq j \in \{1,2\}$. This implies that the algorithm converges to the pure symmetric NE $(\theta_{1},\theta_1)$ exponentially fast;
\item If $\Delta_1^{(1)}<0$, $\Delta_2^{(1)}<0$, then we have $p^{(t+1)}_{i,1}\leq {r^{(1)}_i}\exp\left(\eta \Delta_{j}^{(1)}\right)$ for any player pair $i \neq j \in \{1,2\}$. This implies that the algorithm converges to the pure symmetric NE $(\theta_{2},\theta_2)$ exponentially fast.
\end{itemize}
\end{theorem}
We now provide the proof of Theorem~\ref{thm:partthree-easycase-formal}, considering each condition listed above.

\paragraph{Condition 1: $\Delta_1^{(1)}>0$, $\Delta_2^{(1)}>0$.} Recall that for any player pair $i \neq j \in\{1,2\}$, we have $r_i^{(2)}=r_i^{(1)}\exp(\Delta_j^{(1)}).$ Based on induction, it is easy to see that the ratios for both players, $r_1^{(t)}$ and  $r_2^{(t)}$, are increasing in $t$. Moreover, for each player $i \in \{1,2\}$, the functional $\Delta_i^{(t)}=f(r_i^{(t)})$ is also increasing in $t$ (since $f(\cdot)$ is increasing in its argument and $r_i^{(t)}$ is increasing in $t$). Therefore, we have
$$ r_i^{(t+1)}=r^{(1)}_i\exp\left(\eta\sum_{k=1}^t \Delta^{(k)}_j\right)\geq r^{(1)}_i\exp\left( \eta t\Delta_j^{(1)}\right)$$
for any player pair $i \neq j \in \{1,2\}$.
Noting that $\Delta_j^{(1)} > 0$ is a positive constant, this proves the desired exponential rate of convergence to the symmetric pure NE $(\theta_1,\theta_1)$.

\paragraph{Condition 2: $\Delta_1^{(1)}<0$, $\Delta_2^{(1)}<0$.} 
Using a similar inductive argument, it is easy to see that the ratios for both players $r_1^{(t)}$ and $r_2^{(t)}$ are now decreasing in $t$.
Because $f(\cdot)$ is increasing, the functional $\Delta_i^{(t)}=f(r_i^{(t)})$ is also decreasing in $t$ (since $f(\cdot)$ is increasing in its argument and $r_i^{(t)}$ is decreasing in $t$).
Therefore, we have 
$$ r_i^{(t+1)}=r^{(1)}_i\exp\left(\eta\sum_{k=1}^t \Delta^{(k)}_j\right)\leq r^{(1)}_i\exp\left( \eta t\Delta_j^{(1)}\right)$$
for any player pair $i \neq j \in \{1,2\}$.
Noting that $\Delta_j^{(1)} < 0$ is a negative constant, this proves the desired exponential rate of convergence to the symmetric pure NE $(\theta_2,\theta_2)$.

\subsection{Convergence with Opposite Sign Initialization}
Finally, we consider the case where $\Delta_1^{(1)}$ and $\Delta_2 ^{(1)}$ have different signs.
\paragraph{Condition 3: $\Delta_1^{(1)}$ and $\Delta_2 ^{(1)}$ have different signs. } Suppose that we have $\Delta_i^{(t)}>0$ and $\Delta_j^{(t)}<0$ for some player pair $i \neq j\in\{1,2\}$ for some $t\geq 1$.  We start by excluding some corner cases for the update to round $t+1$, i.e. where one or both of the $\Delta_i^{(t+1)}$ becomes equal to $0$.
\begin{itemize}
    \item Suppose that $\Delta_i^{(t+1)}=0$ and $\Delta_j^{(t+1)}>0$. Then, we have $\Delta_{j}^{(t+2)}=\Delta_{j}^{(t+1)}>0$, and $r_i^{(t+2)}=r_i^{(t+1)}\exp(\eta \Delta_{j}^{(t+1)})>r_i^{(t+1)}$.
    Because $f(\cdot)$ is increasing, we have $\Delta_{i}^{(t+2)}>\Delta_i^{(t+1)}=0$. 
    This means that $\Delta_i^{(t+2)} > 0$ for both players $i \in \{1,2\}$, and the situation reduces to Condition 1.
    \item Suppose that $\Delta_i^{(t+1)}=0$ and $\Delta_j^{(t+1)}<0$. A similar argument to the first corner case yields $\Delta^{(t+2)}_i < 0$ for both players $i \in \{1,2\}$, and the situation reduces to Condition 2.
    \item Finally, suppose that $\Delta_{j}^{(t+1)}=\Delta_{i}^{(t+1)}=0$. Then, since $f$ is strictly monotone and $f(r^*)=0$, we have that $r^{(t+1)}_i=r^{(t+1)}_j=r^*$. Moreover, since $r^{(t+1)}_i=r_i^{(t)}\exp(\Delta_{_j}^{(t)})$, the algorithm will stay at  $r^{(t+1)}_i=r^{(t+1)}_j=r^*$, which is the unique \emph{strictly} mixed NE of the game.
\end{itemize}
Therefore, we can exclude the corner cases above. Then, as in the proof of Theorem~\ref{thm:e1l0e2g0}, three events can transpire going from $t \to t+1$:
\begin{itemize}
    \item Event 1 (``2-sign-flip"): $\Delta_i^{(t+1)}<0$ and $\Delta_j^{(t+1)}>0$;
    \item Event 2 (``1-sign-flip"): $\Delta_i^{(t+1)}$ and $\Delta_j^{(t+1)}$ have the same sign;
    \item Event 3 (``0-sign-flip"): $\Delta_i^{(t+1)}>0$ and $\Delta_j^{(t+1)}<0$;
\end{itemize}
We first focus on Event 1. In particular, suppose Event 1 happens at round $t$. Define the potential function $V_t=|u_{1}^{(t)}|+|u_{2}^{(t)}|$. Because the signs of $u_i^{(t)}$ and $\Delta_i^{(t)}$ match for both players $i \in \{1,2\}$, we have
$$ V_{t}=u_{i}^{(t)}-u_{j}^{(t)},$$
and therefore we can write the potential function update as
\begin{equation}
    \begin{split}
    \label{eqn:vvtandt+1}
V_{t+1} =  u_{j}^{(t+1)}-u_{i}^{(t+1)}  =  u_{j}^{(t)}-u_{i}^{(t)} +(g(u^{(t)}_i)-g(u^{(t)}_j))=(g'(\delta^{(t)})-1)V_t,
    \end{split}
\end{equation}
where the last equality is based on the mean value  theorem for some $\delta^{(t)}\in[u_j^{(t)},u_i^{(t)}]$. From this, we can infer the following:
\begin{itemize}
    \item If $\eta\in\left(0,\frac{4}{|\epsilon_1|+|\epsilon_2|}\right)$, then $g'(u)<1$. Thus, \eqref{eqn:vvtandt+1} implies that $V_{t+1}<0$, which contradicts that $V_{t+1}\geq 0$ by definition. Thus, if   $\eta\in\left(0,\frac{4}{|\epsilon_1|+|\epsilon_2|}\right)$, Event 1 cannot ever happen.
    \item If instead $\eta\in\left[\frac{4}{|\epsilon_1|+|\epsilon_2|},\frac{8}{|\epsilon_1|+|\epsilon_2|}\right)$, we have $g'(u)-1<1$. In this case, $V_{t+1}$ will shrink by a constant every time Event 1 happens. 
\end{itemize}
To summarize, if $\eta\in\left(0,\frac{8}{|\epsilon_1|+|\epsilon_2|}\right)$, one of the two outcomes will transpire:
\begin{itemize}
    \item \emph{If Event 1 a finite number of times:} i.e., there exists a large enough $t_0$ such that Event 1 will not happen after $t_0$. Now, consider all rounds $t>t_0$. If Event 2 happen at some round $t$, then the problem reduces to Conditions 1 or 2. Otherwise, only Event 3 happens at each round. Based on the monotone convergence theorem and monotonicity of $f(\cdot)$, the algorithm converges to the strictly mixed NE in this case.
    \item \emph{If Event 1 happens infinitely many times:} In this case, $\lim_{t\rightarrow\infty } V_t=0$, which means that both $|u_{i}^{(t)}|$ and $|u_{j}^{(t)}|$ converge to 0. In this case, the algorithm converges to the strictly mixed NE. 
\end{itemize}
In conclusion, for $\eta\in\left(0,\frac{8}{|\epsilon_1|+|\epsilon_2|}\right)$, the algorithm will converge to either the pure NE or the strictly mixed NE.
This completes the proof of Theorem~\ref{thm:theoppositcase}.

\subsection{An example of oscillation with large step size under opposite-sign initialization}

In this section, we provide a simple example of a game in which the algorithm oscillates if the step size is instead too large, i.e. $\eta > \frac{8}{|\epsilon_1| + |\epsilon_2|}$, with a very specific initialization.
This eample is analogous and similar to the one provided in Condition 5 in the proof of Theorem~\ref{thm:e1l0e2g0}.
For simplicity, we consider $\eta = 1$ and $\epsilon_1 = |\epsilon_2| = \epsilon$, where $\epsilon$ will be a parameter that we set later.
Note that this implies that the unique root of $f(\cdot)$ is $r^* = 1$.

We now consider \emph{equal and opposite} initialization, i.e. $\Delta_1^{(1)} = - \Delta_1^{(2)} = \Delta^{(1)}$.
It is easy to verify that in this case, $r_2^{(1)} = 1/r_1^{(1)}$.
Moreover, we can show via induction that $r_2^{(t)} = 1/r_1^{(t)}$.
Therefore, it suffices to keep track of a common ratio $r^{(t)} := r_1^{(t)}$ and the corresponding common variable $u^{(t)} = \ln \frac{r^{(t)}}{r^*} = \ln r^{(t)}$.
We can then define a $1$-dimensional dynamical system on $u^{(t)}$, where the update $t \to t + 1$ is given by
\begin{align*}
    u^{(t+1)} &= u^{(t)} - \Delta^{(t)} \\
    &= u^{(t)} - f(r^{(t)}) \\
    &= u^{(t)} - f(\exp(u^{(t)})).
\end{align*}
Recalling the definition $f(r) = \frac{\epsilon_1 r + \epsilon_2}{1 + r}$ and noting that we are considering a game where $\epsilon_1 = - \epsilon_2 = \epsilon$, we have $f(r) = \epsilon \frac{r-1}{1+r}$, we have
\begin{align*}
    u^{(t+1)} &= u^{(t)} - \epsilon \frac{\exp(u^{(t)}) - 1}{1 + \exp(u^{(t)})} \\
    &= u^{(t)} - \epsilon \tanh{\left(\frac{u^{(t)}}{2}\right)} =: T(u^{(t)}).
\end{align*}
Ultimately, we have the map $u \to T(u) := u - \epsilon \tanh{\left(\frac{u}{2}\right)}$.
Analogous to the example provided in Condition 5 of the proof of Theorem~\ref{thm:e1l0e2g0}, we now consider $a > 0$ to be any constant, and set
$$\epsilon = 2a \coth{\left(\frac{a}{2}\right)}.$$
As previously argued, this choice satisfies $\eta > \frac{8}{|\epsilon_1| + |\epsilon_2|}$.
Therefore, we have
\begin{align*}
    T(a) &= a - \epsilon \tanh\left(\frac{a}{2}\right) \\
    &= a - 2a \coth{\left(\frac{a}{2}\right)}  \tanh\left(\frac{a}{2}\right) = -a.
\end{align*}
Similarly, noting that $T(u)$ is again an even function, we have $T(-a) = a$.
Therefore, the algorithm will oscillate if initialized at $a$ or $-a$.
Note that such a game can be defined for any positive constant $a$.
 
\section{Proof of Theorem \ref{thm:forcasewithonezero}}
\label{appendix:thm:forcasewithonezero}

\begin{figure}
  \centering

  \begin{subfigure}{0.45\textwidth}
    \includegraphics[width=\linewidth]{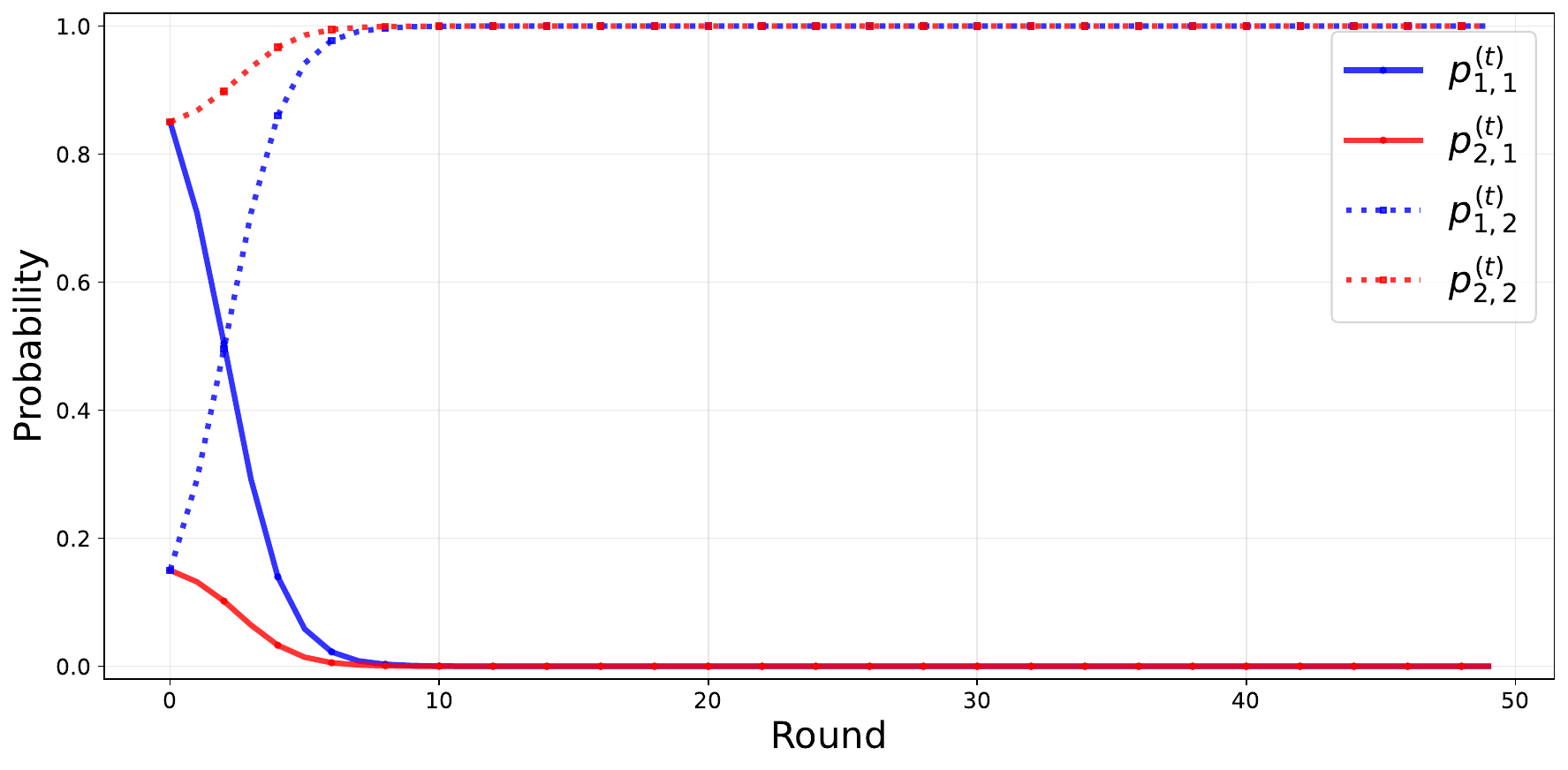}
    \caption*{r8. $\epsilon_2<0$}
  \end{subfigure}
  \begin{subfigure}{0.45\textwidth}
    \includegraphics[width=\linewidth]{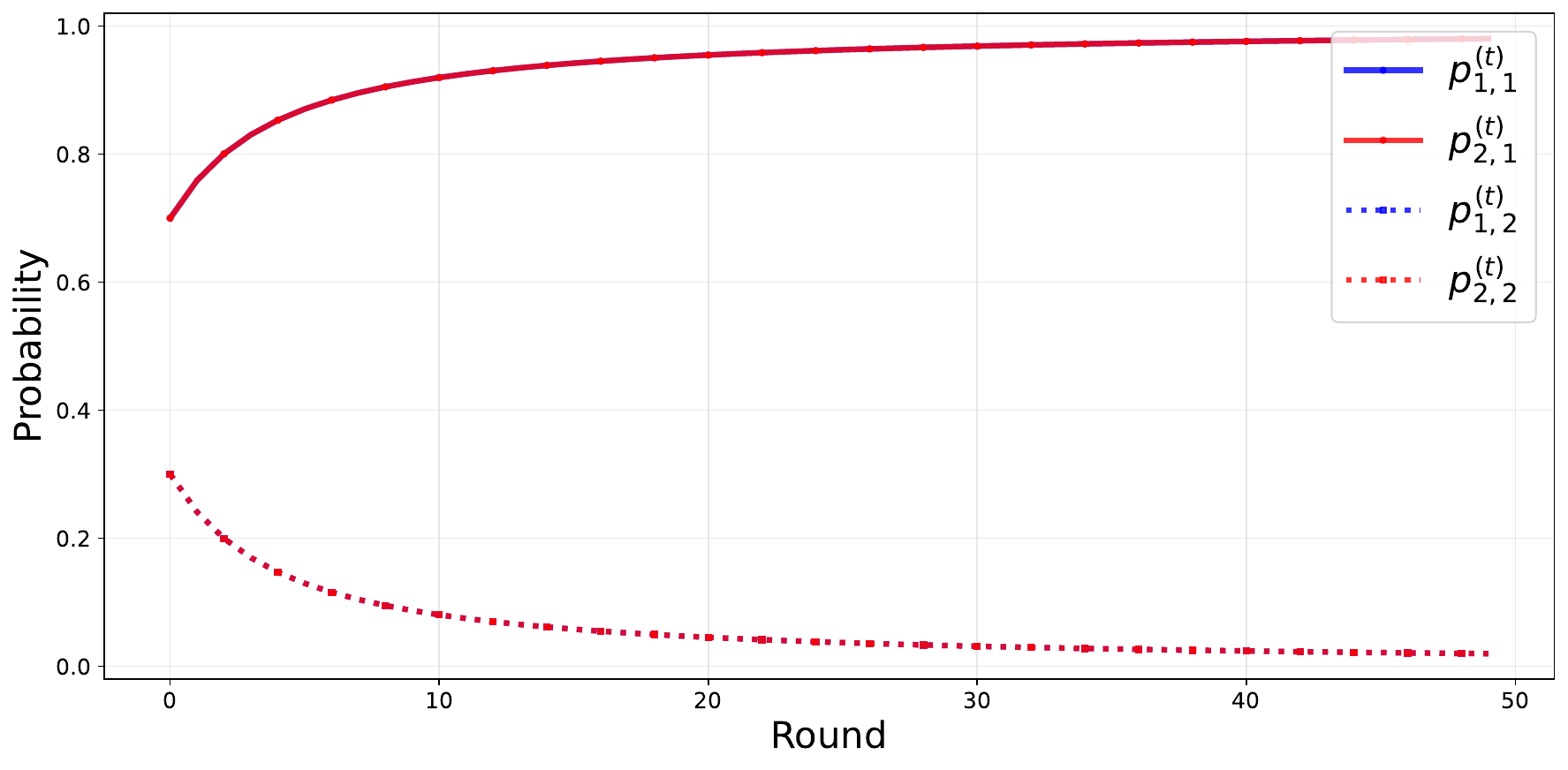}
    \caption*{r9. Identical initialization}
  \end{subfigure}
  
  \begin{subfigure}{0.45\textwidth}
    \includegraphics[width=\linewidth]{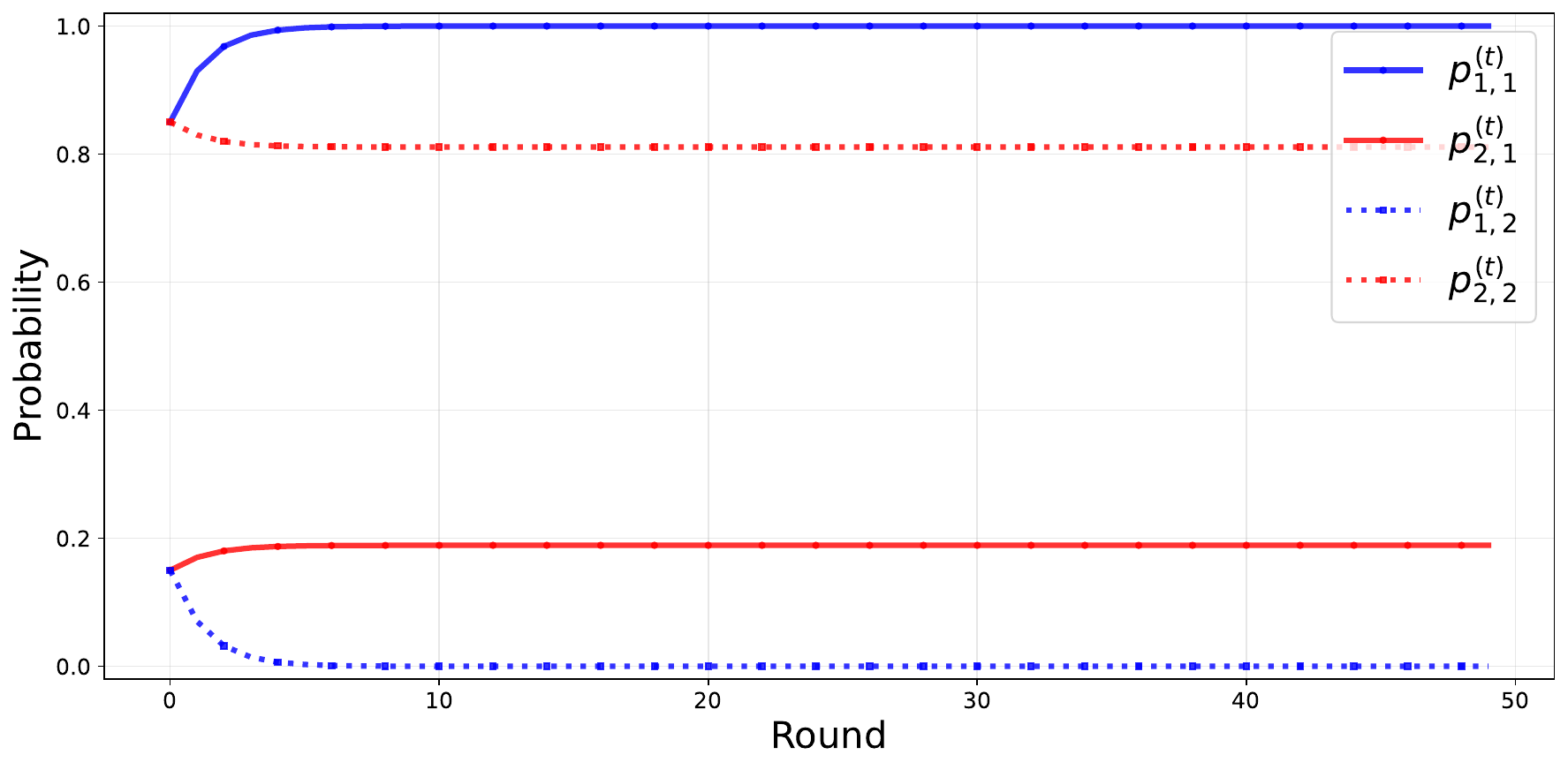}
    \caption*{r9. Opposite-Sign initialization}
  \end{subfigure}
  \begin{subfigure}{0.45\textwidth}
    \includegraphics[width=\linewidth]{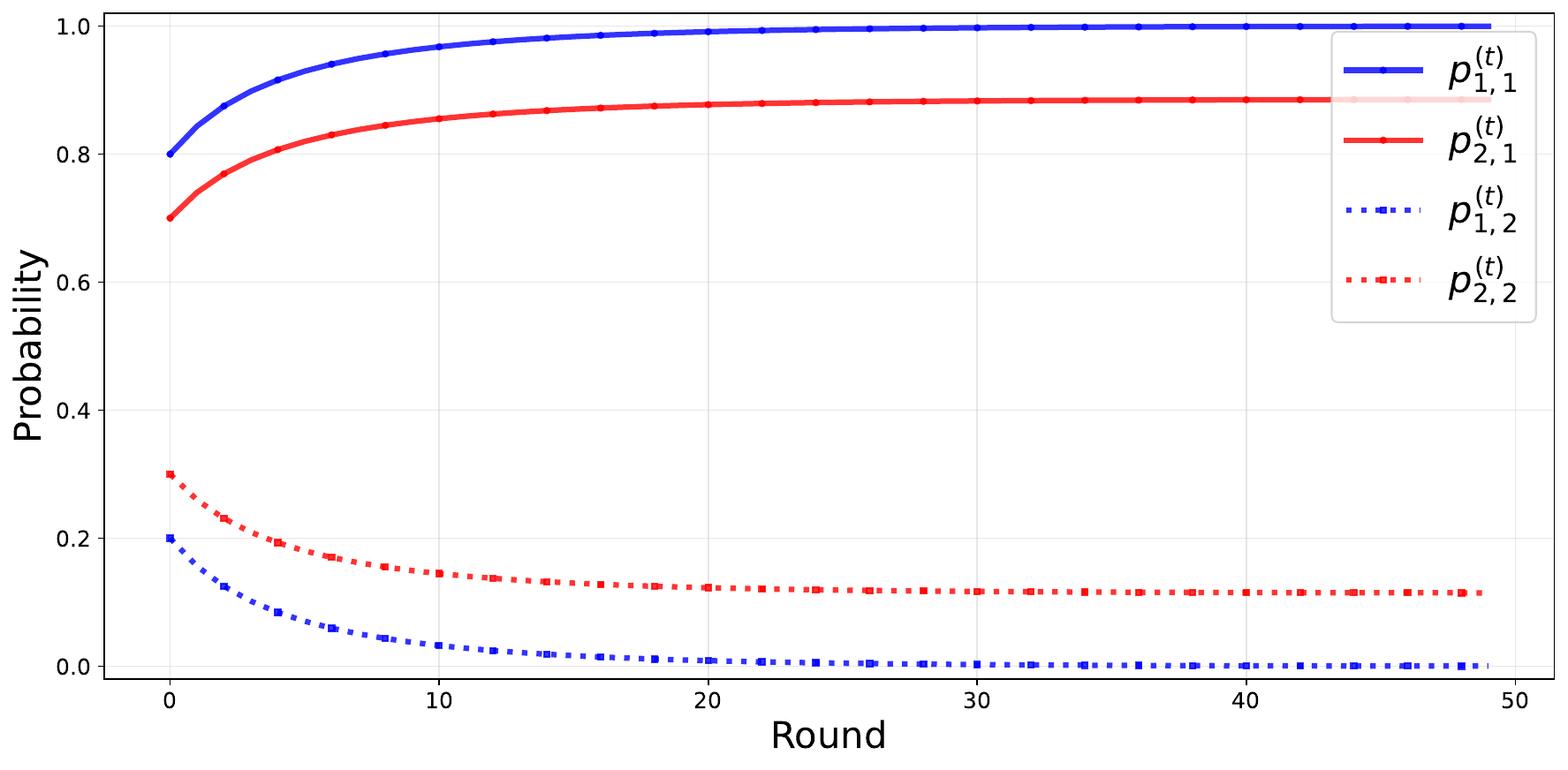}
    \caption*{r9. Same-sign initialization}
  \end{subfigure}

  \caption{Simulation results for Theorem \ref{thm:forcasewithonezero}. The row in Table~\ref{Table:NEconvergence} to which each case corresponds is marked in the sub-captions.}
  \vspace{-0.4cm}
  \label{figure1111111111-2}
\end{figure}

In this section, we prove Theorem~\ref{thm:forcasewithonezero}. In this case, we have (without loss of generality) that $\epsilon_1 = 0$ and $\epsilon_2 \neq 0$.
The corresponding simulation results are given in Figure~\ref{figure1111111111} (we did not include them in the main paper due to space limitations).
We study the convergence of the algorithm under different conditions on the game, listed below.

\paragraph{Condition 1: $\epsilon_1=0$, $\epsilon_2<0$.}

Recall from Equation~\eqref{eq:rt-update} that we have $r^{(t+1)}_i=r^{(t)}_i\exp\left(\eta f(r_j^{(t)})\right)$ for any player pair $i \neq j \in \{1,2\}$, where $f(r)=\frac{\epsilon_2}{1+r}< 0$, and $f'(r)=-\frac{\epsilon_2}{(1+r)^2}>0$. Since $f(r)< 0$, we can argue through induction that the ratio $\{r^{(t)}_i\}_{t=1}^{\infty}$ is
non-increasing for both players $i \in \{1,2\}$. This means that $r^{(t)}_i\leq r^{(1)}_i$, and, because $f(\cdot)$ is increasing, we have $f(r^{(t)}_i)\leq f(r^{(1)}_i)$. As a consequence, we have for any player pair $i \neq j \in \{1,2\}$: 
$$r^{(t+1)}_i=r_i^{(1)}\exp\left(\eta \sum_{k=1}^tf(r_j^{(k)})\right)\leq r_i^{(1)}\exp(\eta t\Delta_{j}^{(1)}).$$
This implies that $p^{(t+1)}_{i,1}\leq p^{(t+1)}_{i,2}r^{(1)}_i\exp(\eta t\Delta_{j}^{(1)})\leq r^{(1)}_i\exp(\eta t\Delta_{j}^{(1)})=r^{(1)}_i\exp(\eta tp^{(1)}_j\epsilon_2),$ and $\epsilon_2<0$. Therefore, $p_{i,1}^{(t)}$ converges to $0$ at an exponentially fast rate under this condition.

\paragraph{Condition 2: $\epsilon_1=0$, $\epsilon_2>0$, $\Delta_1^{(1)}=\Delta_2^{(1)}$.}  In this case, we have $f(r)>0$.
A similar inductive argument as in Condition 1 then tells us that $\{r_i^{(t)}\}_{t=1}^{\infty}$ is strictly increasing. Moreover, we have 
$$ \ln r^{(t+1)}_i=\ln r^{(t)}_i+\eta \frac{\epsilon_2}{1+r_i^{(t)}}. $$
Suppose that$\lim_{t\rightarrow\infty}r_i^{(t)}=\ell_i<\infty$ for some $\ell_i>0$. Taking the limit as $t \to \infty$ on both sides, we have
$$ \ln\ell_i=\ln\ell_i +\eta \frac{\epsilon_2}{1+\ell_i}.$$
The above equation is solved at $\ell_i=\infty,$ which contradicts with the assumption that $\ell_i<\infty$. Thus, $\lim_{t\rightarrow \infty}r_i^{(t)}=\infty$ for $i\in\{1,2\}$.  

\paragraph{Condition 3: $\epsilon_1=0$, $\epsilon_2>0$, $\Delta^{(1)}_1\not=\Delta^{(1)}_2$.} 
For this case, we first show that at least one player will converge to $\theta_1$, which means that the algorithm will at least converge to a mixed NE. Specifically, we have 
$$\ln r_i^{(t+1)}=\ln r_i^{(t)} + \eta\epsilon_2\frac{1}{1+r_j^{(t)}}.$$
Since the ratio is always non-negative, i.e. $r\geq 0$, we know both $r_i^{(t)}$ and $r_j^{(t)}$ are monotonically increasing. Based on the monotone convergence theorem,  $r_i^{(t)}$ (and $r_j^{(t)}$) either converge to some $0<\ell_i<\infty$ ($0<\ell_j<\infty$), or go to infinity. Suppose both $r^{(t)}_i$ and   $r^{(t)}_j$ converge to some constants. Let $\lim_{t\rightarrow\infty}r_i^{(t)}=\ell_i$ and $\lim_{t\rightarrow\infty}r_j^{(t)}=\ell_j$. Then, we have 
$$\ell_i=\lim_{t\rightarrow\infty}\ln r_i^{(t+1)}=\lim_{t\rightarrow\infty}\ln r_i^{(t)} + \lim_{t\rightarrow\infty}\eta\epsilon_2\frac{1}{1+r_j^{(t)}}=\ell_i+\frac{\eta\epsilon_2}{1+\ell_{j}},$$
The above equation can only be solved at $\ell_j=\infty$, which leads to a contradiction. Therefore, at least one of $r^{(t)}_i$ and $r^{(t)}_j$ will go to infinity. 

\paragraph{Example for converging to strictly mixed NE} Here, we construct an instance where $r^{(t)}_j$ converges to a constant, while $r^{(t)}_i$ goes to infinity. 
\begin{theorem}
For any player pair $i,j\in\{1,2\}$, $i\not=j$, let $r_j^{(1)}$ be a constant, and let $A$ be any constant such that $A>r_j^{(1)}$. Let 
$$r^{(1)}_i\geq \frac{\eta \epsilon_2}{\left(1-\exp\left(-\frac{\eta\epsilon_2}{1+A}\right)\right)\ln\frac{A}{r_j^{(1)}}}>0.$$
Then, we have $\lim_{t\rightarrow \infty}r_j^{(t)}=A'\in(0,A]$, and $\lim_{t\rightarrow \infty}r^{(t)}_i=\infty$. This implies that $\lim_{t\rightarrow \infty}p^{(t)}_{q,1}=\frac{A'}{1+A'}$, while $\lim_{t\rightarrow \infty}p^{(t)}_{p,1}=\infty.$
\end{theorem}
\begin{proof}
We prove the theorem by induction. Note that $r^{(1)}_j\leq A$ by assumption. Assuming that $r^{(t)}_j\leq A$,we have 
$$ r_i^{(t+1)} = r_i^{(t)}\exp\left(\eta\epsilon_2 \frac{1}{1+r_j^{(t)}} \right)\geq r_i^{(1)}\exp\left(\eta\epsilon_2 \frac{t}{1+A} \right)$$
for all $t \geq 1$.
Therefore, we have
\begin{equation}
    \begin{split}
\ln r_j^{(t+1)}= &\ln r_j^{(1)} + \eta\epsilon_2\sum_{k=1}^t \frac{1}{1+r_i^{(k)}}\leq  \ln r_j^{(1)} + \eta\epsilon_2\sum_{k=1}^t \frac{1}{r_i^{(k)}}\\
\leq &   \ln r_j^{(1)} + \frac{\eta\epsilon_2}{\left(1-\exp\left(-\frac{\eta\epsilon_2}{1+A}\right)\right)}\cdot\frac{1}{r_i^{(1)}}\leq \ln A,
    \end{split}
\end{equation}
and thus $r_{j}^{(t+1)}\leq A$. To summarize, by induction, we showed that $r_j^{(t)}\leq A$ for all $t$. Since the ratio $\{r_j^{(t)}\}_{t \geq 1}$ is also increasing in $t$, we have $\lim_{t\rightarrow\infty}r_j^{(t)}=A'$ for some $A'\in(0,A]$ by the monotone convergence theorem. Moreover, because $r_i^{(t)}\geq r_i^{(1)}\exp\left(\frac{\eta \epsilon_2 t}{1+A}\right),$
we have $\lim_{t\rightarrow \infty}r_i^{(t)}=\infty.$
This completes the proof of Theorem~\ref{thm:forcasewithonezero}.
\end{proof}

\section{Experimental Result For The Bank Game}
\label{sec: exp-results}
 We now illustrate the convergence of our framework through experiments based on the bank setup introduced in Section \ref{sectionMPPP}.

\subsection{The Bank Game}
\label{sbusectionthebankgame}
\begin{table}[!h]
\centering
\begin{tabular}{@{}ccccc@{}}
\toprule
\textbf{} &
  \textbf{$(\tau_{\ell},\gamma_{\ell})$} &
  \textbf{$(\tau_{\ell},\gamma_{h})$} &
  $(\tau_{h},\gamma_{\ell})$ &
  $(\tau_{h},\gamma_{h})$ \\ \midrule
$(\tau_{\ell},\gamma_{\ell})$ &
  $\frac{1}{2}h(\gamma_{\ell},\tau_{\ell},1)$ &
  0 &
  $\frac{1}{2}h(\gamma_{\ell},\tau_{h},1)$ &
  0 \\ \midrule
$(\tau_{\ell},\gamma_{h})$ &
  $h(\gamma_{\ell},\tau_{\ell},1)$ &
  $\frac{1}{2}h(\gamma_{h},\tau_{\ell},1)$ &
  $h(\gamma_{\ell},\tau_{h},1)$ &
  $\frac{1}{2}h(\gamma_{h},\tau_h,1)$ \\ \midrule
$(\tau_{h},\gamma_{\ell})$ &
  $h(\gamma_{\ell},\tau_{\ell},\tau_h)+\frac{1}{2}h(\gamma_{\ell},\tau_h,1)$ &
  $h(\gamma_{h},\tau_{\ell},\tau_h)$ &
  $\frac{1}{2}h(\gamma_{\ell},\tau_h,1)$ &
  0 \\ \midrule
$(\tau_{h},\gamma_{h})$ &
  $h(\gamma_{\ell},\tau_{\ell},1)$ &
  $h(\gamma_{h},\tau_{\ell},\tau_h)+\frac{1}{2}h(\gamma_{h},\tau_h,1)$ &
  $h(\gamma_{\ell},\tau_{h},1)$ &
  $\frac{1}{2}h(\gamma_{h},\tau_h,1)$ \\ \bottomrule
\end{tabular}
\caption{Utility matrix for bank 1. The rows are the strategies for bank 2, while the columns are the strategies for bank 1.}
\label{tab:my-table2}
\end{table}
In this section, we provide more details about the bank game defined in Section \ref{sectionMPPP}. More specifically, for each player, there are four actions: $(\tau_{\ell},\gamma_{\ell}), (\tau_{\ell},\gamma_{h}), (\tau_h,\gamma_{\ell})$
 and  $(\tau_{h},\gamma_h)$. The utilities of both banks under different pair of decisions are governed by \eqref{eqn:utility-onebank}. We provide a full characterization of the utility of Bank 1  in Table \ref{tab:my-table2}. 
 
\paragraph{Thresholds selection} For fixed interest rates $0<\gamma_{\ell}<\gamma_{h}<1$, we always set the corresponding thresholds as
$$\tau_{\ell}=\frac{1}{1+\gamma_h},\ \ \text{and}\ \  \tau_{h}=\frac{1}{1+\gamma_{\ell}}.$$ These choices for the thresholds are natural from Equation~\eqref{eqn:utility-onebank}. Indeed, it is a dominant strategy for a \emph{rational} bank with interest rate $\gamma$ to set a threshold of $\tau^*(\gamma) := \frac{1}{ 2 + \gamma}$. Setting a higher threshold $\tau > \tau^*(\gamma)$ leads to ignoring customers in the range $y \in [\tau^*(\gamma), \tau]$ that provide a utility of $(2 + \gamma) y - 1 > 0$, while setting a lower threshold $\tau < \tau^*(\gamma)$ leads to potentially incorporating some customers in the range $y < \tau^*(\gamma)$ that would yield a utility of $(2 + \gamma) y - 1 < 0$ as a result of low probability of repayment; both cases cannot improve total utility.

Next, we show that the two decisions $(\tau_{\ell},\gamma_{\ell})$ and $(\tau_{h},\gamma_{h})$ are always dominated decisions, so the weights assigned to these actions by the exponential weights algorithm will decrease exponentially fast. We firstly introduce several easy to verify properties of $h$. 
\begin{lemma}
\label{lem:property of f}
Assume $\emph{Supp}(D_y)=[0,1]$. We have the following:
\begin{enumerate}
\item Let $\tau_a<\tau_b<\tau_c$. Then $h_{D_y}(\gamma_1,\tau_a,\tau_c)=h_{D_y}(\gamma_1,\tau_a,\tau_b)+h_{D_y}(\gamma_1,\tau_b,\tau_c)$.
    \item $h(\gamma,\tau_a,\tau_b)$ is monotonically increasing with respect to $\gamma$. 
     \item We have $h(\gamma,\tau_a,\tau_b)\in[-1,2]$. 
    \item Let $y_{\gamma}=\frac{1}{2+\gamma}$. If $y_{\gamma}\in(\tau_a,\tau_b)$, then $h(\gamma,\tau_a,y_{\gamma})<0$, and  $h(\gamma,y_{\gamma},\tau_b)>0$.    
\end{enumerate} 

\end{lemma}
Next, based on these properties, it can be easily observe that $u_1((\tau_{h},\gamma_{\ell}),\theta)>u_1((\tau_{\ell},\gamma_{\ell}),\theta)$ for all $\theta\in\{(\tau_{\ell},\gamma_{\ell}), (\tau_{\ell},\gamma_{h}), (\tau_h,\gamma_{\ell}),(\tau_{h},\gamma_h)\}$. Excluding $(\tau_{\ell},\gamma_{\ell})$, we have $u_1((\tau_{\ell},\gamma_{h}),\theta)>u_1((\tau_{h},\gamma_{h}),\theta)$ for $\theta\in\{(\tau_{\ell},\gamma_{h}), (\tau_h,\gamma_{\ell}),(\tau_{h},\gamma_h)\}$.Therefore, $(\tau_{\ell},\gamma_{\ell})$ and $(\tau_{h},\gamma_{h})$ are dominated decisions.

\paragraph{Credit score distribution} 
We model customer credit scores using two types of distributions:  

\begin{itemize}
    \item Truncated Gaussian (\(\mu, \sigma\)): This is a standard normal distribution conditioned to lie within the normalized credit score range \( y \in [0,1] \), and models customers with credit scores symmetrically distributed around a mean value $\mu \in [0,1]$.
    \item Piecewise Uniform Distribution: This distribution models \(\beta_1\) fraction of customers with credit scores uniformly distributed in \([0, \tau_l]\), a fraction \(\beta_2\) with scores in \((\tau_l, \tau_h)\), and the remainder in \([\tau_h, 1]\).
  \begin{equation}
    \label{eq:pu-dist}
    y \sim  
    \begin{cases}  
        \frac{\beta_1}{\tau_l}, & \quad y \in [0, \tau_l] \\[4pt]  
        \frac{\beta_2}{\tau_h - \tau_l} & \quad y \in (\tau_l, \tau_h) \\[4pt]  
        \frac{1 - (\beta_1 + \beta_2)}{1 - \tau_h}, & \quad y \in [\tau_h, 1]
    \end{cases}
\end{equation}  
\end{itemize}

\paragraph{Interest Rates  and Credit Thresholds} 
We ran experiments across a wide range of interest rates, where $0 < \gamma_l < \gamma_h < 1$, and the corresponding thresholds 
To simulate different conditions for $\epsilon_1$ and $\epsilon_2$, we consider the following parameter configurations (Here, $(+,+)$ means $\epsilon_1>0$, $\epsilon_2>0$): 
\begin{itemize}
    \item (+, +): $\gamma_l = 0.4$, $\gamma_h = 0.8$, with a truncated Gaussian distribution $(\mu, \sigma) = (0.3, 0.1)$.
    \item (-, -): $\gamma_l = 0.4$, $\gamma_h = 0.8$, with a truncated Gaussian distribution $(\mu, \sigma) = (0.1, 0.3)$.
    \item (+, -): $\gamma_l = 0.4$, $\gamma_h = 0.8$, with a truncated Gaussian distribution $(\mu, \sigma) = (0.1, 0.2)$.
    \item (-, +): $\gamma_l = 0.6$, $\gamma_h = 0.7$, with piecewise uniform distribution $\beta_1 = 0.01, \beta_2 = 0.95$. 
\end{itemize}

\paragraph{Algorithm Configuration} We set the step size $\eta$ in Algorithm \ref{alg:Hedge} to be $0.1$. To make the setting more general, we consider applying the exponential weights  algorithm to the more general setting where we have four decisions: $(\tau_{\ell},\gamma_{\ell}), (\tau_{\ell},\gamma_{h}), (\tau_h,\gamma_{\ell})$
 and  $(\tau_{h},\gamma_h)$.
 As discussed in Section \ref{sectionMPPP},  the first and last action are dominated and the weights will converge to 0 fast.

The experimental results are give in Figures \ref{fig:dyna-mm}--\ref{fig:dyna-pm-2}. It can be seen that the convergence results matches our theorem. We start with the `-\,-' and `+\,+' cases in Figures~\ref{fig:dyna-mm} and \ref{fig:dyna-pp} (corresponding to Theorem \ref{thm:same-sign}), and see convergence to the unique pure symmetric NE. 

\begin{figure}[H]
    \centering
    \begin{subfigure}{0.49\linewidth}
        \centering
        \includegraphics[width=\linewidth]{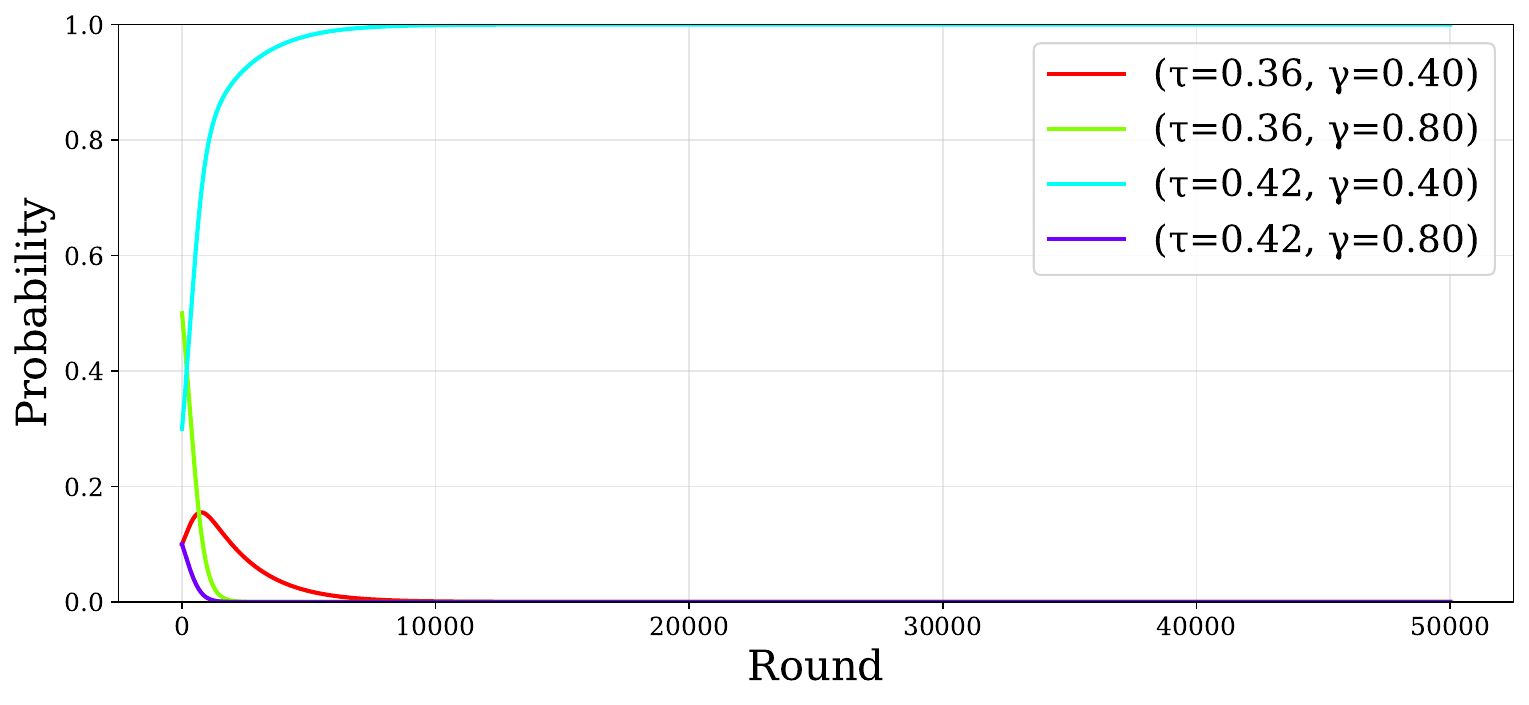}
    \end{subfigure}
    \begin{subfigure}{0.49\linewidth}
        \centering
        \includegraphics[width=\linewidth]{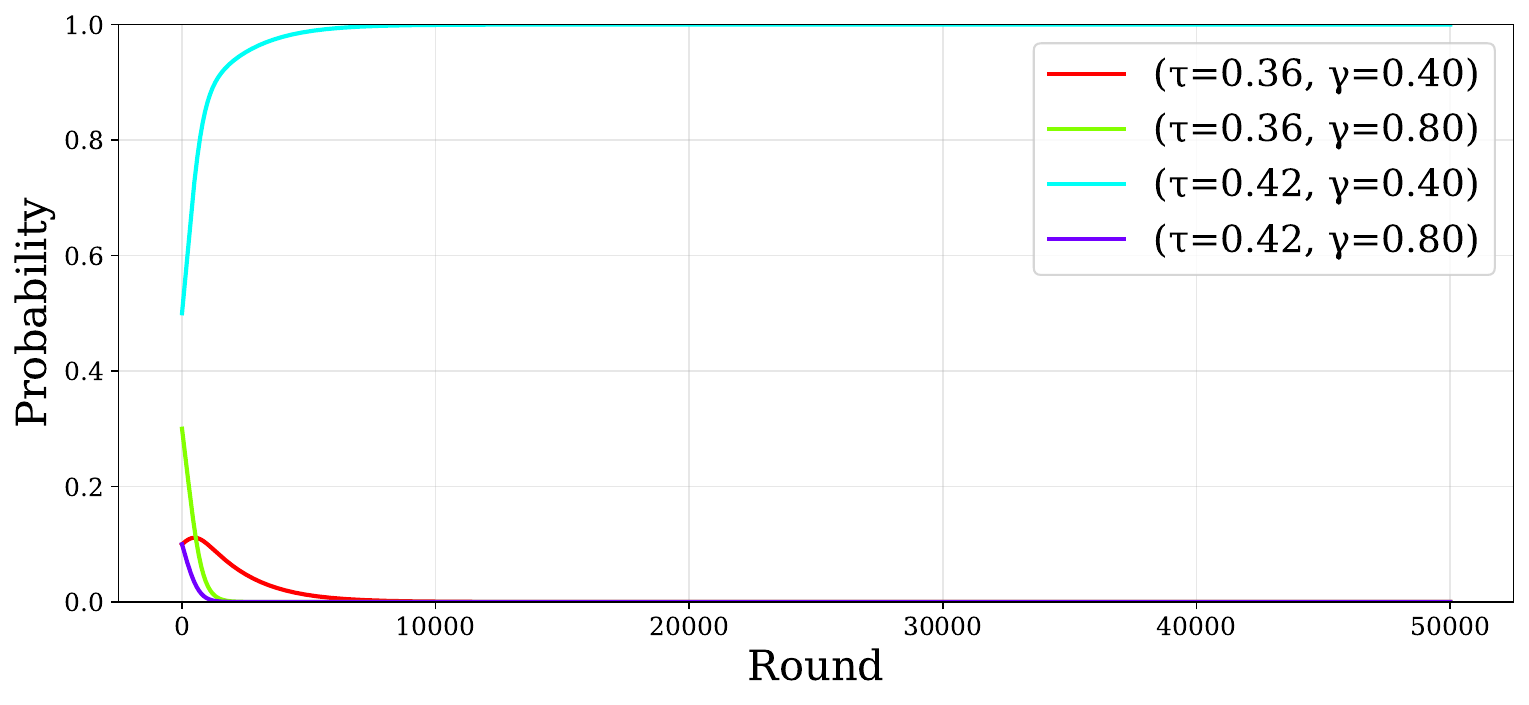}
    \end{subfigure}
    \caption{\textbf{Case - -} Exponential weights dynamics for Bank1 (left) and Bank2 (right), converging to $((\tau_h, \gamma_l),(\tau_h, \gamma_l))$.  $y \sim$ truncated Gaussian ($\mu=0.1, \sigma=0.3$), $\gamma_l = 0.4$, $\gamma_h = 0.8$.  Initial weights: Bank1 $(0.1, 0.5, 0.3, 0.1)$, Bank2 $(0.1, 0.3, 0.5, 0.1)$. \label{fig:dyna-mm}}
\end{figure}

\begin{figure}[H]
    \centering
    \begin{subfigure}{0.49\linewidth}
        \centering
        \includegraphics[width=\linewidth]{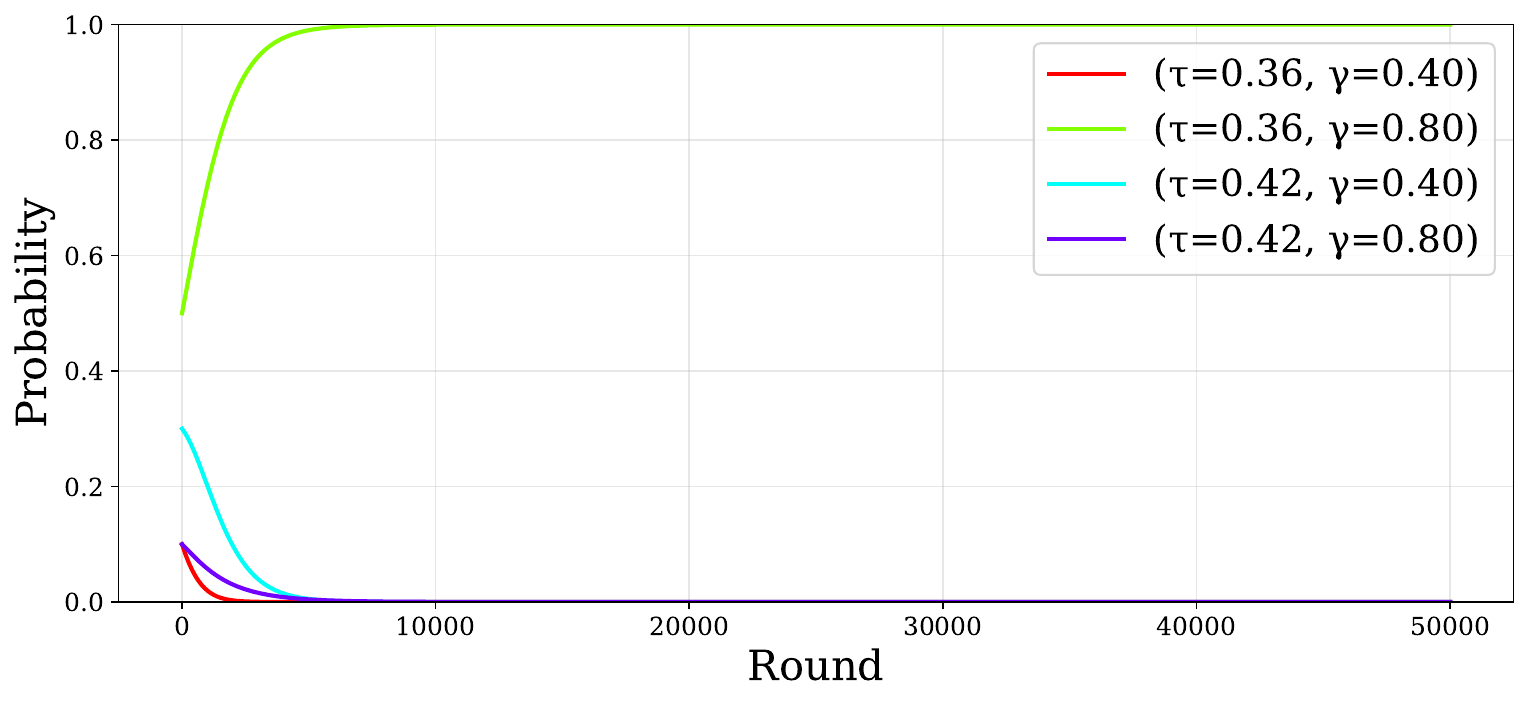}
    \end{subfigure}
    \begin{subfigure}{0.49\linewidth}
        \centering
        \includegraphics[width=\linewidth]{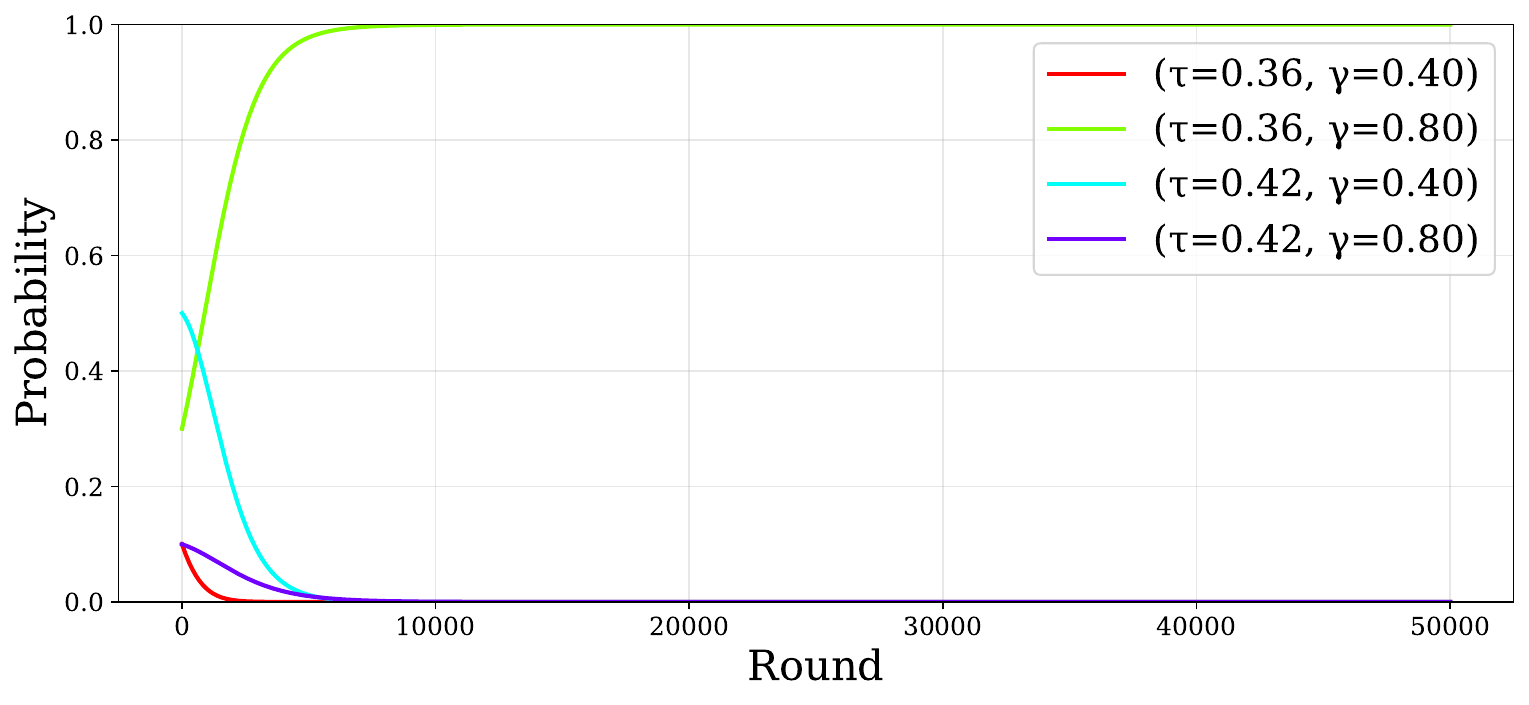}
    \end{subfigure}

    \caption{\textbf{Case +\,+} Exponential weights dynamics for Bank1 (left) and Bank2 (right), converging to $((\tau_l, \gamma_h),(\tau_l, \gamma_h))$. $y \sim$ truncated Gaussian ($\mu=0.3, \sigma=0.1$), $\gamma_l = 0.4$, $\gamma_h = 0.8$. Initial weights: Bank1 $(0.1, 0.5, 0.3, 0.1)$, Bank2 $(0.1, 0.3, 0.5, 0.1)$. \label{fig:dyna-pp}}
\end{figure}

We now present two representative figures for the `\(-\,+\)' case (corresponding to Theorem \ref{thm:e1l0e2g0}). Figure~\ref{fig:dyna-mp-ass} illustrates convergence to the pure asymmetric NE $((\tau_l, \gamma_h), (\tau_h, \gamma_l))$.
\begin{figure}[H]
    \centering
    \begin{subfigure}{0.49\linewidth}
        \centering
        \includegraphics[width=\linewidth]{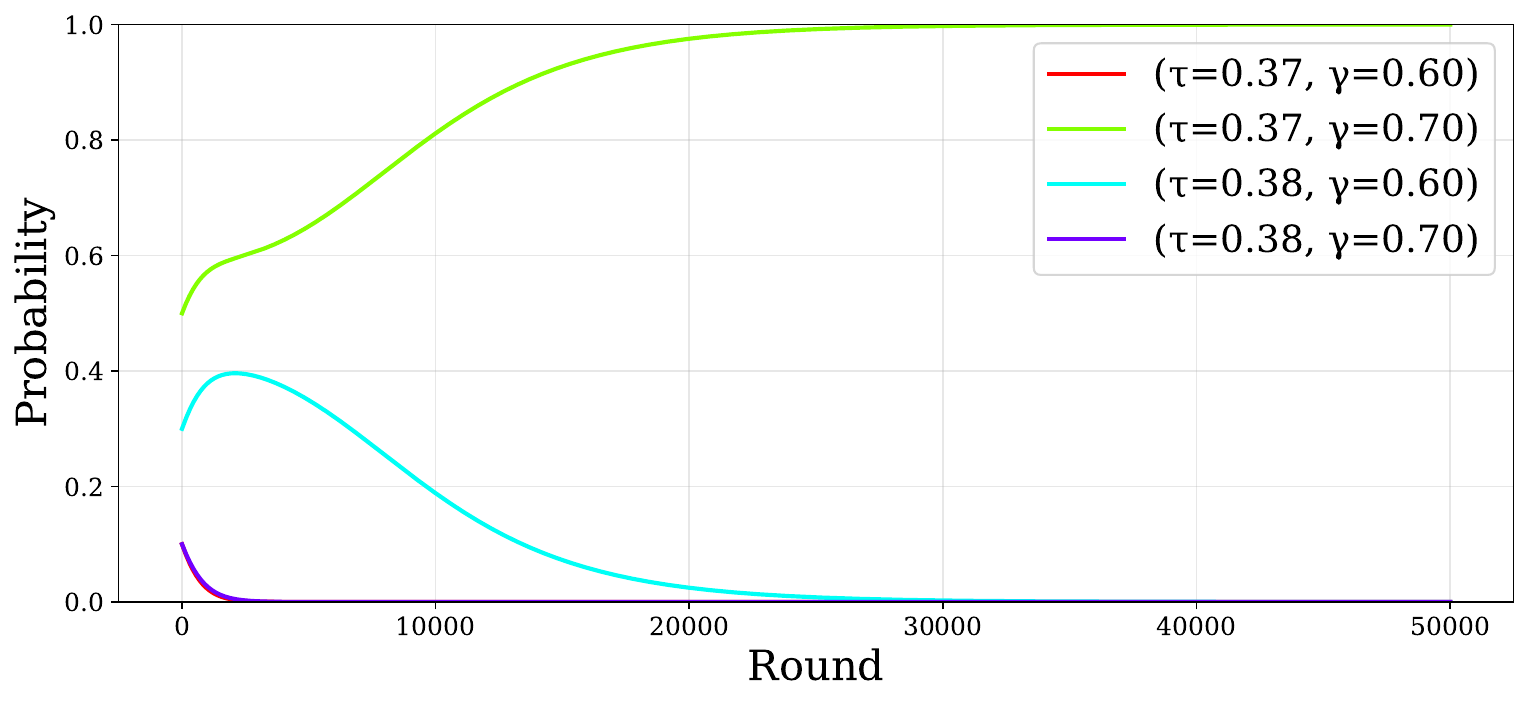}
    \end{subfigure}
    \begin{subfigure}{0.49\linewidth}
        \centering
        \includegraphics[width=\linewidth]{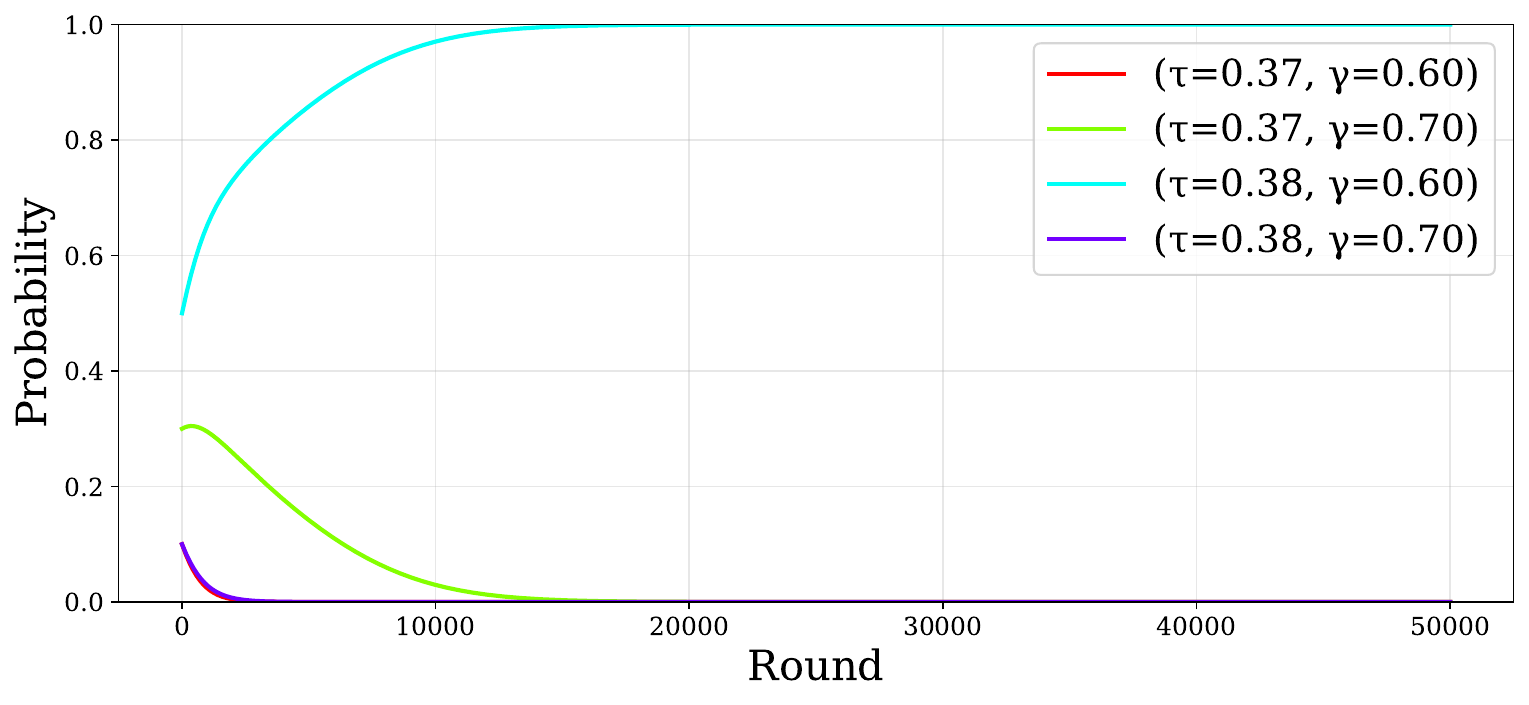}
    \end{subfigure}
    \caption{\textbf{Case -\,+ } Exponential weights dynamics for Bank1 (left) and Bank2 (right) converging to the pure NE $((\tau_l, \gamma_h), (\tau_h, \gamma_l))$. Parameters: $\epsilon_1<0, \epsilon_2>0$, $y \sim$ piecewise uniform distribution, $\gamma_l = 0.6$, $\gamma_h = 0.7$. Initial weights: Bank1 $(0.1, 0.5, 0.3, 0.1)$, Bank2 $(0.1, 0.3, 0.5, 0.1)$. \label{fig:dyna-mp-ass}}
\end{figure}

Finally, we present three figures for the `\(+\,-\)' case. Figure \ref{fig:dyna-pm-2} shows convergence to the pure symmetric NE $((\tau_h, \gamma_l), (\tau_h, \gamma_l))$.

\begin{figure}[H]
    \centering
    \begin{subfigure}{0.49\linewidth}
        \centering
        \includegraphics[width=\linewidth]{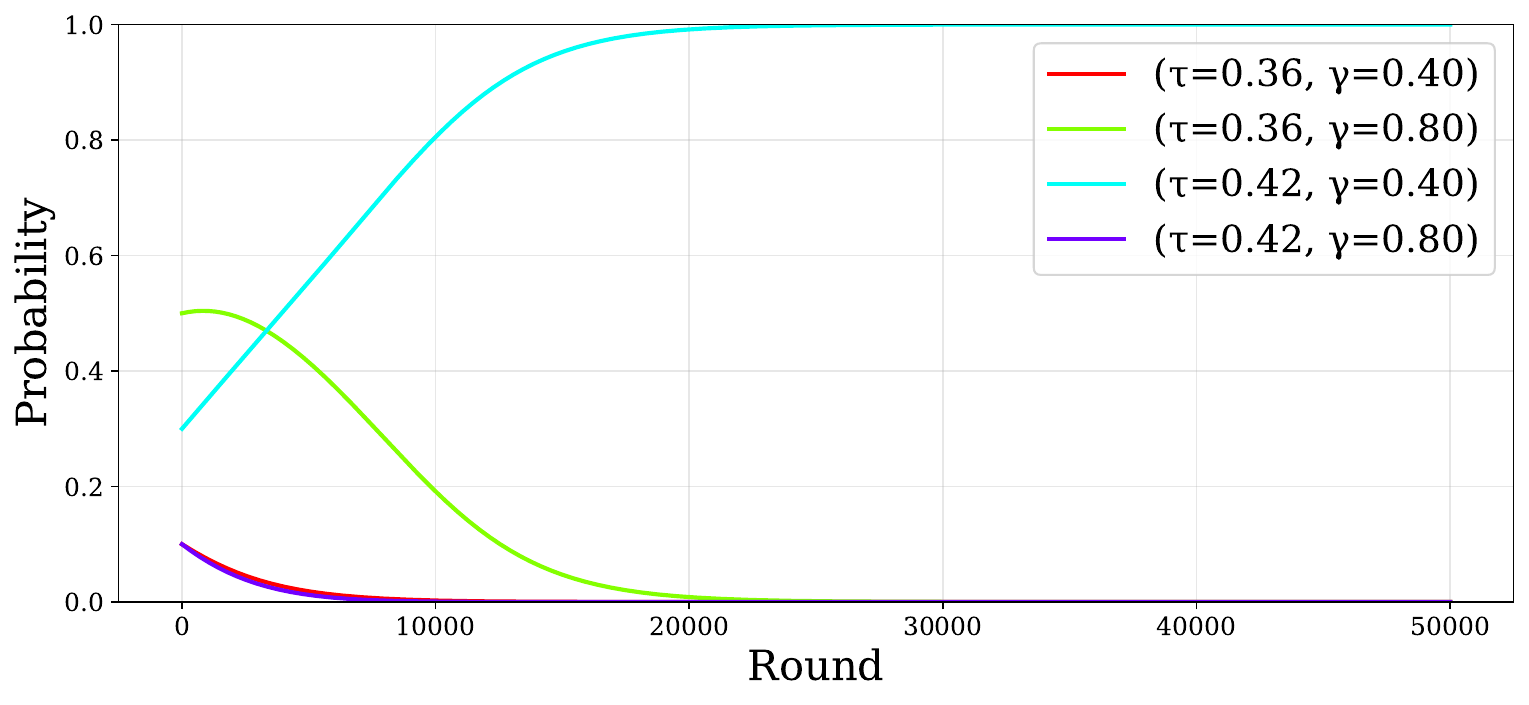}
    \end{subfigure}
    \begin{subfigure}{0.49\linewidth}
        \centering
        \includegraphics[width=\linewidth]{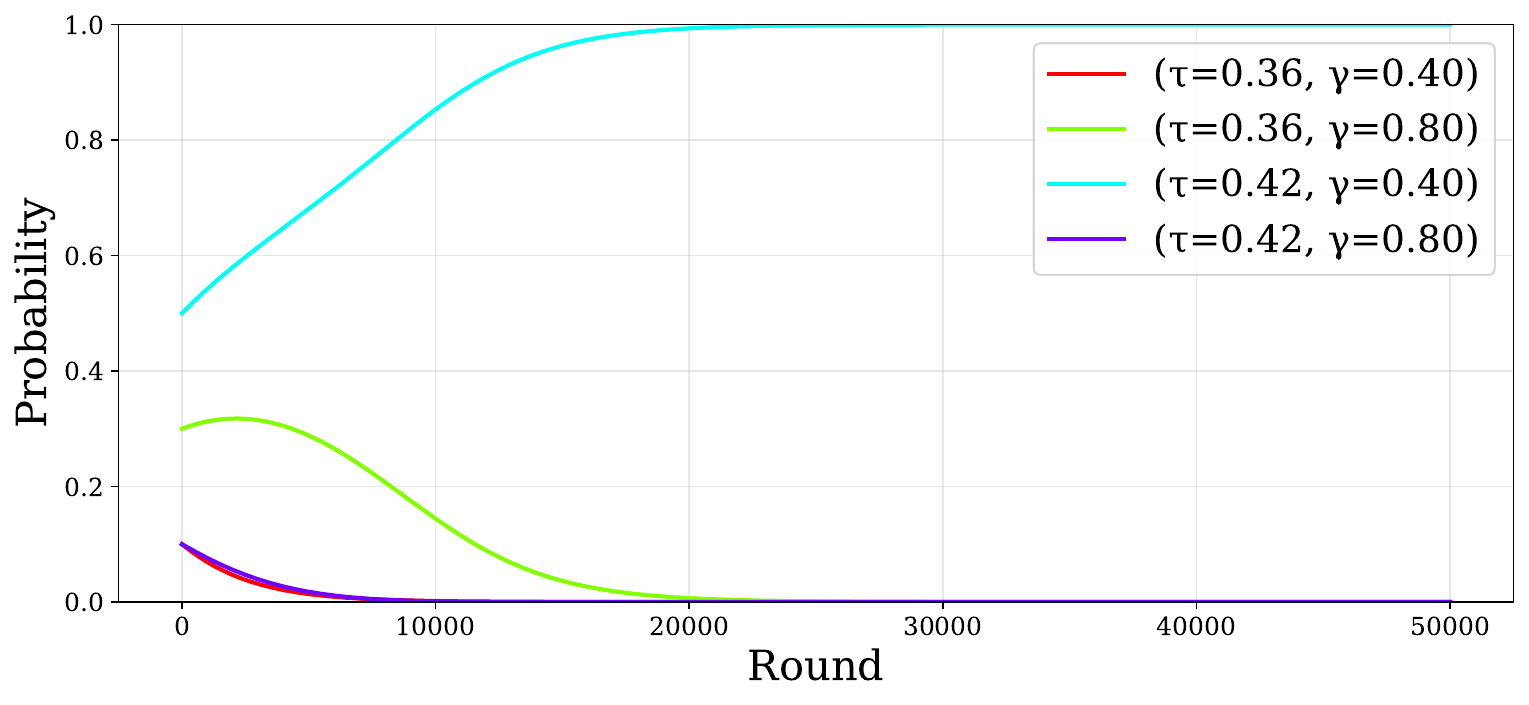}
    \end{subfigure}
    \caption{\textbf{Case +\,-}  Exponential weights dynamics for Bank1 (left) and Bank2 (right), converging to the pure NE $((\tau_h, \gamma_l),(\tau_h, \gamma_l))$. $y \sim$ truncated Gaussian ($\mu=0.1, \sigma=0.2$), $\gamma_l = 0.4$, $\gamma_h = 0.8$. Initial weights: Bank1 $(0.1, 0.5, 0.3, 0.1)$, Bank2 $(0.1, 0.3, 0.5, 0.1)$ \label{fig:dyna-pm-2}}.
\end{figure}

\end{document}